\documentclass[a4paper, 11pt]{article}

\usepackage[utf8]{inputenc}

\usepackage{
amssymb, amsmath, amsthm, amstext, amsfonts, amsopn, amscd, amsbsy, comment, enumerate
}
\usepackage{mathrsfs, dsfont,
mathdots}
\usepackage{bm}
\usepackage{accents}

\usepackage{graphicx}
\usepackage[usenames,dvipsnames]{xcolor}

\usepackage[colorlinks=true, linkcolor=black, breaklinks=true, citecolor=blue, urlcolor=blue]{hyperref} 
\usepackage{float}
\usepackage{subcaption}
\usepackage{caption}
\usepackage{tikz}

\usepackage[numbers]{natbib}
\usepackage{url}

\newlength{\dinwidth}
\newlength{\dinmargin}

\numberwithin{equation}{section}

\renewcommand{\emptyset}{\varnothing}

\newcommand{\fA}{\mathfrak{A}}\newcommand{\cA}{\mathcal{A}}
\newcommand{\fB}{\mathfrak{B}}\newcommand{\cB}{\mathcal{B}}

\newcommand{\cD}{\mathcal{D}}
\newcommand{\cE}{\mathcal{E}}
\newcommand{\fF}{\mathfrak{F}}\newcommand{\sF}{\mathscr{F}}
\newcommand{\fh}{\mathfrak{h}}
\newcommand{\cH}{\mathcal{H}}

\newcommand{\cL}{\mathcal{L}}

\newcommand{\cO}{\mathcal{O}}

\newcommand{\dR}{\mathds{R}}
\newcommand{\fS}{\mathfrak{S}}\newcommand{\cS}{\mathcal{S}}

\newcommand{\cU}{\mathcal{U}}
\newcommand{\cW}{\mathcal{W}}

\newcommand{\CC}{\mathbb{C}}
\newcommand{\ZZ}{\mathbb{Z}}
\newcommand{\TT}{\mathbb{T}}
\newcommand{\RR}{\mathbb{R}}
\newcommand{\NN}{\mathbb{N}}

\newcommand{\vep}{\varepsilon}

\newcommand{\ltwo}{\ell^{2}}

\DeclareMathOperator{\Ad}{Ad}
\DeclareMathOperator{\id}{id}
\DeclareMathOperator{\Tr}{Tr}
\DeclareMathOperator{\vol}{vol}

\DeclareMathOperator*{\sotlim}{SOT-lim}

\DeclareMathOperator{\supp}{supp}

\DeclareMathOperator{\CAR}{CAR}

\def\Undertilde#1{\mathord{\vtop{\ialign{##\crcr
$\hfil\displaystyle{#1}\hfil$\crcr\noalign{\kern1.5pt\nointerlineskip}
$\hfil\widetilde{}\hfil$\crcr\noalign{\kern1.5pt}}}}}

\def\S2{S^{1(2)}}

\def\sl2{{{\rm SL}(2,\RR)}}
\def\psl2{{{\rm PSL}(2,\RR)}}
\def\u1{{{\rm V}(1)}}
\def\su2{{{\rm SV}(2)}}

\def\so3{{{\rm SO}(3)}}

\newtheorem{theorem}{Theorem } [section]
\newtheorem{proposition}[theorem]{Proposition}
\newtheorem{lemma}[theorem]{Lemma}
\newtheorem{definition}[theorem]{Definition}
\newtheorem{corollary}[theorem]{Corollary}
\newtheorem{remark}[theorem]{Remark}
\newtheorem{example}[theorem]{Example}
\newtheorem{criterion}[theorem]{Criterion}
\newtheorem{conjecture}{Conjecture}
\newtheorem{assumption}{Assumption}

\newcommand{\bea}{\begin{assumption}}
	\newcommand{\eea}{\end{assumption}}
\newcommand{\beco}{\begin{conjecture} }
	\newcommand{\eeco}{\end{conjecture} }
\newcommand{\beq}{\begin{equation}}
	\newcommand{\eeq}{\end{equation}}
\newcommand{\beqa}{\begin{eqnarray}}
	\newcommand{\eeqa}{\end{eqnarray}}
\newcommand{\ben}{\begin{arabicenumerate}}
	\newcommand{\een}{\end{arabicenumerate}}
\newcommand{\bex}{\begin{example}}
	\newcommand{\eex}{\end{example}}
\newcommand{\ber}{\begin{remark}}
	\newcommand{\eer}{\end{remark}}
\newcommand{\bec}{\begin{corollary}}
	\newcommand{\eec}{\end{corollary}}
\newcommand{\bep}{\begin{proposition}}
	\newcommand{\eep}{\end{proposition}}
\newcommand{\becr}{\begin{criterion}}
	\newcommand{\eecr}{\end{criterion}}

\begin{document}

\tikzset{->-/.style={decoration={
  markings,
  mark=at position #1 with {\arrow{latex}}},postaction={decorate}}}

\title{
Scaling limits of lattice quantum fields by wavelets
}
\author{Vincenzo Morinelli${}^{1,}$\footnote{\tt \, morinelli@math.fau.de; ${}^\bigtriangleup$ morsella@mat.uniroma2.it; \newline ${}^\circ$ alexander.stottmeister@itp.uni-hannover.de; ${}^\bullet$ hoyt@mat.uniroma2.it}, Gerardo Morsella$^{{}2,\bigtriangleup}$, Alexander Stottmeister${}^{3, \circ}$, Yoh Tanimoto${}^{2, \bullet}$}
\date{
\small{${}^{1}$ Department of Mathematics, FAU Erlangen-N\"urnberg,\\ Cauerstra\ss e 11, 91058 Erlangen, Germany\smallskip\\
${}^{2}$ Dipartimento di Matematica, Universit\`a di Roma Tor Vergata\\
Via della Ricerca Scientifica 1, I-00133 Roma, Italy\smallskip \\
${}^3$ Institut f\"ur Theoretische Physik, Leibniz Universit\"at Hannover
\\Appelstra\ss e 2, 30167 Hannover, Germany} \\[0.5cm]
\today}

\maketitle

\begin{abstract}
We present a rigorous renormalization group scheme for lattice quantum field theories in terms of operator algebras.
The renormalization group is considered as an inductive system of scaling maps between lattice field algebras.
We construct scaling maps for scalar lattice fields using Daubechies' wavelets, and show that
the inductive limit of free lattice ground states exists and the limit state extends to the familiar massive continuum free field,
with the continuum action of spacetime translations.
In particular, lattice fields are identified with the continuum field smeared with Daubechies' scaling functions.
We compare our scaling maps with other renormalization schemes and their features, such as the momentum shell method or block-spin transformations.
\end{abstract}

\section{Introduction}
\label{sec:intro}

The Wilson-Kadanoff renormalization group \cite{KadanoffScalingLawsFor, Wilson-71-Renormalization1, Wilson-71-Renormalization2} is a cornerstone in the understanding of classical and quantum many-body system. It provides a conceptual framework that unifies the theory of critical phenomena and universality in (quantum) statistical mechanics with
quantum field theory via the existence of infrared fixed points under scale-changing operations.
Moreover, the construction of continuum models from lattice approximations via scaling limits has been developed to a wide extent in a mathematically rigorous form in constructive quantum field theory, especially in the classical probabilist framework of the Euclidean approach \cite{GlimmQuantumFieldTheory, GlimmQuantumPhysicsA, FernandezRandomWalksCritical}.
While in these works interacting models have been rigorously constructed, the field operators and observables of the lattice and continuum theories are related only indirectly in terms of correlation functions and an analytic continuation of the latter together with a Streater-Wightman reconstruction is necessary to access the operator-algebraic content of the quantum field theory. 
In the present paper, we start from a lattice approximation to quantum field theories closer in spirit to the earlier constructive works by Friedrichs, Glimm-Jaffe and others \cite{GlimmQuantumFieldTheory, FriedrichsPerturbationOfSpectra} and the setting of quantum statistical mechanics. We explicitly implement an operator-algebraic approach \cite{BrothierConstructionsOfConformal, BrothierAnOperatorAlgebraic} to the Wilson-Kadanoff renormalization group for scalar lattice quantum fields in any dimension
(see also \cite{StottmeisterOperatorAlgebraicRenormalization} for a short overview).
Specifically, using results of wavelet theory~\cite{DaubechiesTenLecturesOn, MeyerWaveletsAndOperators}, we build a real-space (as opposed to momentum-space) renormalization group by identifying the lattice fields with the continuum fields smeared with suitable functions which generate a wavelet basis.
This results in a sequence of homomorphisms between lattice field algebras at various scales. We supplement the lattice algebras with a sequence of states that approaches a critical point in an appropriate way (the massless free field for the massive free field), which assures the continuity of asymptotic symmetry groups avoiding certain obstacles encountered in \cite{JonesANoGo, KlieschContinuumLimitsOf, OsborneQuantumFieldsFor}, while keeping track of local algebras,
cf.\!  \cite{ZiniConformalFieldTheories}. Therefore, we succeed to construct a continuum field theory in the Haag-Kastler sense~\cite{HaagLocalQuantumPhysics} in the real-space framework of the Wilson-Kadanoff renormalization group.

Our implementation is inspired by the strong connection between renormalization of classical lattice systems and wavelet theory \cite{BattleWaveletsAndRenormalization}, which prevails in the quantum setting, cp.~\cite{BrennenMultiscaleQuantumSimulation, EvenblyEntanglementRenormalizationIn, MatsuedaAnalyticOptimizationOf, FriesRenormalizationOfLattice, WitteveenQuantumCircuitApproximations, WitteveenWaveletConstructionOf}.  As an obvious application, we show rigorously and in detail that the renormalization group flow of the ground states of free lattice fields in the vicinity of the unstable, massless fixed point allows us to reconstruct the massive continuum free field as a scaling limit. By using compactly supported wavelets (or finite low-pass filters), this construction directly yields local time-zero algebras of the continuum theory which allows us to pass to the infinite-volume limit in a simple way.
From a mathematical point of view, operator algebras naturally accommodate the quantum version of the Wilson-Kadanoff approach as the former are generally understood to incorporate non-commutative measure theory, see \cite{LangHamiltonianRenormalizationI, ThiemannCQGRenormalisation} for another approach reversing the classical Euclidean approach to obtain a quantum version. Furthermore, in the specific formulation we use, the renormalization group is realized as an inductive system of
$*$-morphisms in our examples, and the scaling limit is identified with a state on the inductive-limit algebra \cite{KadisonFundamentalsOfThe2}.
To some extent, this reflects a weak form of universality \cite[Chapter 3]{ZinnJustinPhaseTransitionsAnd} in mathematical sense: The scaling limit is independent of the specific details of the inductive system at finite scales, e.g.~modifications of the dispersion relations of the free lattice ground states are irrelevant as long as the asymptotic behavior is unchanged. As both, the infrared limits of the renormalization group and the inductive limits of operators algebras, are only concerned with asymptotic properties of their building blocks (with a certain notion of coherence), this suggests a link at a conceptual level. This should not be confused with the concept of universality classes which are associated with different fixed points of the renormalization group when taking into account additional interactions. In the latter sense, we are only concerned with the universality class of the massless free field in this work.
As can be seen from the free field example, criticality of the lattice models (i.e.~divergence of correlation lengths) is intimately linked to continuity properties of the operator-algebraic limits.

Our formulation of the renormalization group is also linked to other developments in quantum theoretical descriptions of low-dimensional many-body systems, e.g.~the density matrix renormalization group (DMRG) \cite{WhiteRealSpaceQuantum1, WhiteDensityMatrixFormulation, WhiteDensityMatrixAlgorithms, SchollwoeckTheDensityMatrix}, which is a real-space formulation that has turned out to be especially efficient in one dimension.
As the Wilson-Kadanoff approach fits naturally into an information geometrical setup \cite{BenyInformationGeometricApproach}, real-space renormalization schemes have also drawn increasing attention in the theory of quantum information, specifically concerning tensor networks \cite{CiracRenormalizationAndTensor} and the multi-scale entanglement renormalization ansatz (MERA) \cite{VidalAClassOf, EvenblyAlgorithmsForEntanglement, EvenblyEntanglementRenormalizationIn, PfeiferEntanglementRenormalizationScale}.

\subsection*{Summary of the formulation and results}
We formulate a renormalization group for scalar lattice field theories using a scaling function $\phi$ of wavelet theory. Such a function satisfies a self-similarity equation,
\begin{align*}
\phi(x) & = \sum_{n\in\ZZ^{d}}h_{n}\!\ 2^{\frac{d}{2}}\phi(2x-n),
\end{align*}
which encodes the decomposition of $\phi$ at coarse scale in terms of itself at a fine scale (in the following $d$ will always denote the spatial dimension). Clearly, this equation is analogous in spirit to Kadanoff's block-spin transformation, which considers a spin on a coarse lattice to be an average of spins on a fine lattice. But, while the block-spin transformation typically averages uniformly over adjacent sites, the scaling equation allows for a weighted average over sites. The coefficients $h_{n}$ are commonly known as a low-pass filter as they describe the coarse-scale features captured by $\phi$. Wavelets enter the picture as a parametrization of the orthogonal complement of $\phi$ after increasing the resolution leading to a complementary high-pass filter tracking the fine-scale details. Thus, it is tempting to utilize the scaling equation to define the relation among lattice quantum fields at successive scales $\vep_{N}$, $\vep_{N+1} = \vep_N/2$  (and similarly for momenta):
\begin{align*}
\Phi_{N}(x) & = 2^{-\frac{d}{2}}\sum_{n\in\ZZ^{d}}h_{n}\!\ \Phi_{N+1}(x+\vep_{N+1}n).
\end{align*}
In terms of the algebras of observables $\fA_{N}$ at scales $\vep_{N}$ this precisely defines an inductive system,
\begin{align*}
\alpha^{N}_{N'} : \fA_{N} & \longrightarrow \fA_{N'}, & N & < N,
\end{align*}
with the algebraic and analytic properties of the $\alpha$'s encoded into the coefficients $h_{n}$. Specifically, by using a Daubechies' scaling function which is compactly supported and orthogonal (integer translates of $\phi$ are orthogonal at a fixed scale), it is possible to realize the $\alpha$'s as unital $*$-morphisms that equip the limit algebra $\fA_{\infty} = \varinjlim_{N}\fA_{N}$ with a natural quasi-local structure because the low-pass filter is finite in this case and, therefore, limits the spatial support of a localized finite-scale quantum field $\Phi_{N}(x)$ in the limit. Indeed, the image of the latter in $\fA_{\infty}$ can be identified with the continuum quantum field $\Phi$ smeared with $\phi$ at scale $\vep_{N}$,
\begin{align*}
\Phi_{N}(x) & = \vep_{N}^{-d}\int\!\ dy\!\ \Phi(y) \phi(\vep_{N}^{-1}(y-x)).
\end{align*}
Now, the $\alpha$'s are what we will call the renormalization group or scaling maps, and the more common picture of the renormalization group transformations linking states of a physical system among different scales arises by their dual action on the states spaces $\fS_{N+1}\rightarrow\fS_{N}$ (generalized density matrices). Following Wilson \cite{WilsonTheRenormalizationGroupKondo}, we consider the renormalization group flow of initial states $\omega^{(N)}_{0}$ of the finite scale systems $\fA_{N}$,
\begin{align*}
\omega^{(N)}_{M} & = \omega^{(N+M)}_{0}\circ\alpha^{N}_{N+M},
\end{align*}
to construct a scaling limit $\omega^{(\infty)}_{\infty} = \varprojlim_{N}\lim_{M}\omega^{(N)}_{M}$ on $\fA_{\infty}$. To make these abstract considerations maximally explicit and to test the construction, we provide a rather complete treatment of the important case of harmonic (free) lattice fields on
the torus $\mathbb{T}^d_L=[-L,L)^d$,
\begin{align*}
H^{(N)}_{0} & = \tfrac{1}{2}\vep_{N}^{d}\bigg(\sum_{x\in\Lambda_{N}}\Pi_{N|x}^{2}+\mu^{2}_{N}\vep_{N}^{-2}\Phi_{N|x}^{2}-2\sum_{\langle x,y\rangle\subset\Lambda_{N}}\vep_{N}^{-2}\Phi_{N|x}\Phi_{N|y}\bigg),
\end{align*}
and their ground states $\omega^{(N)}_{L}$, which are expected to admit sensible scaling limits in the vicinity of the Gaussian fixed point $\mu^2_{N}=2d$. More precisely, we show that:
\begin{itemize}
	\item[(1)] the scaling limit $\omega^{(\infty)}_{L,\infty}$ exists in finite volume $(2L)^d$ and is given by the continuum free field of mass $m$ assuming the renormalization condition $\lim_{N\rightarrow\infty}\vep_{N}^{-2}(\mu_{N}^{2}-2d)= m^{2}$ holds,
	\item[(2)] the finite scale dynamics converges to the free dynamics in the Gelfand-Naimark-Segal representation of $\omega^{(\infty)}_{L, \infty}$ (a similar statement holds for spatial translations),
	\item[(3)] the infinite volume (or thermodynamical) limit $L\to \infty$ of $\omega^{(\infty)}_{L,\infty}$ exists and the local time-zero and spacetime algebras in finite volume are unitarily equivalent to those in infinite volume,
\end{itemize}
if the renormalization group is realized in terms of Daubechies' wavelets of sufficiently high regularity. The latter observation also fits nicely with the idea of universality as the details of a valid approximation should not matter, i.e.~the scaling limit is independent of the precise regularity of wavelets as long as it is sufficient. Thus, our operator-algebraic implementation of Wilson-Kadanoff renormalization enables a direct construction of the local net of observables of the continuum free field (in the sense of the Haag-Kastler axioms \cite{HaagLocalQuantumPhysics}) from free lattice fields. In addition, we observe that our approach yields a rigorous way to deduce spacetime locality 
in the continuum from Lieb-Robinson bounds \cite{LiebTheFiniteGroup, CramerLocalityOfDynamics, OsborneContinuumLimitsOf, NachtergaeleLRBoundsHarmonic, NachtergaeleQuasiLocalityBounds} albeit asymptotically optimal estimates of the Lieb-Robinson velocity are required to conclude that $c=1$ in natural units.
Moreover, we also illustrate how our wavelet renormalization group yields an analytic version of the multi-scale entanglement renormalization ansatz (MERA) for scalar fields in arbitrary dimensions, cp.~\cite{EvenblyEntanglementRenormalizationAnd, EvenblyRepresentationAndDesign, HaegemanRigorousFreeFermion, WitteveenWaveletConstructionOf}.
Finally, we also comment on the approach to the renormalization group via scaling algebras \cite{BuchholzScalingAlgebrasAnd1, BuchholzScalingAlgebrasAnd2, DAntoniScalingAlgebras, BostelmannScalingAlgebras}, where the principal difference to our work lies in the consideration of the ultraviolet limit of theories already defined in the continuum rather than of the infrared limit of lattice theories.

The article is structured as follows.
In Section \ref{sec:pre}, we recall the operator-algebraic formulation of Wilson-Kadanoff renormalization, and we introduce the required objects for the construction of scaling limits such as scaling maps, limit algebras and scaling limits of states and dynamics.
We also introduce a notion of spatial local algebras for lattice models and give a tentative definition of a continuum theory in terms of local observables in finite and infinite volume that reasonable scaling limit should fit into.
In Section \ref{sec:waveletscalar}
we invoke the framework of second quantization for the description of scalar fields at a fixed scale, and wavelet scaling maps are introduced as a generalization of the block-spin transformation, which we complement by essential ingredients from the theory of wavelets.
In Section \ref{sec:continuumlimit},
we construct the scaling limit of free, massive lattice scalar fields and show that the limit state gives the massive free state while
the lattice fields are represented as the continuum field smeared with the scaling function.
Section \ref{sec:otherscaling}, we discuss the block-spin transformation from the perspective of the wavelet scaling map, and we explain how the former leads to a continuum limit albeit in a more singular sense. In addition, we illustrate two other commonly used renormalization groups, point-like localizations and sharp momentum cutoffs (momentum shells),
and we show that the first leads to singular limits similar to the block-spin case, while the second is well-adapted to dynamics but intrinsically non-local. Furthermore, we explain the relation of our approach to the scaling algebras of Buchholz-Verch and to the MERA.
Section \ref{sec:out} outlines various consequences of the previous sections thereby giving an outlook on topics for future research.

\subsection*{Acknowledgements}
AS would like to thank Tobias J. Osborne and Reinhard F. Werner for useful discussions and comments. VM thanks Domenico Marinucci and Claudio Durastanti for comments and discussions. AS was in part supported by Alexander-von-Humboldt Foundation through a Feodor Lynen Return Fellowship.
The European Research Council Advanced Grant 669240 QUEST supported VM  and partially GM. VM was supported by Indam from March 2019 to February 2020.
Until February 2020 YT was supported by the Programma per giovani ricercatori, anno 2014 ``Rita Levi Montalcini'' of the Italian Ministry of Education, University and Research.
VM, GM and YT also acknowledge the MIUR Excellence Department Project awarded to
the Department of Mathematics, University of Rome ``Tor Vergata'', CUP E83C18000100006 and the University of
Rome ``Tor Vergata'' funding scheme ``Beyond Borders'', CUP E84I19002200005.

\section{Preliminaries}
\label{sec:pre}

\subsection{An operator-algebraic renormalization group scheme}
\label{sec:oaren}
\paragraph{The Wilson-Kadanoff renormalization group.}
A widely used scheme to analyze continuum
limits using effective lattice models
is the Wilson-Kadanoff renormalization group \cite{KadanoffScalingLawsFor, WilsonTheRenormalizationGroupKondo, WegnerCorrectionsToScaling}.
For a family of lattices $\{\Lambda_{N}\}_{N\in\NN_{0}}$, say in $\RR^{d}$ or in the torus $\TT_L^{d}$ of side length $2L$,
we consider a sequence of Hamiltonian quantum systems $\{\fA_{N},\cH_{N},H^{(N)}_{0}\}_{N\in\NN_{0}}$ indexed by the level $N\in\NN_{0}$, the logarithmic scale accounting for the relative density of lattice points (increasing with $N$).
At each level $N$, $\fA_{N}\subset \cB(\cH_{N})$ is a concrete $C^{*}$-algebra of observables acting on a Hilbert space $\cH_{N}$, and $H^{(N)}_{0}$ is an initial Hamiltonian with domain $\cD_{N}\subset\cB(\cH_{N})$. The essential equation of renormalization group theory defines the renormalized Hamiltonians $H^{(N)}_{M}$
for $N'-N=M\geq0$ by demanding the equality of partition functions of $H^{(N)}_{M}$ and $H^{(N+M)}_{0}$ (\cite{FisherTheRenormalizationGroup}, see also \cite{FisherRenormalizationGroupTheory, WhiteDensityMatrixFormulation, WhiteDensityMatrixAlgorithms}):
\begin{align}
\label{eq:basicrg}
Z^{(N+M)}_{0} = \Tr_{N+M}\Big(e^{-H^{(N+M)}_{0}}\Big) & = \Tr_{N}\Big(e^{-H^{(N)}_{M}}\Big) = Z^{(N)}_{M},
\end{align}
where both exponentials are assumed to be trace class.
More precisely, the renormalization group should define coarse-graining transformations (or quantum operations) 
 between the spaces of density matrices, $\cS^{1}_{N+M}$ and $\cS^{1}_{N}$, on $\fA_{N+M}$ respectively $\fA_{N}$,
such that: 
\begin{align}
\label{eq:densityrg}
\cE^{N+M}_{N}\Big(\rho^{(N+M)}_{0}\Big) & = \rho^{(N)}_{M},
\end{align}
with $\rho^{(N)}_{M} = (Z^{(N)}_{M})^{-1}e^{-H^{(N)}_{M}}$. The transformation are also know as the descending superoperators \cite{EvenblyAlgorithmsForEntanglement}. Since the choice of levels $N<N'=N+M$ is arbitrary, we naturally require the semi-group property:
\begin{align}
\label{eq:semigrouprg}
\cE^{N'}_{N}\circ\cE^{N''}_{N'} & = \cE^{N''}_{N}, & N < & N' < N''.
\end{align}
At the level of Hamiltonians with corresponding maps, $H^{(N+M)}_{0} \longmapsto H^{(N)}_{M}$, this is summarized by Wilson's \textit{triangle of renormalization} \cite[p. 790]{WilsonTheRenormalizationGroupKondo} in Figure \ref{fig:trianglerg}.
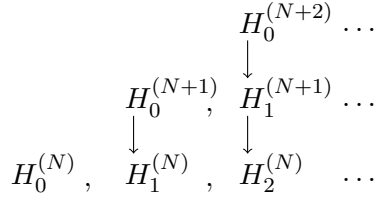
\begin{figure}[ht]
\centering
\begin{tikzpicture}
	\draw (-0.25,0) node[right]{$H^{(N)}_{0}$} (1.25,0) node[right]{$H^{(N)}_{1}$} (2.75,0) node[right]{$H^{(N)}_{2}$};
	\draw (0.9,-0.2) node{$,$};
	\draw (2.5,-0.2) node{$,$};
	\draw (4.5,-0.1) node{$\dots$};
	\draw (1.25,1) node[right]{$H^{(N+1)}_{0}$} (2.75,1) node[right]{$H^{(N+1)}_{1}$};
	\draw (2.5,0.8) node{$,$};
	\draw (4.5,0.9) node{$\dots$};
	\draw[->] (1.5,0.75) to (1.5,0.25);
	\draw[->] (3,0.75) to (3,0.25);
	\draw (2.75,2) node[right]{$H^{(N+2)}_{0}$};
	\draw (4.5,1.9) node{$\dots$};
	\draw[->] (3,1.75) to (3,1.25);
\end{tikzpicture}
\caption{\small Wilson's triangle of renormalization in terms of Hamiltonians of subsequent scales:
Horizontal lines represent sequences of renormalized Hamiltonians $H^{(N)}_{M}$ at a fixed scale that arise via coarse graining from the initial Hamiltonians $H^{(N)}_{0}$ at subsequent scale (vertical lines).}
\label{fig:trianglerg}
\end{figure}

\paragraph{Operator-algebraic renormalization.}
Let us translate \eqref{eq:densityrg} into a statement involving the algebras $\{\fA_{N}\}_{N\in\NN_{0}}$ and states on them (generalizing density matrices).
We find for $a_{N}\in\fA_{N}$:
\begin{align}
\label{eq:densityrgstate}
\Tr_{N}\!\Big(\rho^{(N)}_{M}a_{N}\!\Big) & = \Tr_{N}\!\Big(\cE^{N+M}_{N}(\rho^{(N+M)}_{0})a_{N}\!\Big) = \Tr_{N+M}\!\Big(\rho^{(N+M)}_{0}\alpha^{N}_{N+M}(a_{N})\!\Big)
\end{align}
Here, $\alpha^{N}_{N+M}:\fA_{N}\rightarrow\fA_{N+M}$ is the dual of $\cE^{N+M}_{N}$ (or an ascending superoperators as in \cite{EvenblyAlgorithmsForEntanglement}). Thus, for a given family of initial states
$\{\omega^{(N)}_{0}\}_{N\in\NN_{0}}$ instead of Hamiltonian $\{H^{N}_{0}\}_{N\in\NN_{0}}$, \eqref{eq:densityrg} generalizes to:
\begin{align}
\label{eq:staterg}
\omega^{(N+M)}_{0}\circ\alpha^{N}_{N+M} & = \omega^{(N)}_{M}.
\end{align}
A natural requirement is that $\alpha^{N}_{N+M}$ is unital and completely positive because it should map states into states and preserve probability as expressed by the equality of partition functions \eqref{eq:basicrg}. We further require that $\alpha^{N}_{N'}$ is a $*$-morphism,
so that we can define the inductive limit algebra.
The semi-group property \eqref{eq:semigrouprg} translates to:
\begin{align}
\label{eq:semigrouprgmorph}
\alpha^{N'}_{N''}\circ\alpha^{N}_{N'} & = \alpha^{N}_{N''}.
\end{align}
We call the collection $\{\alpha^{N}_{N'}\}_{N<N'\in\NN_{0}}$, the \textit{scaling maps} or \textit{renormalization group}. At the level of algebras and states,
Figure \ref{fig:statetrianglerg} provides an analogue to Figure \ref{fig:trianglerg}:
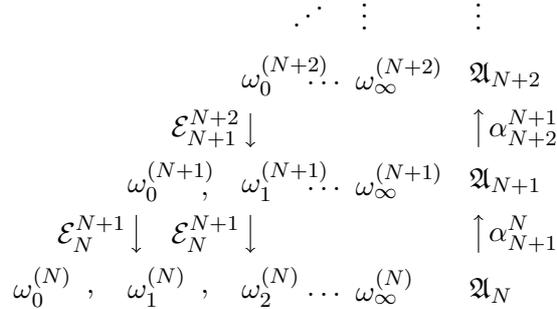
\begin{figure}[ht]
\centering
\begin{tikzpicture}
	\draw (-0.25,0) node[right]{$\omega^{(N)}_{0}$} (1.25,0) node[right]{$\omega^{(N)}_{1}$} (2.75,0) node[right]{$\omega^{(N)}_{2}$} (4.25,0) node[right]{$\omega^{(N)}_\infty$} (5.75,-0.05) node[right]{$\fA_N$};
	\draw (0.9,-0.1) node{$,$};
	\draw (2.4,-0.1) node{$,$};
	\draw (4,-0.1) node{$\dots$};
	\draw (1.25,1.4) node[right]{$\omega^{(N+1)}_{0}$} (2.75,1.4) node[right]{$\omega^{(N+1)}_{1}$} (4.25,1.4) node[right]{$\omega^{(N+1)}_\infty$}  (5.75,1.4) node[right]{$\fA_{N+1}$};
	\draw (2.4,1.2) node{$,$};
	\draw (4,1.3) node{$\dots$};
	\draw (1.5,0.7) node[left]{$\cE^{N+1}_{N}$} (3,0.7) node[left]{$\cE^{N+1}_{N}$};
	\draw[->] (1.5,0.95) to (1.5,0.45);
	\draw[->] (3,0.95) to (3,0.45);
	\draw[<-] (6,0.95) to (6,0.45); 
	\draw (6,0.7) node[right]{$\alpha^{N}_{N+1}$};
	\draw (2.75,2.8) node[right]{$\omega^{(N+2)}_{0}$} (4.25,2.8) node[right]{$\omega^{(N+2)}_\infty$} (5.75,2.8) node[right]{$\fA_{N+2}$};
	\draw (4,2.7) node{$\dots$};
	\draw (3,2.1) node[left]{$\cE^{N+2}_{N+1}$};
	\draw[->] (3,2.3) to (3,1.85);
	\draw[<-] (6,2.35) to (6,1.85); 
	\draw (6,2.1) node[right]{$\alpha^{N+1}_{N+2}$};

	\node at (3.75,3.7) {$\iddots$};
	\node at (4.5,3.7) {$\vdots$};
	\node at (6,3.7) {$\vdots$};
\end{tikzpicture}
\caption{\small Wilson's triangle of renormalization in terms of algebras and states: Vertical lines represent renormalization steps, either by coarse graining states ($\cE$'s) or by refining observables or basic fields ($\alpha$'s). Horizontal lines represent sequences of renormalized states at a fixed scale.}
\label{fig:statetrianglerg}
\end{figure}

\begin{example}
\label{ex:twistedpartialtrace}
A simple example of quantum operations $\{\cE^{N'}_{N}\}_{N<N'\in\NN_{0}}$ are partial traces $\Tr_{N'\rightarrow N}$ for $\cH_{N} = \cH_{0}^{\otimes 2^{N}}$ and $\fA_{N} = B(\cH_{N})$ such that $\alpha^{N}_{N'}(a_{N}) = a_{N}\otimes\mathds{1}_{N'- N}$. Natural generalizations of the latter arising in the context of lattice gauge theory are \cite{MilstedQuantumYangMills, BrothierConstructionsOfConformal, BrothierAnOperatorAlgebraic}:
\begin{align}
\label{eq:twistedpartialtrace}
\alpha^{N}_{N'}(a_{N}) & = U_{N'}(a_{N}\otimes\mathds{1}_{N'- N})U_{N'}^{*},	
\end{align}
for some unitary $U_{N'}\in\cU(\cH_{N'})$.
\end{example}

Equation~\eqref{eq:twistedpartialtrace} is satisfied for our construction of a renormalization group for lattice scalar fields
and is analogous to MERA (multi-scale entanglement renormalization), see \eqref{eq:tensorconj}.

\paragraph{Continuum limit, inductive limit algebra and scaling limit state.}
Our operator-algebraic renormalization provides an algorithm for the construction of a continuum limit of lattice models, which we roughly split into three parts:
First, we construct a family of scaling maps $\{\alpha_{N+1}^N\}_{N\in\NN_{0}}$ between algebras of lattice
fields to define the renormalization group arriving at an inductive limit algebra (see Section \ref{sec:indlim} for a summary of the construction):
\begin{align}
\label{eq:algrglimit}
\fA_{\infty} & = \varinjlim_{N\in\NN_{0}}\fA_{N}.
\end{align}
Second, we consider a sequence of initial states $\{\omega_0^{(N)}\}_{N\in\NN_{0}}$ (at the upside of the triangle in Figure \ref{fig:statetrianglerg}). Their restrictions to coarser lattices are determined by the renormalization group elements $\alpha_{N'}^N$ (generalizing the coarse-graining transformations $\{\cE_{N'}^N\}$ via pullback) according to \eqref{eq:staterg}. On each lattice algebra $\fA_{N}$, the sequence $\{\omega^{(N)}_{M}\}_{M\in\NN_{0}}$ (horizontal lines)
should admit a limit state,
\begin{align}
\label{eq:staterglimit}
\omega^{(N)}_{\infty} & = \lim_{M\rightarrow\infty}\omega^{(N)}_{M},
\end{align}
at least for a subsequence, 
and we expect the limit state to be stable under coarse graining,
\begin{align}
\label{eq:statergconsistency}
\omega^{(N')}_{\infty}\circ\alpha^{N}_{N'} & = \omega^{(N)}_{\infty},
\end{align}
because formally $\alpha^{N'}_{\infty}\circ\alpha^{N}_{N'}\!=\!\alpha^{N}_{\infty}$ for $\alpha^{N}_{\infty}\!=\!\lim_{N'\rightarrow\infty}\alpha^{N}_{N'}$. Algebraically, the consistency expressed by this stability condition would allow for the existence of a projective-limit state on the inductive-limit algebra $\fA_{\infty}$ (cf.~Section \ref{sec:indlim}),
\begin{align}
\label{eq:projstaterg}
\omega^{(\infty)}_{\infty} & = \varprojlim_{N\in\NN_{0}}\omega^{(N)}_{\infty}.
\end{align}
Such projective-limit states will be called \textit{scaling limits} of the initial sequence $\{\omega^{(N)}_{0}\}_{N\in\NN_{0}}$, and their existence together with the property \eqref{eq:statergconsistency} is assured under rather mild conditions as shown in Proposition \ref{prop:weakstarconv} below. The expected non-uniqueness of \eqref{eq:staterglimit} and, therefore, of the scaling limits is physically meaningful (e.g.~renormalization group trajectories and phase transitions) and is partially reflected in the need for \textit{renormalization conditions}, e.g.~on the couplings implicitly present in the initial states.\\
Finally given a scaling limit $\omega^{(\infty)}_{\infty}$ together with the algebra $\fA_{\infty}$, we perform the Gelfand-Naimark-Segal (GNS) construction to arrive at a Hilbert space representation,
\begin{align}
\label{eq:GNSrglimit}
\pi^{(\infty)}_{\infty} : \fA_{\infty} & \longrightarrow \cH^{(\infty)}_{\infty}, & \omega^{(\infty)}_{\infty}(a) & = \langle\Omega_{\infty},\pi^{(\infty)}_{\infty}(a)\Omega_{\infty}\rangle_{\cH^{(\infty)}_{\infty}},
\end{align}
such that the scaling limit is implemented by the vector state $\Omega_{\infty}\in\cH^{(\infty)}_{\infty}$. Now, $\cH^{(\infty)}_{\infty}$ should be regarded as the Hilbert space of the continuum limit, we can identify the elements of each $\fA_{N}$ with certain operators of the continuum field via scaling maps as expressed by \eqref{eq:algrglimit}) and compare the expected properties. As we are working in a Hamiltonian setting, the continuum limit should be interpreted in terms of  time-zero fields. Moreover, natural candidates for non-trivial scaling limits are families of ground states $\{\omega_{\lambda_{N}}\}_{N\in\NN_{0}}$ for lattice Hamiltonians $\{H_{\lambda_{N}}\}_{N\in\NN_{0}}$ admitting quantum critical points\footnote{Physically, a necessary condition for the existence of a non-trivial sequence of limit states $\{\omega^{(N)}_{\infty}\}_{N\in\NN_{0}}$ is the divergence (in units of lattice length) of the sequence of correlation lengths $\{\xi^{(N)}_{N'}\}_{N'\in\NN_{0}}$ associated with pairs of sequences $\{\alpha^{N}_{N'}(a_{N}), \alpha^{N}_{N'}(b_{N})\}_{N'\in\NN_{0}}$ of local operators for any $N'\in\NN_{0}$.} (not necessarily quantum phase transitions) with respect to their respective (dimensionless) couplings $\{\lambda_{N}\}_{N\in\NN_{0}}$, see e.g.~\cite{SachdevQuantumPhaseTransitions}.
In the case of harmonic lattice fields, we will see that families of lattice ground states admit the vacuum states of arbitrary masses (reflecting the non-uniqueness) of the continuum free field as scaling limits (see also Section \ref{sec:other}).

\paragraph{Dynamics.}
Besides the existence of a scaling limit $\{\omega^{(\infty)}_{\infty}\}$, we can also analyze the convergence of the family of time-evolution groups $\{\eta^{(N)}_{t} = e^{itH^{(N)}_{0}}(\!\ .\!\ )e^{-itH^{(N)}_{0}}\}_{N\in\NN_{0}}$. In view of the similarities of our renormalization group scheme with the construction of thermodynamical limits in quantum statistical mechanics, we consider, for this purpose, sequences of the form,
\begin{align}
\label{eq:limdyn}
\{\alpha^{N'}_{\infty}(\eta^{(N')}_{t}(\alpha^{N}_{N'}(a_{N})))\}_{N'>N},
\end{align}
for fixed $t\in\RR$ and $a_{N}\in\fA_{N}$ and ask whether these are Cauchy sequences in the Hilbert space representation relative to the scaling limit $\omega^{(\infty)}_{\infty}$, i.e.~for a suitable operator topology on $\pi^{(\infty)}_{\infty}(\fA_{\infty})$ relative to the scaling limit  \cite{HaagLocalQuantumPhysics, BratteliOperatorAlgebrasAnd2, NachtergaeleOnTheExistence}. This way we may define the limit $\eta^{(\infty)}_{t}\!=\!\lim_{N\rightarrow\infty}\eta^{(N)}_{t}$ on the closure $\cA = \overline{\pi^{(\infty)}_{\infty}(\fA_{\infty})}$ and obtain a scaling-limit Hamiltonian $H^{(\infty)}_{\infty}$ (corresponding to an extrapolation of the upside of Wilson's triangle in Figure \ref{fig:trianglerg}) given suitable continuity properties of $\eta^{(\infty)}:\RR \curvearrowright \cA$.

\paragraph{Free field as an example.}
For general interacting lattice models, we do not expect to find all the above objects in closed form, but suitable expansion or perturbation methods will be required to obtain approximations \cite{BorgsConfinementDeconfinementAnd}. Moreover, the extent to which it will be possible to carry out this formulation of Wilson-Kadanoff renormalization will depend sensitively on the choice of scaling maps $\alpha^{N}_{N'}$, and the amount of control over the state space $\fS(\fA_{N}) = \fS_{N}$ and the effective Hamiltonian $H^{(N)}_{0}$ of each lattice system $(\fA_{N},\cH_{N})$ . Nevertheless, we show in the following that this scheme can be carried out in real space for the free scalar field in any dimension with full control over all involved objects.

\subsection{Inductive limit of \texorpdfstring{$C^*$}{C*}-algebras and representations}
\label{sec:indlim}

The construction of a sequence of lattice $C^*$-algebras and its inductive limit is central to our construction of a continuum limit. Therefore, we recall this fundamental architecture to present detailed and self-contained understanding of the limit procedure. For further references see for instance \cite{KadisonFundamentalsOfThe2, TakesakiTheoryOfOperator3, BlackadarOperatorAlgebras}

Following the notation of the previous section, let $\{\mathfrak{A}_N\}_{N\in\NN_{0}}$ be a sequence of unital $C^*$-algebras.
For $N < N'<N''$, we assume that there exists a unital injective  $^*$-morphism $\alpha_{N'}^{N}$ from $\fA_{N}$ to $\fA_{N'}$ such that
$\alpha_{N''}^{N'}\circ\alpha_{N'}^{N}=\alpha_{N''}^{N}$. Since $\alpha_{N'}^{N}$ is injective,
it is also isometric. $\{\fA_N\}_{N\in\NN_{0}}$ together with $\{\alpha_{N'}^{N}:N, N'\in\NN_{0}, N<N'\}$ is called a \textit{directed system of $C^*$-algebras}. Such a directed system of $C^*$-algebras can be embedded up to $^*$-isomorphisms into a $C^*$-algebra called the \textit{inductive limit}. In particular, there exists a $C^*$-algebra $\fA_{\infty}$ such that:
\begin{enumerate}
\item[(i)] for every $N\in\NN_{0}$ there exists a unital injective $^*$-homomorphism $\alpha^{N}_{\infty}$ from $\fA_{N}$ to $\fA_{\infty}$. Furthermore, if $ N< N'$ then $\alpha^{N}_{\infty} = \alpha^{N'}_{\infty}\circ \alpha_{N'}^{N}$ and $\bigcup_{N\in\NN_{0}} \alpha^{N}_{\infty}(\fA_{N})$ is dense in $\fA_{\infty}$;
\item[(ii)] the $C^*$-algebra $\fA_{\infty}$ in (i) is unique up to a $^*$-isomorphism; in particular, if $\fB_{\infty}$ is another $C^*$-algebra with $^*$-homomorphisms $\beta^{N}_{\infty}$ satisfying conditions as in (i), then there exists a $^*$-isomorphism $\mathcal J:\fB_{\infty}\rightarrow\fA_{\infty}$ such that  $\alpha^{N}_{\infty}=\mathcal J\circ\beta^{N}_{\infty}$.
\end{enumerate}

Next, we consider a family of states $\{\omega^{(N)}\}_{N\in\NN_{0}}$ such that $\omega^{(N)}$ is a state on $\fA_{N}$, and we assume that $\omega^{(N)}=\omega^{(N')}\circ\alpha_{N'}^{N}$. Such a family is said to be \emph{projectively consistent}. Then, there is a uniquely defined state $\omega^{(\infty)}$ on $\fA_{\infty}$ such that  $\omega^{(\infty)}\circ\alpha^{N}_{\infty}=\omega^{(N)}$. ${\omega^{(\infty)} = \varprojlim_{N\in\NN_{0}}\omega^{(N)}}$ is called the \textit{projective limit} of the family $\{ \omega^{(N)}\}$.

As our notion of continuum limit from the previous section is intimately linked with the GNS construction, we also recall the following: The GNS representation $(\pi^{(\infty)}, \cH^{(\infty)},\Omega_{\infty})$ of $(\fA_{\infty}, \omega^{(\infty)})$
extends the GNS representation $(\pi^{(N)},\cH^{(N)},\Omega_{N})$ associated with $(\fA_{N},\omega^{(N)})$. More precisely, the GNS construction induces unique isometries $V^{N}_{N'}$ from $\cH^{(N)}$ into $\cH^{(N')}$ such that $\pi^{(N')}(\alpha^{N}_{N'}(a_{N}))V^{N}_{N'} = V^{N}_{N'}\pi^{(N)}(a_{N})$ for all $a_{N}\in\fA_{N}$ and $V^{N}_{N'}\Omega_{N} = \Omega_{N'}$. Therefore, the GNS representations $(\pi^{(N)},\cH^{(N)},\Omega_{N})$ naturally acquire the structure of a compatible direct system of Hilbert spaces, which admits an inductive limit given by $(\pi^{(\infty)}, \cH^{(\infty)},\Omega_{\infty})$ together with isometries $V^{N}_{\infty}$ from $\cH^{(N)}$ into $\cH^{(\infty)}$ such that $\pi^{(\infty)}(\alpha^{N}_{\infty}(a_{N}))V^{N}_{\infty} = V^{N}_{\infty}\pi^{(N)}(a_{N})$ for all $a_{N}\in\fA_{N}$ and $V^{N}_{\infty}\Omega_{N} = \Omega_{\infty}$. The unique normal extension of $\omega^{(\infty)}$ to the von Neumann algebra $\pi^{(\infty)}(\fA_{\infty})''$ will be denoted by the same symbol.

In view of the construction and existence of projectively consistent families of states (cf.~\eqref{eq:staterglimit} and \eqref{eq:projstaterg}), we make the important observation (which generalizes to directed index sets):
\begin{proposition}
\label{prop:weakstarconv}
Let $\{\fA_N\}_{N \in \NN_{0}}$, $\{\alpha^N_{N'}\}_{N < N'}$ be a directed system of $C^*$-algebras, and let $\{\omega^{(N)}\}_{N\in\NN_{0}}$ be a family of states on $\{\fA_N\}_{N\in\NN_{0}}$.
For $M>0$, we define a state by $\omega^{(N)}_{M} := \omega^{(N+M)}\circ\alpha^{N}_{N+M}$ on $\fA_N$.
If, for each $N$, $\lim_{M\rightarrow\infty}\omega^{(N)}_{M} =: \omega^{(N)}_{\infty}$ exists in the weak$^*$ topology,
then it defines a projectively consistent family $\{\omega^{(N)}_{\infty}\}_{N\in\NN_{0}}$
and hence
$\varprojlim_{N\in\NN_{0}}\omega^{(N)}_{\infty} = \omega^{(\infty)}_{\infty}$ is well-defined.
\end{proposition}
\begin{proof}
Because of weak$^*$ convergence, we have:
\begin{align*}
\lim_{M\rightarrow\infty}\Big|\omega^{(N')}_{\infty}\circ\alpha^{N}_{N'}(a_{N}) - \omega^{(N')}_{M}\circ\alpha^{N}_{N'}(a_{N})\Big| & = 0 & \text{ for all } a_{N}\in\fA_{N} 
\end{align*}
for all $N<N'$. Therefore, the following is true for all $N<N'$ (assuming $M>N'-N$):
\begin{align*}
\omega^{(N)}_{\infty} & = \lim_{M\rightarrow\infty}\omega^{(N)}_{M} = \lim_{M\rightarrow\infty}\omega^{(N+M)}\circ\alpha^{N}_{N+M}
= \lim_{M\rightarrow\infty}\omega^{(N+M)}\circ\alpha^{N'}_{N+M}\circ\alpha^{N}_{N'} \\
&= \lim_{M\rightarrow\infty}\omega^{(N')}_{N+M-N'}\circ\alpha^{N}_{N'} 
 = \omega^{(N')}_{\infty}\circ\alpha^{N}_{N'}.
\end{align*}
\hfill$\square$
\end{proof}

It may happen that each $\{\omega^{(N)}_{M}\}_{M\in\NN_{0}}$ converges only by taking a subsequence, and, thus, the limit may not be unique. In the context of scaling limits, this corresponds to the non-uniqueness of the vacuum or different coupling constants and, therefore, is physically important. Nevertheless, the proposition implies the existence of scaling limits in general due to the weak$^{*}$ compactness\footnote{Apply a diagonal-sequence argument along $N\in\NN_{0}$ to the convergent subsequences $\{\omega^{(N)}_{M_{k_{N}}}\}_{M_{k_{N}}>0}$.} of the state space of $C^{*}$-algebras (see also Section \ref{sec:other}).

\subsection{Lattice models and local algebras}
\label{sec:localg}
We are mainly interested in field theories on hypercubic lattices $\Lambda_N$
and their scaling limits. This means that each $\fA_N$ is equipped with a local structure as follows.

As in Section \ref{sec:oaren}, we associate to each $N$, the scale, a spatial lattice
$\Lambda_N$ with lattice spacing $\vep_{N}$ such that $\vep_{N+1}<\vep_{N}$. To each lattice point $x \in \Lambda_N$ we assign
a local algebra $\fA_N(x)$. The algebras $\fA_N(x)$ and $\fA_N(y)$ should commute
if $x\neq y$, and they generate the whole algebra $\fA_N$.
For a subset $X$ of $\Lambda_N$, let us denote $\fA_N(X) = \overline{\bigcup_{x\in X} \fA_N(x)}^{\|\cdot\|}$.

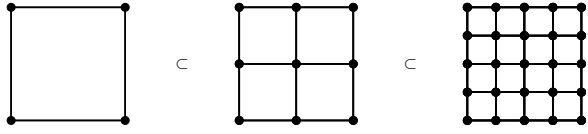
\begin{figure}[t]
\centering
\scalebox{0.75}{
	\begin{tikzpicture}
	
	\foreach \x in {0} 
		\foreach \y in {0,1}
		\foreach \v in {0,1} 	
		{
			\draw[thick] (1+2*\x,4+2*\y) to (3+2*\x,4+2*\y);
			\draw[thick] (1+2*\y,4+2*\x) to (1+2*\y,6+2*\x);
			\filldraw (1+2*\y,4+2*\v) circle (2pt);
		}
	
	\draw (4,5) node{\scriptsize $\subset$};
	
	\foreach \x in {0,1} 
		\foreach \y in {0,...,2}
		\foreach \v in {0,...,2} 	
		{
			\draw[thick] (5+\x,4+\y) to (6+\x,4+\y);
			\draw[thick] (5+\y,4+\x) to (5+\y,5+\x);
			\filldraw (5+\y,4+\v) circle (2pt);
		}
	
	\draw (8,5) node{\scriptsize $\subset$};
	
	\foreach \x in {0,...,3} 
		\foreach \y in {0,...,4}
		\foreach \v in {0,...,4} 	
		{
			\draw[thick] (9+0.5*\x,4+0.5*\y) to (9.5+0.5*\x,4+0.5*\y);
			\draw[thick] (9+0.5*\y,4+0.5*\x) to (9+0.5*\y,4.5+0.5*\x);
			\filldraw (9+0.5*\y,4+0.5*\v) circle (2pt);
		}

\end{tikzpicture}
}
\caption{Geometric subdivisions of the lattice $\Lambda^{L}_{1}$ for $d=2$.}
\label{fig:realsubdiv1}
\end{figure}
Furthermore, we assume that the image of the renormalization map $\alpha_{N+1}^N(\fA_N(x))$ should be
included in $\fA_{N+1}(S_x)$, where $S_x = \{y \in \Lambda_{N+1}: y-x \in [0, \vep_{N+1} r_{\max}]^d\}$, and $r_{\max}>0$ does not depend on $N$. We call the collection $\{\fA_N, \alpha^N_{N'}, \Lambda_N\}_{N, N' \in \NN, N' > N}$ a \emph{lattice field theory}.

In the following, we take standard hypercubic lattices in $\RR^{d}$ with $\vep_N = 2^{-N}\vep$ for some $\vep>0$, hence $\Lambda_N \subset \Lambda_{N+1}$.
However, the renormalization map $\alpha_{N+1}^N$ is in general
{\it not} just the identification of $\alpha_{N+1}^N(\fA_N(x))$ with $\fA_{N+1}(x)$ for $x \in \Lambda_N \subset \Lambda_{N+1}$.
Indeed, the key to obtain the continuum field as the scaling limit is to identify
a lattice field in $\fA_N$ with a smeared field in the continuum, not a point-like field.
This naturally leads to the wavelet scaling in Section \ref{sec:waveletscaling}.

By the above, the increase in support due to $\alpha_{N'}^N$ for an element $a_N \in \fA_N$ is bounded from above by $r_{\max}\vep_{N}\left(1-2^{-(N'-N)}\right)$.
Thus, we can define local algebras $\fA_\infty(S)\subset\fA_\infty$ for suitable open domains $S\subset\TT_{L}^{d} = [-L,L)^d$ by collecting at each level $N$ all the operators $a_{N}$ with support in the sublattice $\Lambda_{N}(S) \subset \Lambda_{N}\cap S$ with the convention that the cube $x+[0,\vep_{N} r_{\max}]^{d}$ does not intersect the boundary $\partial S$ for each site $x\in\Lambda_{N}( S)$, see Figure \ref{fig:locallat}. The bound on the increase of support ensures that this definition is compatible with the equivalence classes formed with respect to the inductive system \eqref{eq:weylmra}:
\begin{align}
\label{eq:spalocalg}
\fA_\infty( S)&\!=\!\varinjlim_{N\in\NN_{0}}\fA_{N}(\Lambda_N( S)),
\end{align}
It immediately follows from locality at level $N$ that
\begin{align}
\label{eq:spaisotony}
& \fA_\infty( S)\subset \fA_\infty( S') &  S&\subset S', \\ 
\label{eq:spalocality}
& [\fA_\infty( S),\fA_\infty( S')] = \{0\} &  S&\cap S'=\emptyset.
\end{align}
We define $\fA_\infty = \overline{\bigcup_{ S}\fA_\infty( S)}$ a quasi-local algebra \cite{BratteliOperatorAlgebrasAnd1}.

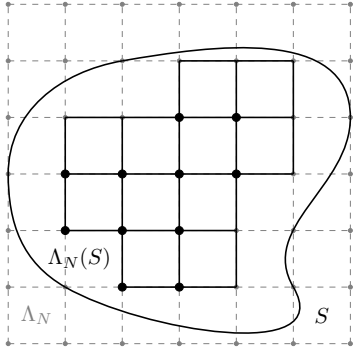
\begin{figure}[ht]
\centering
\scalebox{0.75}{
	\begin{tikzpicture}
	
	
	\foreach \x in {1,...,6} 
		\foreach \y in {1,...,7}	
		{
			\draw[dashed, thin, gray] (1+\x,1+\y) to (2+\x,1+\y);
			\draw[dashed, thin, gray] (1+\y,1+\x) to (1+\y,2+\x);
		}
	\foreach \v in {1,...,7} 
		\foreach \w in {1,...,7}
		{
			\filldraw[gray] (1+\v,1+\w) circle (1pt);
		}
	
	\draw[thick] (3,3) to[out=150, in=270] (2,5);
	\draw[thick] (2,5) to[out=90,in=190] (4,7);
	\draw[thick] (4,7) to[out=10,in=90] (8,6);
	\draw[thick] (8,6) to[out=270,in=120] (7,3);
	\draw[thick] (7,3) to[out=300,in=330] (3,3);
	
	\foreach \x in {0,1} 
		\foreach \y in {0,...,2}	
		{
			\draw[thick] (4+\x,4+\y) to (5+\x,4+\y);
			\draw[thick] (4+\y,4+\x) to (4+\y,5+\x);
		}
	\foreach \v in {0,...,1} 
		\foreach \w in {0,...,1}
		{
			\filldraw (4+\v,4+\w) circle (2pt);
		}
		
	
	\draw[thick] (4,3) to (4,4);
	\draw[thick] (5,3) to (5,4);
	\draw[thick] (6,3) to (6,4);
	\draw[thick] (3,4) to (3,5);
	\draw[thick] (3,5) to (3,6);
	\draw[thick] (5,6) to (5,7);
	\draw[thick] (6,6) to (6,7);
	\draw[thick] (7,6) to (7,7);
	\draw[thick] (7,5) to (7,6);
	
	\filldraw (4,3) circle (2pt) (5,3) circle (2pt) (3,4) circle (2pt) (3,5) circle (2pt) (6,5) circle (2pt) (5,6) circle (2pt) (6,6) circle (2pt);
	
	
	\draw[thick] (4,3) to (5,3);	
	\draw[thick] (5,3) to (6,3);
	\draw[thick] (3,4) to (4,4);
	\draw[thick] (3,5) to (4,5);
	\draw[thick] (3,6) to (4,6);
	\draw[thick] (5,7) to (6,7);
	\draw[thick] (6,7) to (7,7);
	\draw[thick] (6,6) to (7,6);
	\draw[thick] (6,5) to (7,5);
		
	\draw[gray] (2.5,2.5) node{$\Lambda_{N}$};
	\draw (7.5,2.5) node{$ S$};
	\draw (3.25,3.5) node{$\Lambda_{N}( S)$};

\end{tikzpicture}
}
\caption{\small Illustration of a localized sublattice $\Lambda_{N}( S)$ (the black sites) for $r_{\max}=1$ in dimension $d=2$.
Thick lines attached to black sites indicate boundaries of support in the scaling limit.}
\label{fig:locallat}
\end{figure}

\subsection{Continuum field theory}\label{sec:continuum}
Let $\omega^{(\infty)}_\infty$ be the scaling limit state \eqref{eq:projstaterg}
of the sequence $\{\omega_0^{(N)}\}$ on a lattice field theory $\{\fA_N, \alpha^N_{N'}, \Lambda_N\}$.
In the GNS representation $\pi_\infty^{(\infty)}$ of $\fA_\infty$ with respect to $\omega^{(\infty)}_\infty$,
we obtain a family of von Neumann algebras $\pi_\infty^{(\infty)}(\fA_\infty(S))''$.
These algebras should be a continuum field theory, in the following sense.

We say that $(\cA, U, \Omega)$ is a {\bf continuum time-zero net of observables} in $\cH$,
where $\cA(S)$ is a von Neumann algebra for each region $S \subset \TT^d_L$,
$U$ is a strongly continuous unitary representation of $\TT^d_L \times \RR$
and $\Omega$ is a vector such that
\begin{enumerate}[{(}1{)}]
 \item If $S_1 \subset S_2 \subset \TT^d_L$, then $\cA(S_1) \subset \cA(S_2)$.
 \item $\Ad U(a,0)(\cA(S)) = \cA(S + a)$ for $a\in \TT^d_L$.
 \item There is $c' > 0$ such that, if $d(S_1, S_2) > \epsilon c'$
 (where $d(S_1, S_2) = \inf_{x \in S_1, y \in S_2, n \in \ZZ^d} |x-y+2Ln|$ is the distance on $\TT^d_L$ between two regions $S_1, S_2$),
 then $\Ad U(0,t)(\cA(S_1))$ commutes with $\cA(S_2)$
 for $t < \epsilon$.
 \item $U(a,t)\Omega = \Omega$.
 \item The generator of the one-parameter group $U(0,t)$ is positive.
\end{enumerate}
If $\TT^d_L$ is replaced by $\RR^d$, we call it
a {\bf continuum (infinite volume) time-zero net of observables}.
If $c=1$, these properties are restrictions of the Haag-Kastler axioms \cite{HaagLocalQuantumPhysics}
to the time-zero plane and the restriction of any Haag-Kastler net to the $t=0$ plane satisfies them.
We do not include uniqueness of the vacuum, because
it may fail for a physical reason in the construction through continuum limit (phase transition).

The Poincar\'e covariance does not follow from these time-zero axioms.
Indeed, it is an additional requirement that $U$ extends to a representation
of the Poincar\'e group. A counterexample should be obtained
if we start with the states with a wrong (nonrelativistic) dispersion relation (see e.g.~\cite{YngvasonNoteOnEssentialDuality}).

\section{Wavelet scaling of lattice scalar fields}\label{sec:waveletscalar}
\subsection{Scalar fields on lattices}
\label{sec:latscalar}

\subsubsection{First and second quantization}
\label{sec:quant}

The Weyl algebra is defined by the canonical commutation relations
(a general reference is \cite{BratteliOperatorAlgebrasAnd2}).
Let $\fh$ be a Hilbert space (which will be called the {\bf one-particle space}),
$\langle\cdot,\cdot\rangle$ its scalar product and $\sigma (\cdot,\cdot)=\Im \langle\cdot,\cdot\rangle$, where $\Im$ denotes the imaginary part,
the canonical non degenerate symplectic form.
A $C^*$-algebra $\cW(\fh)$ is said to be the {\bf Weyl algebra  associated with $(\fh,\sigma)$}
if $\cW(\fh)$ is generated by elements $\{W(\xi)\}_{\xi\in \fh}$
with the following commutation relations:
\begin{equation}\label{eq:CCR}
W(\xi)W(\eta)=e^{-\frac{i}{2}\sigma(\xi,\eta)}W(\xi+\eta), \qquad \xi,\eta\in \fh.
\end{equation}
Such a $C^*$-algebra is unique (up to isomorphism) and simple.

A fundamental representation of $\cW(\fh)$ is constructed on
the {\bf bosonic Fock space over $\fh$}: $\fF_+(\fh)=\CC\Omega\,\oplus\,\bigoplus_{n\in\NN_0}\fh^{\otimes_s n}$, where $\fh^{\otimes_s n}$ is the symmetric $n$-fold tensor product
of the one-particle space $\fh$ and $\Omega$ is called the {\bf Fock vacuum}.
On $\fF_+(\fh)$, the actions of the Weyl operators $W(\xi)$ are uniquely given by
\begin{equation*}
W(\xi)\Omega=e^{-\frac12\|\xi\|^2}e^\xi, \qquad \text{ where } \xi \in\fh \text{ and }
e^\xi = 1\oplus \xi \oplus \tfrac1{\sqrt{2!}}\xi^{\otimes 2}\oplus \cdots \oplus \tfrac1{\sqrt{n!}}\xi^{\otimes n}\oplus \cdots.
\end{equation*}
and by the canonical commutation relations \eqref{eq:CCR}:
This is the GNS representation of the algebra $\cW(\fh)$ with respect to the state $\omega(W(\xi))=\langle\Omega,W(\xi)\Omega\rangle=e^{-\frac12\|\xi\|^2}$.

We use the following repeatedly, and hence state  it as a proposition.
\begin{proposition}\label{prop:secondquantization}
A symplectic map $R$ from $\fh_1$ into $\fh_2$ induces a natural injective $^*$-homomorphism from $\cW(\fh_1)$ to $\cW(\fh_2)$,
which maps $W(\xi)$ to $W(R(\xi))$, and is a $^*$-isomorphism if $R$ is bijective.
\end{proposition}

Note that, even if $\fh_1 \subset \fh_2$, we {\it may take} a symplectic map $R$ which is different from
the inclusion map. Indeed, such different embedding correspond to
different scaling maps, see Section \ref{sec:waveletscaling}.

\subsubsection{Scalar fields at different scales}
\label{sec:freefield}

As in Section \ref{sec:quant}, to construct the Weyl algebra, one has firstly
to define the one-particle Hilbert space.
Let us consider an initial lattice
\[
 \Lambda_{\vep,r} = \vep\{-r,...,0,...,r-1\}^{d}\subset\RR^{d}
\]
with scale parameters\footnote{The scale parameters $\vep, r$ will be fixed more explicitly below.} $\vep>0$, $r\in\NN$.
We think of $\Lambda_{\vep,r}$ as a discretization of the cube $[-L,L)^{d}=\TT_{L}^{d}$ with periodic boundary conditions ($r\equiv -r$) which fixes the product of scale parameters $\vep r = L>0$. From $\Lambda_{\vep,r}$ we generate a sequence of lattices $\Lambda_{N} = \Lambda_{\vep_{N},r_{N}}$ with $\vep_{N} = 2^{-N}\vep$, $r_{N}=2^{N}r$ for $N\in\NN_{0}$. This way all lattices have the same volume $\vol(\Lambda_{N})=(2\vep_{N}r_{N})^{d}=(2L)^{d}$. 
In the following, we also need the dual lattice $$\Gamma_{\vep,r} = \tfrac{\pi}{\vep r}\{-r,...,0,...,r-1\}^d$$ and its scaled versions $\Gamma_{N} = \Gamma_{\vep_{N},r_{N}}$, with the scaling parameter defined above. 

We introduce two $(2^N r)^d$-dimensional Hilbert spaces associated with the lattices  and their dual\footnote{We emphasize the dependence on the lattice
size $L$ only on Hilbert spaces and algebras, although most of other objects depend implicitly on $L$.}
\begin{align}\label{eq:hilLat}
\fh_{N,L} & = \ell^2(\Lambda_N), & \langle\xi,\eta\rangle_{N,L} & = \sum_{x\in\Lambda_N} \bar\xi(x)\eta(x), \\
\label{eq:hilLatft}
\hat{\fh}_{N,L} & = L^2(\Gamma_N, (2r_N)^{-d}\mu_{\Gamma_N}), & \langle\hat{\xi},\hat{\eta}\rangle_{N,L} & =(2r_N)^{-d}\sum_{k\in\Gamma_N} \bar{\hat{\xi}}(k)\hat{\eta}(k),
\end{align}
where the counting measure $\mu_{\Gamma_{N}}$ acquires a factor $(2 r_{N}\!)^{-d}$ because of the normalization $(2 r_{N}\!)^{-d}\sum_{k\in\Gamma_{N}}e^{i k\cdot(x-x')}=\delta_{x,x'}$, $x,x'\in \Lambda_{N}$.

These two Hilbert spaces are identified via the discrete Fourier transform:
\begin{align}
\label{eq:Latft}
\mathscr{F}_N[\xi](k) = \hat{\xi}(k) & = \!\!\!\sum_{x\in\Lambda_N}\!\!\!\xi(x) e^{-i k\cdot x}, & \mathscr{F}_N^{-1}[\hat{\xi}](x) = \check{\hat{\xi}}(x) & = \tfrac{1}{(2r_N)^{d}}\!\!\!\sum_{k\in\Gamma_N}\!\!\!\hat{\xi}(k) e^{i k\cdot x} = \xi(x).
\end{align}

The kinematical scalar field lattice systems are given in terms of the one-particle spaces
$\fh_{N,L}$, cp.~\cite{BattleWaveletsAndRenormalization, BratteliOperatorAlgebrasAnd2}:
\begin{align}
\label{eq:scalarlatticesystems}
\cW(\fh_{N,L})\!&=\!\cW_{N,L}, & \cH_{N,L} &\!=\!\fF_{+}(\fh_{N,L})\!\cong\!\bigotimes_{x\in\Lambda_{N}}\cH_{x},
\end{align}
where $\cH_{x} = L^{2}(\RR)$, and $\cW(\fh_{N,L})$ is the Weyl algebra,
\begin{align}
\label{eq:latticeweyl}
W_{N}(\xi)W_{N}(\eta) &\!=\!e^{-\frac i2\sigma_{N,L}(\xi,\eta)}W_{N}(\xi+\eta), \qquad \xi,\eta\in\fh_{N,L},
\end{align}
with respect to the standard symplectic form, $\sigma_{N,L} = \Im\langle\!\ \cdot \!\ ,\!\ \cdot \!\ \rangle_{N,L}$.
The decomposition into real (Langragian) subspaces is facilitated by
\begin{align}\label{eq:canonicaldimension}
 \xi = \vep_{N}^{\frac{d+1}{2}}q_\xi + i\vep_{N}^{\frac{d-1}{2}}p_\xi  \text{ for }\xi \in\fh_{N,L}
\end{align}
with real-valued $q_\xi, p_\xi \in \fh_{N,L}$.
We denote by $q$, $p$ generic real-valued elements in $\fh_{N,L}$.
The finite volume field $\Phi_N(q)$ and momentum $\Pi_N(p)$ are the generators of the one-parameter groups
$W_N(tq), W_N(tp)$, respectively, and in terms of $(q,p)$ the symplectic form reads
\begin{equation}\label{eq:symplecticpq}
\sigma_{N,L}( (q_\xi,p_\xi), (q_\eta,p_\eta)) := \sigma_{N,L}(\xi,\eta) = \vep_N^d\sum_{x \in \Lambda_N} q_\xi(x)p_\eta(x)-p_\xi(x)q_\eta(x).
\end{equation}
For the Lagrangian decomposition of $\hat \fh_{N,L}$, we choose the normalization:
\begin{align}
\label{eq:fourierscaling}
\hat{q} &\!=\!\vep_{N}^{\frac{d}{2}}\sF_{N}[q], & \hat{p} &\!=\!\vep_{N}^{\frac{d}{2}}\sF_{N}[p].
\end{align}

The algebra $\cW_{N,L}$ is also naturally equipped with a *-automorphic action $\eta^{(N)}_{L}:\Lambda_N\curvearrowright\cW_{N,L}$ of lattice translations, induced via Prop.~\ref{prop:secondquantization} by the action $\tau^{(N)}_{L} : \Lambda_N \curvearrowright \fh_{N,L}$ of translations on the one-particle space, defined by the symplectic maps $(\tau^{(N)}_{L|a} \xi)(x) := \xi(x-a)$, $a,x \in \Lambda_N$, $\xi \in \fh_{N,L}$, where the difference $x-a$ is to be intended $\!\!\!\!\mod 2L \ZZ^{d}$.

\subsubsection{Scaling maps from one-particle spaces}\label{subsubsec:scalingonepart}
A renormalization group for lattice scalar fields (see Section \ref{sec:freefield}) is obtained
from a symplectic map $R_{N'}^N: \fh_{N,L} \to \fh_{N',L}$, $N'>N$, between one-particle spaces:
it induces via Prop.~\ref{prop:secondquantization} an injective $^*$-homomorphism $\alpha^{N}_{N'} : \cW_{N,L} \to \cW_{N',L}$ such that
\begin{align}
\label{eq:oneparticlerg}
\alpha^{N}_{N'}(W_N(\xi_{N}))\!=\!W_{N'}(R^{N}_{N'}(\xi_{N})).
\end{align}
We thus obtain a $C^{*}$-inductive system of lattice Weyl algebras,
\begin{align}
\label{eq:weylmra}
\cW_{0,L} \rightarrow \cdots \rightarrow \cW_{N,L} \rightarrow \cW_{N+1,L}\rightarrow \cdots \rightarrow \cW_{N',L}\rightarrow \cdots.
\end{align}
Although there is a natural inclusion $\fh_{N,L}\subset \fh_{N+1,L}$,
we {\it do not} take it as $R_{N+1}^N$, because this would mean (cf.\! Section \ref{sec:blockspin})
that we identify
an operator on a lattice vertex with an operator localized in one point in the continuum,
which does not exist. Instead, we consider maps associated with various wavelets.

As Weyl algebras are simple, thus $^*$-homomorphisms are injective, and the inductive limit of a directed system of simple $C^*$-algebras is again simple.

\subsection{An instructive example: The block-spin method}
\label{sec:blockspin}
As an illustration of the general scheme introduced above, we define a version of the block-spin renormalization group for lattice scalar fields known from spin systems (e.g. the Ising model), see
\cite{EvenblyEntanglementRenormalizationIn}.

\begin{definition}
\label{def:blockspinscaling}
The block-spin renormalization group $\{\alpha^{N}_{N'}\}_{N<N'}$ is the inductive family of $^*$-homomorphisms defined by the block-spin scaling map between one-particle Hilbert spaces:
\begin{align*}
R^{N}_{N+1} : \fh_{N,L} & \longrightarrow \fh_{N+1,L},
\end{align*}
where
\begin{align}
\label{eq:blockspinscale}
R^{N}_{N+1}(q,p)(x') & = \sum_{x\in\Lambda_{N}}(q,p)(x)\chi_{[0,\vep_{N})^{d}}(x'-x), & N & \in\NN_{0},
\end{align}
(note that $q$ and $p$ in $\fh_{N,L}$ are scaled as \eqref{eq:canonicaldimension},
that compensate the volume of the cube)
namely, the function with support in $x$ is mapped to the step function
supported in the (discrete) cube with ``lower left'' corner $x$.
Other scaling maps are defined by composition, so that the semigroup property is automatic:
\begin{align*}
R^{N}_{N'} & = R^{N'-1}_{N'}\circ R^{N'-2}_{N'-1}\circ ... \circ R^{N}_{N+1}, & N & < N'.
\end{align*}
\end{definition}
We note that $R^{N}_{N+1}$ is symplectic, because the inner product is preserved:
\begin{align*}
\sigma_{N+1,L}(R^{N}_{N+1}(q,p),R^{N}_{N+1}(q',p')) & = \sigma_{N,L}((q,p),(q',p')),
\end{align*}
Intuitively, \eqref{eq:blockspinscale} encodes the idea that the field and its momentum spatially localized at $x$ at scale $N$ result from averaging\footnote{Since we are working with the algebra of field operators, the block-spin transformation results in a refining operation $\fA_{N} = \cW_{N,L}\rightarrow\fA_{N+1} = \cW_{N+1,L}$ \eqref{eq:densityrgstate} in contrast with the familiar coarse-graining operation on states, respectively density matrices \eqref{eq:densityrg}.} over points $x'$ close by at scale $N+1$.
At the level of fields, this yields:
\begin{align}
\label{eq:blockscalefield}
\alpha^N_{N+1}(\Phi_{N}(x))
 & = 2^{-d}\sum_{x'\in\Lambda_{N+1}}\chi_{[0,\vep_{N})}(x'-x)\Phi_{N+1}(x'), \\
\label{eq:blockscalemom}
\alpha^N_{N+1}(\Pi_{N}(x))
 & = 2^{-d}\sum_{x'\in\Lambda_{N+1}}\chi_{[0,\vep_{N})}(x'-x)\Pi_{N+1}(x').
\end{align}

Formally, the block-spin scaling map encodes the relation between block-averaged (continuum) fields and momenta at scales $N\in\NN_{0}$, i.e. we think of the lattice fields and momenta as integrated against characteristic functions of lattice cubes with $\Phi, \Pi$ the continuum field and the continuum momenta,
\begin{align}
\label{eq:blockfield}
\Phi_{N}(x) & = \vep_{N}^{-d}\int_{\mathds{T}^{d}_{L}}d^{d}x'\chi_{[0,\vep_{N})^{d}}(x'-x)\Phi(x') = \vep_{N}^{-d}\Phi(\chi_{x+[0,\vep_{N})^{d}}) \text{ (formal)}, \\
\label{eq:blockmom}
\Pi_{N}(x) & = \vep_{N}^{-d}\int_{\mathds{T}^{d}_{L}}d^{d}x'\chi_{[0,\vep_{N})^{d}}(x'-x)\Pi(x') = \vep_{N}^{-d}\Pi(\chi_{x+[0,\vep_{N})^{d}}) \text{ (formal)},\end{align}
although
$\Pi$ is ill-defined for nondifferentiable functions.

Yet, the idea of embedding lattice fields into the continuum field can be justified
using wavelets. Indeed, if the problem is the non-differentiability of
step functions, we would only have to identify lattice fields with continuum fields smeared over more regular functions, as we see below.

\subsection{A generalization: The wavelet method}
\label{sec:waveletscaling}

\subsubsection{Scaling map from scaling function.}
The block-spin scaling map is a special case of wavelet scaling map which we define here.
The scaling map is induced by the maps between lattice one-particle spaces:
\begin{align}
\label{eq:oneparticlemra}
\fh_{0,L}\cdots\stackrel{R^{N-1}_{N}}{\rightarrow} \fh_{N,L} \stackrel{R^{N}_{N+1}}{\rightarrow} \fh_{N+1,L}\stackrel{R^{N+1}_{N+2}}{\rightarrow} \cdots \stackrel{R^{N'-1}_{N'}}{\rightarrow}\fh_{N',L} \stackrel{R^{N'}_{N'+1}}{\rightarrow} \cdots,
\end{align}
As we hope to embed lattice fields into the continuum field,
there should be the corresponding spaces in the continuum one-particle space:
\begin{align}
\label{eq:mra}
V_{0}\subset \cdots \subset V_{N}\subset V_{N+1}\subset \cdots \subset L^{2}(\TT^d_L).
\end{align}
In the case of the block-spin renormalization,
$V_{N}$ is spanned by the step functions  $\chi_{[0,\vep_N)^d}(\cdot - x)$,
and this is naturally included in $V_{N+1}$. Furthermore, such a step function is
a linear combination of finer step functions with a fixed set of coefficients
(see Section \ref{sec:Haarwavelet} for details of this case).
By generalizing these properties, we are led to consider wavelets.
A general reference is \cite{DaubechiesTenLecturesOn} (for periodic wavelet bases, see \cite[Section 10.7]{DaubechiesTenLecturesOn}).

We start with a {\bf scaling function} $\phi \in L^2(\RR^d)$, and define the functions $\phi^{(\vep)}_{N,k} \in L^2(\TT^d_L)$ as the $2L$-periodization of the rescaled scaling functions $\RR^d \ni x \mapsto \vep_N^{-d/2} \phi(\vep_N^{-1}x-\vep^{-1} k)$, $N \in \NN_{0}, k\in\Lambda_{\vep, r_N} = \{-\vep r_N, \dots, \vep(r_N-1)\}^d \subset (\vep\ZZ)^d$ (recall that $L=\vep r$). In particular we set $\phi^{(\vep)} := \phi^{(\vep)}_{0,0}$. The characteristic properties of $\phi$ translate into the following ones\footnote{Here
the convention is that $n$ belongs to a unit lattice (with lattice spacing $1$) while $k$ belongs to a scaled lattice.}
for the functions $\phi^{(\vep)}_{N,k}$.
\begin{itemize}
 \item $\{\phi^{(\vep)}(\ \cdot \ - k)\}_{k\in\Lambda_{\vep, r}}$ is an orthonormal system.
 \item  It holds that
  \begin{align}\label{eq:unitscalingeq}
   \phi^{(\vep)}(x) & = \sum_{n\in\Lambda_{1, r_1}}h_{n}\phi^{(\vep)}_{1,\vep n}(x),
  \end{align}
 with $h_n \in \RR$.
 \item $\{\phi^{(\vep)}_{N,k}\}_{ N  \in\NN_{0}, k\in\Lambda_{\vep, r_N}}$ span $L^2(\TT^d_L)$.
 \item $\phi^{(\vep)}$ is normalized in such a way that
 \begin{equation}\label{eq:scalingnorm}
 \int_{\TT^d_L} \phi^{(\vep)}(x)dx = \vep^{d/2}.
 \end{equation}
\end{itemize}
Note that from~\eqref{eq:unitscalingeq} and~\eqref{eq:scalingnorm} the sum rule $\sum_{n \in \Lambda_{1,r_1}} h_n = 2^{d/2}$ follows.
 
Such $\phi^{(\vep)}$ gives rise to a (half-sided) multiresolution analysis (MRA)\footnote{In the non-periodic
setting this sequence is two-sided infinite, i.e. indexed by $\ZZ$.}, i.e.~ the sequence of subspaces $V_{N}$
spanned by $\{\phi^{(\vep)}_{N,k}\}_{ k\in\Lambda_{\vep, r_N}}$, $N \in \NN_0$,
satisfying \eqref{eq:mra}
with the properties:
\begin{align}
\label{eq:MRAunion}
& \overline{\bigcup_{N\in\NN_{0}}V_{N}} = L^{2}(\TT_{L}^{d}), \\
\label{eq:MRAintersec}
& \bigcap_{N\in\NN_{0}}V_{N} = V_{0}, \\
\label{eq:MRAscaling}
& f\in V_{0} \Leftrightarrow f(2^{N}\ \cdot \ )\in V_{N}, \\
\label{eq:MRAtranslate}
& f\in V_{0} \Leftrightarrow f(\ \cdot \ - k)\in V_{0} & \forall k\in\vep\ZZ^{d}.
\end{align}

We discuss concrete examples of scaling functions in the next sections.
In case the scaling function $\phi$ is chosen compactly supported, there is $r_{\max}> 0$ such that $h_{n}\neq 0$ only for $\|n\|_{\infty}\le r_{\max}$,
where $\|n\|_{\infty}$ is the largest absolute value of the components of $n\in \Lambda_{1, r_1} \subset \ZZ^d$,
because these are the expansion coefficients of $\phi^{(\vep)}$ in the basis $\{\phi^{(\vep)}_{1,\vep n}\}$,
i.e.~$\{h_{n}\}_{\|n\|_{\infty}\le r_{\max}}$ yields a finite impulse response filter scheme,
and is called a low-pass filter \cite{DaubechiesTenLecturesOn}.
For a fixed $N$, $\{\phi^{(\vep)}_{N,k}\}_{k\in\Lambda_{\vep, r_N}}$ play the role of approximate $\delta$-functions.
We will also use the notation $\phi^{(\vep_N)}_x := \phi^{(\vep)}_{N, 2^N x}$, $x \in \Lambda_N (=\Lambda_{2^{-N}\vep, 2^N r})$, $N \in \NN_0$. Correspondingly, we denote the standard basis of $\ltwo(\Lambda_{N})$ by $\{\delta^{(N)}_{x}\}_{x\in\Lambda_{N}}$. At this point, we demand $\log_{2}r\in\NN_{0}$ as otherwise the completeness of the restricted half-sided MRA \eqref{eq:mra} to $L^{2}(\TT_{L}^{d})$ is problematic \cite{DaubechiesTenLecturesOn}.

Having in mind the identification $\delta^{(0)}_{x} \!\sim\!\phi^{(\vep)}_{x},  x \in\Lambda_{\vep,r}\!=\!\Lambda_{0}$, we make the analogue of Definition \ref{def:blockspinscaling}.

\begin{definition}
\label{def:waveletscaling}
Given a compactly supported scaling function $\phi^{(\vep)}$ with the scaling coefficients
\eqref{eq:unitscalingeq},
the wavelet renormalization group $\{\alpha^{N}_{N'}\}_{N<N'}$ is the inductive family of $^*$-homomorphisms defined by the wavelet scaling map between one-particle Hilbert spaces, $R^{N}_{N+1} : \fh_{N,L} \longrightarrow \fh_{N+1,L}$:
\begin{align}
\label{eq:waveletscale}
R^{N}_{N+1}(q,p)(x) & = 2^{\frac{d}{2}}\sum_{x'\in\Lambda_{N}} (q,p)(x') \sum_{n\in\ZZ^{d}, \|n\|_\infty \le r_{\mathrm{max}}}h_{n}\delta^{(N+1)}_{x'+n\vep_{N+1}}(x),
\end{align}
where the index of $\delta^{(N+1)}_{x'+n\vep_{N+1}}\in\ltwo(\Lambda_{N+1})$ is interpreted $\mod 2r_{N+1}(\vep_{N+1}\ZZ^{d})$, according to the periodization convention above, and
\begin{align*}
R^{N}_{N'} & = R^{N'-1}_{N'}\circ R^{N'-2}_{N'-1}\circ ... \circ R^{N}_{N+1}, & N & < N'.
\end{align*}
\end{definition}

The numerical coefficients in \eqref{eq:waveletscale} are motivated by the formal relations between fields and momenta at successive scales implied by \eqref{eq:unitscalingeq}, cp.~\eqref{eq:blockscalefield} and \eqref{eq:blockscalemom}. In the following, to simplify the notation, we will indicate sums as the second one in~\eqref{eq:waveletscale} simply as sums over $n \in \ZZ^d$, with the convention that $h_n = 0$ for $\|n\|_\infty > r_{\mathrm{max}}$.

Thanks to the property \eqref{eq:unitscalingeq}, we obtain a family of symplectic maps.
\begin{lemma}
\label{lem:waveletscaling}
The scaling maps $R^{N}_{N'}$, $N<N'$, associated with a scaling function $\phi$ are symplectic, i.e.
\begin{align*}
\sigma_{N',L}(R^{N}_{N'}(q,p),R^{N}_{N'}(q',p')) & = \sigma_{N,L}((q,p),(q',p')).
\end{align*}
\end{lemma}
\begin{proof}
A direct computation for $R^{N}_{N+1}$ with $N\in\NN_{0}$ arbitrary yields:
\begin{align*}
 & \vep_{N+1}^{d}\sum_{x'\in\Lambda_{N+1}}R^{N}_{N+1}(p)(x')R^{N}_{N+1}(q')(x') \\
 & = \frac{\vep_{N}^{d}}{2^{d}}\sum_{x'\in\Lambda_{N+1}}\left(2^{\frac{d}{2}}\sum_{x\in\Lambda_{N}}p(x)\sum_{n\in\ZZ^{d}}h_{n}\delta^{(N+1)}_{x+n\vep_{N+1}}(x')\right)\left(2^{\frac{d}{2}}\sum_{y\in\Lambda_{N}}q'(y)\sum_{m\in\ZZ^{d}}h_{m}\delta^{(N+1)}_{y+m\vep_{N+1}}(x')\right) \\
 & = \vep_{N}^{d}\sum_{x,y\in\Lambda_{N}}p(x)q'(y)\sum_{n,m\in\ZZ^{d}}h_{n}h_{m}\delta_{n-m,2\vep_{N}^{-1}(y-x)} = \vep_{N}^{d}\sum_{x,y\in\Lambda_{N}}p(x)q'(y)\underbrace{\sum_{n\in\ZZ^{d}_{r}}h_{n}h_{n+2\vep_{N}^{-1}(x-y)}}_{= \delta_{x,y}} \\[-0.7cm]
 & = \vep_{N}^{d}\sum_{x\in\Lambda_{N}}p(x)q'(x),
\end{align*}
where we used the fact that $2\vep_{N}^{-1}(x-y)\in 2\ZZ^{d}$ for $x,y\in\Lambda_{N}$ and the orthogonality relation $\sum_{n\in\ZZ^{d}}h_{n}h_{n+2m} = \delta_{m,0}$ which follows from the scaling equation \eqref{eq:unitscalingeq} and
the orthonormality of $\{\phi^{(\vep)}_{0,k}\}_{k \in \Lambda_{\vep,r}}$. The statement is then obtained by~\eqref{eq:symplecticpq}.\hfill$\square$
\end{proof}
We note that a priori there is no need to choose the same low-pass filter $\{h_{n}\}_{n\in\ZZ^{d}}$ for the complementary real, Lagrangian subspaces of the decomposition \eqref{eq:canonicaldimension} of $\fh_{N,L}$, and the scaling map $R^{N}_{N'}$ will be symplectic for biorthogonal wavelets as well \cite{DaubechiesTenLecturesOn}.

Therefore, according to the discussion in Section \ref{subsubsec:scalingonepart}, we denote by $\cW_{\infty,L}$ the inductive limit $C^*$-algebra obtained from the inductive system defined by~\eqref{eq:oneparticlerg}, \eqref{eq:weylmra}. Moreover, it is easy to check that for the one-particle lattice translations $\tau^{(N)}_{L|a}$, $a \in \Lambda_N$, there holds $\tau^{(N')}_{L|a}\circ R^N_{N'} = R^N_{N'} \circ\tau^{(N)}_{L|a}$, which of course entails,
for the lattice translation automorphisms $\eta^{(N)}_{L|a}$ of $\cW_{N,L}$,
\begin{equation}\label{eq:translationcoherence}
\eta^{(N')}_{L|a}\circ\alpha^{N}_{N'} = \alpha^{N}_{N'}\circ\eta^{(N)}_{L|a}, \qquad a \in \Lambda_N, \, N' > N.
\end{equation}
This, together with the inductive limit uniqueness, implies the existence of an automorphic action of the dyadic traslations on the inductive limit algebra as follows. Given $a \in \bigcup_N \Lambda_N$ and $N' \in \NN$, since $\Lambda_N \subset \Lambda_{N+1}$ we can assume that $a \in \Lambda_N$ with $N \geq N'$, and we can define an injective *-morphism $\beta^{N'}_\infty : \cW_{N',L} \to \cW_{\infty,L}$ by
\[
\beta^{N'}_\infty(A) := \alpha^{N}_{\infty}(\eta^{(N)}_{L|a}(\alpha^{N'}_{N}(A))), \qquad A \in \cW_{N',L}.
\]
Thanks to the intertwining property~\eqref{eq:translationcoherence}, $\beta^{N'}_{\infty}$ is independent of the chosen $N$ such that $a \in \Lambda_N$, and one immediately checks that $\beta^{N'}_{\infty}\circ\alpha^{N''}_{N'} = \beta^{N''}_\infty$ for $N'\geq N''$. Therefore, by the uniqueness of the $C^*$-inductive limit, there exists a *-automorphism $\eta^{(\infty)}_{L|a} : \cW_{\infty,L} \to \cW_{\infty,L}$ such that $\eta^{(\infty)}_{L|a}\circ\alpha^{N'}_{\infty} = \beta^{N'}_{\infty}$ for all $a \in \bigcup_N \Lambda_N$ and $N' \in \NN$, i.e.,
\begin{equation}\label{eq:dyadictrans}
\eta^{(\infty)}_{L|a}(\alpha^{N'}_{\infty}(A)) = \alpha^{N}_{\infty}(\eta^{(N)}_{L|a}(\alpha^{N'}_{N}((A)), \qquad A \in \cW_{N',L}, \,a \in \Lambda_N,\, N \geq N',
\end{equation}
and $\eta^{(\infty)}_{L}: \bigcup_N \Lambda_N  \curvearrowright \cW_{\infty,L}$ defines the required action of the dyadic translation group.
 
Invoking the discrete Fourier transform \eqref{eq:Latft}, we obtain the important identity in momentum space, i.e. with respect to the inclusion of dual lattices $\Gamma_{N}\subset\Gamma_{N+1}$:
\begin{align}
\label{eq:oneparticlescalingfourier}
R^{N}_{N+1}(\hat{q}_{N},\hat{p}_{N})& = 2^{\frac{d}{2}}m_{0}(\vep_{N+1}\!\ \cdot \!\ )(\hat{q}_{N+1},\hat{p}_{N+1}),
\end{align}
where $m_{0}(\vep_{N+1}k)\!=\!2^{-\frac{d}{2}}\sum_{n\in\ZZ^{d}}h_{n}e^{-i\vep_{N+1} n k}$ (the discrete Fourier transform of the
coefficients $h_n$) and $\hat{q}_{N+1},\hat{p}_{N+1}$ are the periodic extension of $\hat{q}_{N},\hat{p}_{N}$ from $\Gamma_{N}$ to $\Gamma_{N+1}$. From \eqref{eq:oneparticlescalingfourier}, we easily compute the iterated maps,
 \begin{align}
 \label{eq:oneparticlescalingiterated}
 \hspace{-0.15cm}R^{N}_{N'}(\hat{q}_{N},\hat{p}_{N}) & = 2^{\frac{d(N'-N)}{2}}\!\!\prod^{N'-N}_{n=1}\!\!\!m_{0}(\vep_{N+n}\!\ \cdot \!\ )(\hat{q}_{N'},\hat{p}_{N'}).
 \end{align}
This expression already indicates the special role of the scaling function $\phi$ in the limit $N'\rightarrow\infty$
(if $r_{\max}$ is finite) because \cite[(6.2.2)]{DaubechiesTenLecturesOn}:
\begin{align}
\label{eq:scalingprod}
\lim_{N'\rightarrow\infty}\prod^{N'-N}_{n=1}\!\!\!m_{0}(\vep_{N+n}\!\ \cdot \!\ ) & = \hat \phi(\vep_N\cdot)= \vep^{-d/2} \hat{\phi}^{(\vep)}(2^{-N}\cdot) = \vep_N^{-d/2} \hat{\phi}^{(\vep_N)}_0
\end{align}
pointwise, which results from the (continuum) Fourier transform (normalized as in~\eqref{eq:minkft}) of the scaling equation \eqref{eq:unitscalingeq}:
\begin{align}
\label{eq:unitscalingeqft}
\hat{\phi}^{(\vep)}(2^{-N}k) & = m_{0}(\vep_{N+1}k)\hat{\phi}^{(\vep)}(2^{-(N+1)}k), & k\in\tfrac{\pi}{L}\ZZ^{d},
\end{align}
together with the normalization~\eqref{eq:scalingnorm}.
Thus, apart from an infinite field-strength renormalization accounting for the rescaling of the symplectic structure $\sigma_{N',L}$ with $N'$, the inductive limit map, $R^{N}_{\infty}:\fh_{N,L}\rightarrow\fh_{\infty,L} = \varinjlim_{N}\fh_{N,L}$,
is given by multiplication with the Fourier transform $\vep^{-\frac{d}{2}}_{N}\hat{\phi}^{(\vep_{N})}_{0}$ in momentum space.

\subsubsection{Wavelet bases.}
Before turning to a detailed description of the wavelet method, let us also comment on the use of the wavelet basis of $L^{2}(\TT_{L}^{d})$ constructed from the scaling function \eqref{eq:Haarwav}. In general, such a basis is obtained by rewriting the MRA \eqref{eq:mra} as a direct sum decomposition:
\begin{align}
\label{eq:directsumwav}
V_{0}\oplus V'_{0}\oplus ... \oplus V'_{N}\oplus V'_{N+1}\oplus ... & = L^{2}(\TT_{L}^{d}).
\end{align}
This is achieved by finding an orthogonal decomposition of the orthogonal complement of $V_{0}$ inside $V_{1}$:
\begin{align}
\label{eq:wavcomp}
V_{1} & = V_{0}\oplus V'_{0} = V_{0}\oplus\bigoplus^{2^{d}-1}_{m=1} V^{\prime, m}_{0},
\end{align}
such that  there are distinguished orthonormal functions -- the {\bf wavelets}, $\{^{m}\psi^{(\vep)}_{0}\}^{2^{d}-1}_{m=1}$, which together with their integer translates span the spaces $V^{\prime, m}_{0}$. The spaces $V^{\prime, m}_{N}$ are constructed by rescaling the wavelets, thus, these spaces inherit the scaling properties of the $V_{N}$. \\[0.1cm]
As the use of wavelets compared to that of the scaling function turns out to be important for the characterization of the inductive-limit algebra $\cW_{\infty,L}$ and the scaling limit $\omega^{(\infty)}_{\infty} = \omega^{(\infty)}_{L,\infty}$ \eqref{eq:projstaterg} (see Section \ref{sec:continuumlimit}), we also provide some details on the construction of wavelets.

Let us discuss wavelets in the unit lattice, $\vep=1$ for the case $d=1$.
Higher dimensional wavelets can be obtained by tensor product, cf.~\cite{DaubechiesTenLecturesOn}.

Given the scaling function, $\phi^{(\vep=1)}_{0} = \phi$, the wavelet $^{m=1}\psi^{(\vep=1)}_{0}=\psi$ can be constructed by the formula:
\begin{align*}
\psi(x) & = \sum_{n\in\ZZ}g_{n}\phi_{1,n}(x),
\end{align*}
where the wavelet coefficients $\{g_{n}\}_{n\in\ZZ}$ (also called a high-pass filter) are obtained explicitly in terms of the expansion coefficients $\{h_{n}\}_{n\in\ZZ}$ of the scaling equation \eqref{eq:scalingeq}. Albeit, the choice of wavelet coefficients is not unique, convenient options include:
\begin{align}
\label{eq:waveletcoeff}
g_{n} & = (-1)^{n}\overline{h_{1-n}}, & \textup{ or } & & g_{n} & = (-1)^{n}h_{1-n+2M}, & \forall n&\in\ZZ,
\end{align}
for some $M\in\ZZ$, and we take one of them.
By construction, the set of functions $\{\phi\}\cup\{\psi_{N,k}\}_{N\in\NN_{0},k\in\{-2^{N}r,...,2^{N}r-1\}}$,
where $\psi_{N,k}(x) = 2^{\frac{Nd}{2}}\psi(2^{N}x- k)$,
forms an orthonormal basis of $L^{2}(\TT_{L=r})$. 
The orthogonal decomposition $V_{N+1}=V_{N}\oplus V'_{N}$ analogous to \eqref{eq:wavcomp} is captured by an important projection identity involving the wavelets $\psi_{N,k}$ and scaling functions $\phi_{N,k}$:
\begin{align}
\label{eq:waveletproj}
\textup{Proj}_{V_{N+1}}(f) & = \sum_{k=-2^{N+1}r}^{2^{N+1}r-1}\langle\phi_{N+1,k},f\rangle_{L^{2}}\phi_{N+1,k} = \textup{Proj}_{V_{N}}(f) + \sum_{k=-2^Nr}^{2^Nr-1}\langle\psi_{N,k},f\rangle_{L^{2}}\psi_{N,k} \\ \nonumber
 & = \sum_{k=-2^Nr}^{2^Nr-1}\langle\phi_{N,k},f\rangle_{L^{2}}\phi_{N,k} + \sum_{k=-2^Nr}^{2^Nr-1}\langle\psi_{N,k},f\rangle_{L^{2}}\psi_{N,k},\hspace{1cm} f\in L^{2}(\TT_{L}).
\end{align}
Iterating and using $\lim_{N\rightarrow\infty}\textup{Proj}_{V_{N}} = \id_{L^{2}(\TT_{L})}$, we have an expansion:
\begin{align}
\label{eq:waveletexp}
f & = \sum_{k=-r}^{r-1}\langle\phi_{0,k},f\rangle_{L^{2}}\phi_{0,k} + \sum_{N\in\NN_{0}}\sum_{k=-2^Nr}^{2^Nr-1}\langle\psi_{N,k},f\rangle_{L^{2}}\psi_{N,k}.
\end{align}

\subsubsection{Block-spin renormalization in terms of orthogonal Haar wavelets}
\label{sec:Haarwavelet}
The above interpretation \eqref{eq:blockfield} and \eqref{eq:blockmom} of lattice fields as continuum field smeared with characteristic functions, $\{\vep_{N}^{-d}\chi_{x+[0,\vep_{N})^{d}}\}_{x\in\Lambda_{N}}$, can be understood as a special instance of the general scheme of wavelet scaling discussed above.
We take the characteristic function $\chi_{[0,\vep)^{d}}$ as the scaling function underlying the construction of the (periodic) Haar wavelets basis of $L^{2}(\TT_{L}^{d})$ with scaling parameters $\vep r = L$ \cite{BattleWaveletsAndRenormalization, MeyerWaveletsAndOperators, CohenWaveletsOnThe, DaubechiesTenLecturesOn}. More precisely, the function,
\begin{align}
\label{eq:Haarwav}
\phi^{(\vep)}_{x}(y) & = \vep^{-\frac{d}{2}}\chi_{[0,\vep)^{d}}(y-x) = \vep^{-\frac{d}{2}}\chi_{[0,1)^{d}}(\vep^{-1}(y-x)),
\end{align}
has an associated (half-sided) multiresolution analysis (MRA), i.e.~there is a sequence of subspaces of step functions
$\{V_{N}\}_{N\in\NN_{0}}$ satisfying~\eqref{eq:MRAunion}-\eqref{eq:MRAtranslate}.
In this case, we have $r_{\max} = 1$.

By dimensionality, we can identify the real subspace spanned by this basis $\{ \phi^{(\vep)}_{N,k}\}$ with the Lagrangian subspaces of $\fh_{N,L}$, cp. Section \ref{sec:freefield}. Moreover, Definition \ref{def:blockspinscaling} of the block-spin scaling map can be obtained from the expansion of $\phi^{(\vep)}_{0}$ into the functions $\{\phi^{(\vep)}_{1,k}\}$, which is equivalent to the inclusion $V_{0}\subset V_{1}$ and similarly, an embedding $\fh_{0,L}\rightarrow\fh_{1,L}$ is given by
writing a step function as a linear combination of finer step functions:
\begin{align}
\label{eq:Haarscale}
\phi^{(\vep)}_{0}(x) & = \sum_{k\in\NN_{0}^{d}, \|k\|_\infty \leq r_{\max}} 2^{-\frac{d}{2}} \phi^{(\vep)}_{1,\vep k}(x).
\end{align}
Now, \eqref{eq:Haarscale} can be interpreted as the well-known scaling equation,
\begin{align}
\label{eq:scalingeq}
\phi^{(\vep)}_{0}(x) & = \sum_{n\in\ZZ^{d}}h_{n}\phi^{(\vep)}_{1, \vep n}(x),
\end{align}
of the theory of wavelets
associated with the scaling function \eqref{eq:Haarwav} (and the MRA \eqref{eq:mra}).
The Haar wavelet basis in $L^{2}([0,1))$ ($d=1$) is usually given as:
\begin{align*}
\psi_{0,0}(x)  = (-1)\phi(2x-1)+(+1)\phi(2x) =
 \begin{cases}
   1 & x\in[0,\tfrac{1}{2}) \\
   -1 & x\in[\tfrac{1}{2},1),
 \end{cases}
\end{align*}
\begin{align*}
\psi_{N,k}(x) & = 2^{\frac{N}{2}}\psi_{0,0}(2^{N}x-k), & N\in\NN_{0}, &\ 0\leq k\leq 2^{N}-1.
\end{align*}

The relation between fields and momenta at successive scales \eqref{eq:blockscalefield} and
\eqref{eq:blockscalemom} in terms of the scaling function $\phi^{(\vep)}_{0}$ can be reformulated as follows:
\begin{align}
\label{eq:blockscalefieldre}
\Phi_{N}(x) & = \vep_{N}^{-\frac{d}{2}}\Phi(\phi^{(\vep_{N})}_{x}) = 2^{-\frac{d}{2}} \vep_{N+1}^{-\frac{d}{2}}\sum_{n\in\ZZ^{d}}h_{n}\Phi(\phi^{(\vep_{N})}_{1, 2x+\vep_{N}n}) \\ \nonumber
 & = 2^{-\frac{d}{2}}\sum_{n\in\ZZ^{d}}h_{n}\vep_{N+1}^{-\frac{d}{2}}\Phi(\phi^{(\vep_{N+1})}_{x+\vep_{N+1}n}) = 2^{-\frac{d}{2}}\sum_{n\in\ZZ^{d}}h_{n}\Phi_{N+1}(x+\vep_{N+1}n),\\
\label{eq:blockscalemomre}
\Pi_{N}(x) & = 2^{-\frac{d}{2}}\sum_{n\in\ZZ^{d}}h_{n}\Pi_{N+1}(x+\vep_{N+1}n).
\end{align}

\subsubsection{The Daubechies wavelets}\label{sec:Daubechies}
In our renormalization scheme of Section \ref{sec:continuum},
observables in the lattice field theories should be mapped to certain observables in the continuum theory.
In order to obtain a continuum field as the scaling limit of lattice fields,
we need to choose a scaling function $\phi$ which is both localized and sufficiently regular.
This goal is achieved by the so-called Daubechies wavelets \cite[Chapter 6]{DaubechiesTenLecturesOn}.

This is a family of scaling functions $\{_K\phi\}_{K \in \NN}$ with various support properties and regularity,
although no closed expression for them is known.
The scaling function $_K\phi$ satisfies the \eqref{eq:scalingeq}
with $\{h_n\}$ with $h_n = 0$ for $n \ge 2K$ \cite[Table 6.1]{DaubechiesTenLecturesOn},
is supported in $[0,2K-1]$ and belongs to the class $C^{\alpha - \epsilon}$ for arbitrary $\epsilon > 0$,
where the dependence of $\alpha$ on $K$ is given by \cite[Table in P.226]{DaubechiesTenLecturesOn}.
In particular, with $K=2$, $\alpha \cong 0.339$, and this suffices for our purpose to
construct both the field and momentum operators in the continuum.

In the following, we take $\phi = _K\!\!\phi$ and define $R_{N+1}^N$
as in \eqref{eq:waveletscale}.

\section{The continuum limit of the free vacua}\label{sec:continuumlimit}
On the field algebras $\{\cW_{N,L}\}$ on the lattices with the family of scaling maps $\{\alpha^{N}_{N'}\}_{N<N'}$ based on the Daubechies wavelets,
we construct a family of initial states $\{\omega^{(N)}_{L,0}\}_{N\in\NN_{0}}$, consisting of the initial data of Wilson's triangle of renormalization (see Figure \ref{fig:statetrianglerg}).
We restrict our attention to free fields, i.e.\! we choose as initial states a family of ground states of the free lattice Hamiltonian $H^{(N)}_{L,0}$.

As $\{\cW_{N,L}\}$ can be interpreted as the time-zero algebra of the lattice field $\Phi_N(x), \Pi_N(x)$ and $\delta$-functions on the lattice
$\Lambda_N$ are mapped to the scaling functions by \eqref{eq:scalingprod}, a simplest choice is to take the vacuum state in the continuum
$\omega_L$ and to restrict it to the image of $\{\cW_{N,L}\}$ by the map. It is straightforward to show that this yields
a time-zero net of local algebras (by using the properties of the continuum free field). This is, however, available just because
we know explicitly what the continuum state should look like.
In a constructive program, one should take a natural Hamiltonian at each scale, and consider the sequence of the ground states.
Below we show that this indeed gives the continuum free field state as well in the scaling limit, and hence gives rise to a continuum net.

\subsection{States on lattice and continuum fields}
\subsubsection{Ground states of lattice free fields}
\label{sec:latvac}
The (self-adjoint) free lattice Hamiltonian $H^{(N)}_{L,0}$ of (unrenormalized) mass $\mu_{N}>\sqrt{2d}$ is defined on a dense domain $\cD_{N,L}\subset\cH_{N,L} = \fF_+(\fh_{N,L})$ by the expression \cite[(1.8.17)]{BattleWaveletsAndRenormalization} (up to a reparametrization)  :
\begin{align}
\label{eq:freeham}
H^{(N)}_{L,0} & = \tfrac{1}{2}\vep_{N}^{d}\bigg(\sum_{x\in\Lambda_{N}}\Pi_{N|x}^{2}+\mu^{2}_{N}\vep_{N}^{-2}\Phi_{N|x}^{2}
-2\sum_{\underset{\text{adjacent}}{x,y \in\Lambda_{N}}}\vep_{N}^{-2}\Phi_{N|x}\Phi_{N|y}\bigg).
\end{align}

To be precise, we define the domain $\cD_{N,L} = \cD_{\mu_{N}}$ depending explicitly on $\mu_{N}$ via the dispersion relation, $\gamma_{\mu_{N}}^{2}(k) = \vep_{N}^{-2}(\mu_{N}^{2}-2d)+2\vep_{N}^{-2}\sum^{d}_{j=1}(1-\cos(\vep_{N}k_{j}))$, $k\in\Gamma_{N}$:
\begin{align}
\label{eq:freehamdom}
\cD_{\mu_{N}} & = \bigg\{\Psi\in\cH_{N,L}\ :\ \sum^{\infty}_{n=0}\bigg\|\sum^{n}_{j=1}\gamma_{\mu_{N}}((\ \cdot \ )_{j})\Psi_{n}\bigg\|_{\ltwo(\Gamma_{N})^{\otimes n}}^{2}<\infty\bigg\},
\end{align}
where $\gamma_{\mu_{N}}((\ \cdot \ )_{j})$ denotes the multiplication operator on the $j$-th tensor component.
The generators of the Weyl algebra $\cW_{N,L}$ are related to the field, $\Phi_{N}$, and momentum, $\Pi_{N}$, in the usual way (in the Fock representation on $\cH_{N,L}$): $W_{N}(\xi)=e^{i(\Phi_{N}(q)+\Pi_{N}(p))}$ with $\xi = \vep_{N}^{\frac{d+1}{2}}q + i\vep_{N}^{\frac{d-1}{2}}p$ as in Section \ref{sec:freefield}. The ground state (or lattice vacuum) $\Omega_{\mu_{N}}$ of $H^{(N)}_{L,0}$ gives the following state on $\cW_{N,L}$:
\begin{align}
\label{eq:freeground}
\omega_{\mu_{N}}(W_{N}(\xi)) & = e^{-\frac{1}{4}\left(\left\|\gamma_{\mu_{N}}^{-1/2}\hat{q}\right\|_{N,L}^{2}+\left\|\gamma_{\mu_{N}}^{1/2}\hat{p}\right\|_{N,L}^{2}\right)},
\end{align}
and the GNS construction applied to $\cW_{N,L}$ with respect to $\omega_{\mu_{N}}$ yields a representation which is unitarily equivalent to that on $\cH_{N,L}$ such that $\Omega_{\mu_{N}}$ is identified with the cyclic GNS vector.

Let us introduce the ($N$-dependent) physical mass $m$ by $\mu_{N}^{2} = \vep_{N}^{2}m^{2}+2d$.
The Fock representation of mass $m$ of the lattice scalar field $\cW_{N,L}$ is determined by the Fock vacuum state\footnote{We drop
the dependence of the state $\omega_{L,0}^{(N)}$ on $m$, because we do not change $m>0$ during the paper.}
\begin{align}
\label{eq:freegroundmass}
\omega_{L,0}^{(N)}(W_N(\xi)) = e^{-\frac{1}{2}\left\|\hat{\xi}^{(m)}\right\|^{2}_{N,L}}, \qquad \xi \in \fh_{N,L},
\end{align}
where
\begin{align*}
&\hat{\xi}^{(m)}(k)  = \tfrac{i}{\sqrt2} \left(\gamma^{(N)}_{m}(k)^{-\frac{1}{2}}\hat{q}(-k)+i\gamma^{(N)}_{m}(k)^{\frac{1}{2}}\hat{p}(-k)\right) 
\end{align*}
with $\gamma^{(N)}_{m}(k)^{2} = m^{2}+2\vep_{N}^{-2}\sum^{d}_{j=1}(1-\cos(\vep_{N}k_{j}))$, $k\in\Gamma_{N}$.
We actually have
\begin{align}
\label{eq:normxiN}
\left\|\hat{\xi}^{(m)}\right\|_{N,L}^{2} & = \tfrac{1} {2{(2r_N)}^{d}}\sum_{k\in\Gamma_N}\left(\gamma^{(N)}_{m}(k)^{-1}\hat{q}(-k)\hat{q}(k) + \gamma^{(N)}_{m}(k)\hat{p}(-k)\hat{p}(k)\right) \\ \nonumber
& = \tfrac{1}{2}\left(\left\|(\gamma^{(N)}_{m})^{-1/2}\hat{q} \right\|_{N,L}^{2}+ \left\|(\gamma^{(N)}_{m})^{1/2}\hat{p} \right\|_{N,L}^{2}\right),
\end{align}
and hence $\omega_{\mu_N} = \omega_{L,0}^{(N)}$, and their GNS representations are the same.
The expression \eqref{eq:normxiN} facilitates its relation to states on the continuum, see next Section.

The lattice ``mass'' $\mu_{N}$ is a dimensionless parameter and one can consider
its ``flow'' with respect to the action of renormalization group \eqref{eq:staterg} on states and Figure \ref{fig:statetrianglerg}, cf. \cite{WilsonTheRenormalizationGroupKondo}, in the following sense: At each scale $N$, we initially fix $\mu_{N} = \mu$ and follow the variation of the parameter $\mu_{N}(M) = \mu_{N+M}$ entering the renormalized states $\omega^{(N)}_{L,M}$ as a function of $M$.

Although, strictly speaking, the ``mass" parameter $\mu_{N+M}$ is not of the type as the initial ``mass'' $\mu_{N}$ because the scaling map $\alpha^{N}_{N+M}$ modifies the form of the dispersion relation and, thereby, the lattice interactions (if we were to interpret $\omega^{(N)}_{L,M}$ as ground state itself). This could be compensated by separating the part of $\omega^{(N)}_{L,M}$ that resembles the initial $\omega^{(N)}_{L,0}$, which would depend on $\mu_{N+M}$, in addition to $\prod^{M}_{n=1}m_{0}(\vep_{N+n})$ (see Section \ref{sec:groundstate}).

In order to obtain a convergent sequence of states as we scale $\vep_N$, the dependence of
$\mu_N$ on this UV cutoff must be fixed in such a way that it determines a finite physical mass $m$ in the infrared.

\subsubsection{Continuum free scalar field}\label{sec:contfree}

Let us next fix the notations about the free scalar field on the continuum, both on the cylinder and on Minkowski spacetime (i.e., for the finite and infinite volume cases respectively).
We pick a mass $m>0$.

\paragraph{On the cylinder.}

We will regard the torus ${ \TT_L^d}=[-L,L)^d$
as the time zero Cauchy surface $\TT_L^d \simeq \{0\} \times \TT_L^d$ of the cylinder spacetime $\RR\times\TT_L^d$ endowed with the Minkowski metric. 
The Fourier transform of a function $\xi$ on $\TT_L^d$ is defined as
\begin{align} 
\label{eq:cylft}
\hat\xi(k)= \int_{[-L,L)^d} \xi(x)e^{-ixk} dx, \qquad k \in \tfrac{\pi}{L}\ZZ^d.
\end{align}

The one-particle Hilbert space $\fh_L$ of the mass $m$ free scalar field on the cylinder is  the completion of $C^\infty(\TT_L^d, \CC)$ with the  scalar product defined as follows: let $\xi,\eta\in C^\infty(\TT_L^d, \CC)$ whose real and imaginary parts are denoted by $\xi=q_\xi+ip_\xi$ and $\eta=q_\eta+ip_\eta$,
then their scalar product is
\begin{align}\label{eq:spcyl}
\langle\xi,\eta\rangle_L=\tfrac1{(2L)^{d}}\!\!\!\!\!\!\sum_{k\in\frac\pi L\ZZ^d}\!\!\!\overline{\left(\!\gamma_m^{-1/2}(k)\hat q_\xi(k)+i\gamma_m^{1/2}(k)\hat p_\xi(k)\!\!\right)}\!\left(\!\gamma_m^{-1/2}(k)\hat q_\eta(k)+i\gamma_m^{1/2}(k)\hat p_\eta(k)\!\!\right),
\end{align}
where $\gamma_m(k)=(k^2+m^2)^{\frac12}$ is the continuum dispersion relation.
The complex structure (``multiplication by the imaginary unit'') on $\fh_L$ is defined as
\[
 q+ip\longmapsto -\gamma_m p+i \gamma_m^{-1}q ,
\]
where $\widehat{(\gamma_m\xi)}(k)=\gamma_m(k)\hat\xi (k)$. The associated symplectic form is
\begin{align}\label{eq:sympcont}
\sigma_L(\xi,\eta) := \Im\langle\xi,\eta\rangle_L =\Im\bigg(\;\int\limits_{[-L,L)^{d}}d^{d}x\, {\bar \xi(x)\eta(x)}\bigg).
\end{align}

The Fock representation $\pi_{L}$ of the Weyl algebra $\cW(\fh_L)$ is the GNS representation induced by the state
\begin{align}\label{eq:vaccyl}
\omega_{L}(W_{\mathrm{ct}}(\xi))=e^{-\frac14\|\xi\|_L^2},\quad\xi\in \fh_L.
\end{align}
It will be convenient in the following not to consider $\pi_{L}$ as acting on the Fock space built on $\fh_L$. Instead, we will realize it on the (mass independent) Fock space  $\fF_+ (L^2(\TT_L^d))\cong\fF_{+}(\ltwo(\tfrac{\pi}{L}\ZZ,(2L)^{-d}))$ by\footnote{Here $L^2$ denotes the space of square-summable functions and not the square of the length $L$.}
\begin{equation}\label{eq:pimL}
\pi_{L}(W_{\mathrm{ct}}(\xi))=e^{i\left[a\left(\gamma_m^{-1/2}q_\xi+i\gamma_m^{1/2}p_\xi\right)+ a^*\left(\gamma_m^{-1/2}q_\xi+i\gamma_m^{1/2}p_\xi\right)\right]}, \qquad \xi \in \fh_L,
\end{equation}
where $a$ and $a^*$ are the standard creation and annihilation operator on $\fF_+ (L^2(\TT_L^d))$. Indeed, it is easy to check that the Fock vacuum $\Omega_L \in \fF_+ (L^2(\TT_L^d))$ is cyclic for the linear span of the operators~\eqref{eq:pimL} and induces the state $\omega_L$.

Let $ S$ be an open subregion of $\TT_L^d$,
the (time-zero) local $C^*$-algebra associated with $ S$ is the  $C^*$-subalgebra $\cW_L( S)$ of $\cW(\fh_L)$ generated by the Weyl operators $W_{\mathrm{ct}}(\xi)$ with $\xi \in C^\infty(\TT_L^d,\mathbb{C})$ compactly supported in $ S$, and let $\cA_L( S) := \pi_{L}(\cW_L( S))''$ be
the (time-zero) local von Neumann algebra associated with $ S$ in the representation.

The free field dynamics satisfies the Klein-Gordon equation. It is realized on $\fh_L$ by a one-parameter group of unitaries $\{\tau_{L|t}\}_{t \in \RR}$:
\begin{align}\label{eq:contdynamics}
\tau_{L|t}(\hat{\xi}) = [\cos(t \gamma_m)+i \gamma_m^{-1}\sin(t\gamma_m)] \hat{q}_\xi + i[\cos(t\gamma_m)+i\gamma_m \sin(t\gamma_m)]\hat{p}_\xi, \qquad t \in \RR.
\end{align}
The unitaries $\{\tau_{L|t}\}$, leaving the symplectic form invariant, induce a one-parameter group $\{\eta_{L|t}\}$ of automorphisms of $\cW(\fh_L)$ for which $\omega_L$ is an invariant state, and which are therefore implemented  in the representation $\pi_{L}$ by unitaries $U_L(0,t)$, $t \in \RR$, on $\fF_+(L^2(\TT_L^d))$.
Moreover, as the Klein-Gordon equation has the speed of propagation $1$, we have $\Ad U_L(0,t)(\cA_L(\cB_r)) \subset \cA_L(\cB_{r+|t|})$, where $\cB_r$ is a ball of radius $r > 0$ and $r + |t| < L$. Similarly, the spacelike translations (on $\fh_L$) are unitarily implemented on $\fF_+(\fh_L)$ by $U_L(a,0)$, $a \in \RR^d$, and $U_L(a,0), U_L(0,t)$ commute.

Altogether, $(\cA_L, U_L, \Omega_L)$ form a continuum finite volume time-zero net of observables as per Section \ref{sec:continuum}.

\paragraph{On Minkowski space.} 
The construction is parallel to the cylinder case.
Our convention for the Fourier transform of a function $\xi : \RR^d \to \CC$ is
\begin{align}
\label{eq:minkft}
\hat\xi(k)= \int_{\RR^d} \xi(x)e^{-ixk} dx.
\end{align}
We note that with this convention and the one~\eqref{eq:cylft} for the torus, if $\xi$ is a function compactly supported in $(-L,L)^d$ then the value of
$\hat \xi(k)$ on $k\in \frac \pi L \ZZ^{d}$ is defined without ambiguity whether we consider $\xi$ on the torus or on $\RR^d$. This will be useful in the following.

The one-particle Hilbert space $\fh_\infty$ for the continuum free field on Minkowski space is  the completion of $C^\infty_0(\RR^d,\CC)$ with respect to the  scalar product  defined for $\xi,\eta\in C^\infty_0(\RR^d,\CC)$, (whose real and imaginary parts are $\xi=q_\xi+ip_\xi$ and $\eta=q_\eta+ip_\eta$) as
\begin{align}\label{eq:scalarmink}
\langle\xi,\eta\rangle_\infty =\tfrac{1}{(2\pi)^d}\!\!\!\int_{\RR^d}\!\!\!dk \overline{\left(\!\gamma_m^{-1/2}(k)\hat q_\xi(k)\!+\!i\gamma_m^{1/2}(k)\hat p_\xi(k)\!\right)}\!\left(\!\gamma_m^{-1/2}(k)\hat q_\eta(k)\!+\!i\gamma_m^{1/2}(k)\hat p_\eta(k)\!\!\right),
\end{align}
with the same dispersion relation as above.
Also $\fh_\infty$ becomes a complex Hilbert space if the complex structure is given by
\[
q+ip\longmapsto -\gamma_m p+i \gamma_m^{-1}q,
\]
and the associated symplectic form is
\begin{align}\label{eq:sympcontinf}
\sigma_\infty(\xi,\eta) := \Im\langle\xi,\eta\rangle_\infty =\Im\left(\;\int\limits_{\RR^{d}}d^{d}x\,{\bar \xi(x)\eta(x)}\right). 
\end{align}

The Fock representation $\pi_{\infty}$ of the Weyl algebra $\cW(\fh_\infty)$ is the GNS representation specified  by the state
\begin{align}\label{eq:vacuummink}
\omega_{\infty}(W_{\mathrm{ct}}(\xi))=e^{-\frac14\|\xi\|_{\infty}^2},\quad\xi\in \fh_\infty,
\end{align}
and it is realized on the Fock space over $L^2(\RR^d)$ (independent of $m$) by
\[
\pi_{\infty}(W_{\mathrm{ct}}(\xi))=e^{i\left[a\left(\gamma_m^{-1/2}q_\xi+i\gamma_m^{1/2}p_\xi\right)+ a^*\left(\gamma_m^{-1/2}q_\xi+i\gamma_m^{1/2}p_\xi\right)\right]}.
\]

The (time-zero) local $C^*$-algebra $\cW_\infty( S)$ and von Neumann algebra $\cA_\infty( S)$ associated with bounded open regions $ S \subset \RR^d$ are defined similarly to the torus case considered above,  and the dynamics is given by unitaries $\{\tau_{\infty|t}\}_{t \in \RR}$ on $\fh_\infty$, whose action on a generic $\xi \in \fh_\infty$ is defined again by formula~\eqref{eq:contdynamics} and which again induce automorphisms of $\cW(\fh_\infty)$ with finite propagation speed, implemented in $\pi_{\infty}$ by unitaries $U_\infty(0,t)$, $t \in \RR$. Thus, together with the unitary implementers $U_\infty(a,0)$, $a \in \RR^d$, of spatial translations, and with the Fock vacuum $\Omega_\infty \in \fF_+(L^2(\RR^d))$, we obtain the continuum infinite volume time-zero net of observables $(\cA_\infty,U_\infty,\Omega_\infty)$ defined by the free scalar field on Minkowski space.

We summarize our notations in Table \ref{tab:note} for one-particle spaces, algebras, Weyl operators and vacuum states on the lattices and for the inductive limit and the continuum theories.

\begin{table}[ht]
\centering
\begin{tabular}{|c|ccc|}
\hline
& lattice object & inductive limit & continuum object \\
& & &(possibly $L=\infty$) \\
\hline
 one-particle space  & $\fh_{N,L}$ & $\fh_{\infty,L} $ & $ \fh_L$ \\
                     &             & (symplectic space) & \\
 full algebra & $\cW_{N,L}$ & $\cW_{\infty,L}$ & $ \cW_L$ \\
 local algebra & $\cW_{N,L}( S)$ & $\cW_{\infty,L}( S) $ & $ \cW_L( S)$ \\
 field & $\Phi_N(x)$ & & $\Phi(x)$ \\
 momentum & $\Pi_N(x)$ & & $\Pi(x)$ \\
 Weyl operator & $W_N(\xi)$ & & $W_{\mathrm{ct}}(\xi)$ \\
 state  & $\omega^{(N)}_{L,M}$ & $\omega^{(\infty)}_{L,\infty}$ & $\omega_L$ \\
 (of mass $m$) & & & \\
 dispersion relation & $\gamma_m^{(N)}$ & & $\gamma_m$ \\
\hline
\end{tabular}
\caption{Notations for one-particle spaces, algebras, Weyl operators and vacuum states.}
\label{tab:note}
\end{table}

\subsection{Constructing the continuum limit}\label{subsec:cont}
In this section we will discuss the relation between the quasi-local and the local algebras of the inductive limit and of the continuum free field. 

\subsubsection{Embedding the lattice algebras into the continuum algebra}
\label{sec:embedding}

The lattice algebras in different scales with volume $(2L)^d$ are embedded into one another by $\alpha_{N+1}^{N}:\cW_{N,L}\rightarrow \cW_{N+1,L}$.
Here we show that they are further embedded into the continuum algebra $\cW(\fh_L)$.
This is realized by the $\RR$-linear map $R^N_\infty: \fh_{N,L} \to \fh_L$ defined by
\begin{equation}\label{eq:jLN}
 R^N_\infty(\xi)(y)=\vep_N^{\frac{d}{2}}\sum_{x\in\Lambda_N} (q(x)+i p(x))\phi^{(\vep_N)}_x(y),
\end{equation}
where we recall that $q:=\vep_N^{-\frac {1+d}2}\Re\xi$, $p:=\vep_N^{\frac{1-d}2}\Im \xi$.

\begin{proposition}\label{prop:isoalg1}
There exists an injective *-homomorphism $\beta_L : \cW_{\infty,L} \to \cW(\fh_L)$ such that
\[
 \beta_L(W_N(\xi)) = W_{\mathrm{ct}}(R^N_\infty(\xi)), \qquad \xi \in \fh_{N,L}.
\]
\end{proposition}
\begin{proof}
The map $R^N_\infty$ preserves the symplectic forms on $\fh_{N,L}$ and $\fh_L$, defined in Sections \ref{sec:freefield} and  \ref{sec:contfree} respectively, indeed:
\begin{align*}
\sigma_{L}(R^N_\infty(\xi),R^N_\infty(\eta)) & =\Im \bigg(\!\vep_N^{d}\!\!\!\!\!\!\sum_{x,x'\in\Lambda_N}\!\!\!\!\overline {(q_\xi(x)\!+\!ip_\xi(x)}(q_\eta(x')\!+\!ip_\eta(x')\!)\vep_N^{-d}\!\!\!\!\!\!\!\!\int\limits_{[-L,L)^{d}}\!\!\!\!\!\!d^{d}y\,\phi\!\left(\tfrac{y-x}{\vep_N}\right)\phi\!\left(\tfrac{y-x'}{\vep_N}\right)\!\!\!\!\bigg)\\
&=\vep_N^d\!\!\!\sum_{x\in\Lambda_N}\!\!\big(q_\xi(x)p_\eta(x)-p_\xi(x)q_\eta(x)\big)=\sigma_{N,L}(\xi,\eta),
\end{align*}
where the second equality holds since $\langle \phi^{(\vep_N)}_x,\phi^{(\vep_N)}_{x'}\rangle_{L}=\delta_{xx'}$ by the property of the scaling function $\phi$.
Therefore by the uniqueness of the Weyl $C^*$-algebra, $R^N_\infty$ induces an injective *-homomorphism of $C^*$-algebras $\beta^N_L : \cW_{N,L} \to \cW(\fh_L)$ such that $\beta^N_L(W_N(\xi)) = W_{\mathrm{ct}}(R^N_\infty(\xi))$. Moreover, there holds
\begin{align*}
R^{N+1}_\infty(R_{N+1}^{N}(\xi)) &= \vep_{N+1}^{\frac{d}{2}}\!\!\!\!\!\!\sum_{x\in\Lambda_{N+1}}\!\!\!\!(q_{R^{N}_{N+1}(\xi)}(x)+i p_{R_{N+1}^{N}(\xi)}(x))\phi^{(\vep_{N+1})}_x \\
&= \vep_{N+1}^{\frac{d}{2}}2^{\frac d 2}\!\!\!\sum_{y\in\Lambda_{N}}\!\!(q(y)+ip(y)\!)\!\!\!\sum_{n \in \ZZ^d}\!\!\!h_n\!\!\!\sum_{x\in \Lambda_{N+1}}\!\!\!\delta_{y+\frac{\vep_N n}2}^{(\vep_{N+1})}(x)\phi^{(\vep_{N+1})}_x\\
&=\vep_N^{\frac{d}{2}}\!\!\!\sum_{y\in\Lambda_{N}}\!\!(q(y)+ip(y)\!)\!\!\!\sum_{n \in \ZZ^d}\!\!h_n\phi^{(\vep_N/2)}_{y+\frac{\vep_N n}2} =\vep_N^{\frac{d}{2}}\!\!\!\sum_{y\in\Lambda_{N}}\!(q(y)+ip(y)\!)\phi^{(\vep_N)}_{y} = R^N_\infty(\xi),
\end{align*}
where in the fourth equality we used~\eqref{eq:scalingeq}. This entails $\beta^{N+1}_L\circ\alpha^{N}_{N+1} = \beta^N_L$, and therefore, by the uniqueness and simplicity of the $C^*$-inductive limit, $\beta_L$ exists as in the statement.\hfill$\square$
\end{proof}

Now that we have realized the lattice algebras $\cW_{N,L}$ inside $\cW(\fh_L)$, we can consider the restrictions, $\omega_{L|\beta_{L}(\cW_{N,L})}$, of the vacuum state $\omega_L$ and take these as the ``initial states''.
This yields in a straightforward way a continuum time-zero net of local algebras.
However, as already remarked, this choice is possible only because we know explicitly the vacuum state of the continuum free field.
Below we study a more interesting case where we take the ground state \eqref{eq:freeground} of the lattice free field at each scale, $\omega^{(N)}_{L,0} = \omega_{\mu_{N}}$, and show that the renormalization group scheme of Section \ref{sec:oaren} constructs $\omega_{L}$ as the scaling limit.

\subsubsection{Scaling limits of free-field ground states}
\label{sec:groundstate}
Following the general scheme of Section \ref{sec:oaren}, the scaling limit of \eqref{eq:freeground} is obtained by a two-step procedure, which we will explicitly implement:
\begin{itemize}
	\item[1.] Prove the convergence of $\lim_{M\rightarrow\infty}\omega^{(N)}_{L,M} = \omega^{(N)}_{L,\infty}$ on $\cW_{N,L}$ at every fixed scale $N\in\NN_{0}$.
	\item[2.] Construct the projective-limit state $\varprojlim_{N}\omega^{(N)}_{L,\infty} = \omega^{(\infty)}_{L,\infty}$ on $\cW_{\infty,L}$.
\end{itemize}
According to Figure \ref{fig:statetrianglerg}, the sequence of states $\{\omega^{(N)}_{L,M}\}_{M\in\NN_{0}}$ at each level $N$ is generated by iterating the \textit{flow equation} \eqref{eq:staterg}. In the case at hand, we use \eqref{eq:oneparticlescalingiterated} and \eqref{eq:oneparticlerg} to define the renormalization group, $\{\alpha^{N}_{N'}\}_{N<N'}$, resulting in:
\begin{align}
\label{eq:freegroundscaling}
\omega^{(N)}_{L,M}(W_{N}(\xi)) & = \omega_{\mu_{N+M},M}(W_{N}(\xi)) \\ \nonumber
& = e^{-\frac{2^{dM}}{4}\Big(\Big\|\gamma_{\mu_{N+M}}^{-1/2}\prod\limits^{M}_{n=1}\!\!\!m_{0}(\vep_{N+n}\!\ \cdot \!\ )\hat{q}_{N+M}\Big\|_{N+M,L}^{2}+\Big\|\gamma_{\mu_{N+M}}^{1/2}\prod\limits^{M}_{n=1}\!\!\!m_{0}(\vep_{N+n}\!\ \cdot \!\ )\hat{p}_{N+M}\Big\|_{N+M,L}^{2}\Big)},
\end{align}
for any $\xi\in\fh_{N,L}$, where we made the scale dependence of one-particle vectors explicit and denoted by $\hat{q}_{N+M},\hat{p}_{N+M}$ the periodic extensions of $\hat{q}_{N},\hat{p}_{N}$ to $\Gamma_{N+M}$. Recall that $m_{0}$ is the trigonometric polynomial associated with the low-pass filter $\{h_{n}\}$ introduced in \eqref{eq:oneparticlescalingfourier}.

Thus, if we choose $\{\mu_N\}$ which satisfies the \textit{renormalization condition},
\begin{align}
\label{eq:Latmassscaling}
\lim_{N\rightarrow\infty}\vep_{N}^{-2}(\mu_{N}^{2}-2d) & = m^{2},
\end{align}
for a fixed ``physical'' mass $m>0$, we find the free continuum dispersion relation of mass $m$ in the scaling limit: $\lim_{M\rightarrow\infty}\gamma_{\mu_{N+M}}(k)^{2} = m^{2}+k^{2} =\gamma_{m}(k)^{2}$.
By \eqref{eq:scalingprod}, the limit state should be formally
\begin{align}
\label{eq:freegroundlimit}
\omega^{(N)}_{L,\infty}(W_N(\xi)) &\!=\!e^{-\frac{1}{4}\big\|\hat{\phi}^{(\vep_{N})}_{0}(\gamma_{m}^{-1/2}\hat{q}_{\infty}+i\gamma_{m}^{1/2}\hat{p}_{\infty})\big\|_{2,L}^{2}},
\end{align}
where $\hat{q}_{\infty},\hat{p}_{\infty}$ are the periodic extensions of $\hat{q}_{N},\hat{p}_{N}$ to the infinite lattice $\Gamma_{\infty}=\tfrac{\pi}{L}\ZZ^{d}$, and $\|\!\ .\!\ \|_{2,L}$ is the standard norm on $\ltwo(\tfrac{\pi}{L}\ZZ^{d},(2L)^{-d}\mu_{\Gamma_{\infty}})$. The norms involved there are finite if the scaling function decays like $|\hat{\phi}(k)|\leq C(1+|k|)^{-\rho}$ with $\rho>\tfrac{d+1}{2}$.
Our choice of the Daubechies wavelet family, $\{_{K}\phi\}_{K\in\NN}$ \cite{DaubechiesTenLecturesOn} (see Section \ref{sec:Daubechies}), is indeed sufficient if
$\phi =\!\!\!\ _{2}\phi^{\otimes d}$ because $_{2}\rho\approx 1.339$. Moreover, \eqref{eq:freegroundlimit} is not only the formal limit of the sequence \eqref{eq:freegroundscaling} but a weak*-limit point as the following lemma shows.

\begin{lemma}
\label{lem:freegroundconv}
Assume that the renormalization condition \eqref{eq:Latmassscaling} holds and that $\phi$ is built via tensor products from the Daubechies family, $\phi = _K\!\!\phi^{\otimes d}$ for $K\ge 2$. Then, the sequence of states $\{\omega^{(N)}_{L,M}\}_{M\in\NN_{0}}$ on $\cW_{N,L}$ is weak* convergent to the state~\eqref{eq:freegroundlimit} for every $N\in\NN_{0}$.
\end{lemma}
\begin{proof}
We state the proof for $\phi = _K\!\!\phi^{\otimes d}$ with $K\ge 2$.

Since the Weyl elements $W_{N}(\xi)$, $\xi\in\fh_{N,L}$, form a total set in $\cW_{N,L}$ in norm, it is sufficient to prove the convergence of \eqref{eq:freegroundscaling} to \eqref{eq:freegroundlimit}. 
To this end, we rewrite the $\|\cdot\|_{N+M,L}$ norms in \eqref{eq:freegroundscaling} in terms of the norm $\left|\left|\!\ .\!\ \right|\right|_{2,L}$ to be able to apply Lebesgue's dominated convergence theorem for the measure space $(\tfrac{\pi}{L}\ZZ^{d}, (2L)^{-d}\mu_{\frac{\pi}{L}\ZZ^{d}})$:
\begin{align*}
& 2^{dM}\Big\|\gamma_{\mu_{N+M}}^{\frac{1}{2}}\!\!\prod\limits^{M}_{n=1}\!m_{0}(\vep_{N+n}\!\ \cdot \!\ ) \hat p_{N+M}\Big\|_{N+M,L}^{2} \\
& =\tfrac{\vep_{N}^{d}}{(2L)^{d}}\!\!\!\!\sum_{k\in\frac{\pi}{L}\ZZ^{d}}\!\!\!\!\big|\chi_{\Gamma_{N+M}}(k)\gamma_{\mu_{N+M}}(k)^{\frac{1}{2}}\!\!\prod\limits^{M}_{n=1}\!m_{0}(\vep_{N+n}k) \hat p_{N+M}(k)\big|^2,
\end{align*}
where $\chi_{S}(k)$ is the characteristic function of the set $S$,
and similarly for the term involving $\hat q_{N+M}$. Now, the periodic extensions $\hat{q}_{N+M},\hat{p}_{N+M}$ of $\hat{q}_{N},\hat{p}_{N}$  are bounded in $k\in\Gamma_{N+M}\subset\Gamma_{\infty} = \tfrac{\pi}{L}\ZZ^{d}$, uniformly for $M \in \NN$, and the same holds for $\gamma_{\mu_{N+M}}^{-1/2}$ thanks to the renormalization condition~\eqref{eq:Latmassscaling}. As we know that
\begin{equation}\label{eq:sobolevconv}
\bigg|\chi_{\Gamma_{N+M}}(k)\gamma_{\mu_{N+M}}(k)^{\alpha}\prod\limits^{M}_{n=1}\!m_{0}(\vep_{N+n}k) \bigg|^2, \qquad -\tfrac12\leq\alpha\leq\tfrac{1}{2}
\end{equation}
converge pointwise to $|\hat\phi_0^{(\vep_N)}(k)\gamma_m(k)^\alpha|^2$ by \eqref{eq:scalingprod}, we only have to show that this sequence is 
bounded by a $\mu_{\frac{\pi}{L}\ZZ^{d}}$-integrable function uniformly for $M \in \NN$.
Next, we observe that because of the renormalization condition \eqref{eq:Latmassscaling}, we can replace
\begin{align*}
\gamma_{\mu_{N+M}}(k)^{\alpha} & = \bigg(\!\!\vep_{N+M}^{-2}(\mu_{N+M}^{2}-2d)+2\vep_{N+M}^{-2}\sum^{d}_{j=1}(1-\cos(\vep_{N+M}k_{j}))\!\!\bigg)^{\frac{\alpha}{2}}
\end{align*}
by
\begin{align*}
\bigg(\!\!m^{2}+2^{2M}\vep_{N}^{-2}2\sum^{d}_{j=1}(1-\cos(\tfrac{\vep_{N}k_{j}}{2^{M}}))\!\!\bigg)^{\frac{\alpha}{2}} & = \bigg(\!\!m^{2}+\sum^{d}_{j=1}(2^{M+1}\vep_{N}^{-1}\sin(\tfrac{\vep_{N} k_{j}}{2^{M+1}}))^{2}\!\!\bigg)^{\frac{\alpha}{2}}.
\end{align*}
Since we assume a tensor product structure for the scaling function, $\phi = _K\!\!\phi^{\otimes d}$, the same is true for the function $m_{0}$, i.e.~$m_{0}= _K\!\!m_{0}^{\otimes d}$. Similarly, the characteristic function $\chi_{\Gamma_{N}}$ factorizes: $\chi_{\Gamma_{N}}(k) = \prod_{j=1}^{d}\chi_{\Gamma^{(1)}_{N}}(k_{j})$ for $k\in\tfrac{\pi}{L}\ZZ^{d}$ with $\Gamma^{(1)}_{N}$ the $d=1$ analogue of $\Gamma_{N}$. Therefore, we can use the estimate
\begin{align*}
\bigg(\!\!m^{2}+\sum^{d}_{j=1}(2^{M+1}\vep_{N}^{-1}\sin(\tfrac{\vep_{N} k_{j}}{2^{M+1}}))^{2}\!\!\bigg)^{\frac{\alpha}{2}} & \leq \tfrac{\max(1,(\vep_{N}m)^{\alpha})}{\vep_{N}^{\alpha}}\!\!\prod_{j=1}^{d}\!\!\left(1+(2^{M+1}\sin(\tfrac{\vep_{N} k_{j}}{2^{M+1}}))^{2}\right)^{\frac{\alpha}{2}}
\end{align*}
to reduce the problem to the $d=1$ case. Finally, to conclude the estimation of~\eqref{eq:sobolevconv}, we only need to find a $\mu_{\frac{\pi}{L}\ZZ}$-integrable function that dominates the sequence:
\begin{align*}
\left|_{K}g^{(N)}_{M}(k)\right|^{2} & = \bigg|\chi_{\Gamma^{(1)}_{N+M}}(k)\bigg(1+(2^{M+1}\sin(\tfrac{\vep_N k}{2^{M+1}}))^{2}\bigg)^{\frac{\alpha}{2}}\prod\limits^{M}_{n=1}\!_{K}m_{0}(2^{-n}\vep_{N}k)\bigg|^2.
\end{align*}
In the following, we set $l=\vep_{N}k\in\pi\ZZ$ to write:
\begin{align}
\label{eq:sobolevseq}
\left|_{K}g^{(N)}_{M}(l)\right|^{2} & = \bigg|\chi_{\vep_N \Gamma^{(1)}_{N+M}}(l)\bigg(1+(2^{M+1}\sin(\tfrac{l}{2^{M+1}}))^{2}\bigg)^{\frac{\alpha}{2}}\prod\limits^{M}_{n=1}\!_{K}m_{0}(2^{-n}l)\bigg|^2.
\end{align}
Now, we observe that $_{K}m_{0}$ factorizes according to \cite[(6.1.2)]{DaubechiesTenLecturesOn},
\begin{align*}
_{K}m_{0}(l) & = \left(\frac{1+e^{-il}}{2}\right)^{K}\ \!_{K}\cL(l),
\end{align*}
where $_K\cL$ is a certain trigonometric polynomial, which leads to:
\begin{align*}
\left|_{K}g^{(N)}_{M}(l)\right|^{2} & =\chi_{\vep_N \Gamma^{(1)}_{N+M}}(l)\left(1+(2^{M+1}\sin(\tfrac{l}{2^{M+1}}))^{2}\right)^{\alpha}\prod\limits^{M}_{n=1}|\cos(2^{-n}\tfrac{1}{2}l)|^{2K}|\!_{K}\cL(2^{-n}l)|^{2}.
\end{align*}
Using the analogue of Vi\`ete's (respectively Euler's) formula for finite products,
\begin{align*}
\prod\limits^{M}_{n=1}\cos(2^{-n}\tfrac{1}{2}l) & = \frac{2\sin(\frac{1}{2}l)}{2^{M+1}\sin(\tfrac{l}{2^{M+1}})},
\end{align*}
(which is a consequence of iterating the cosine product identity, $\cos(\varphi)\cos(\vartheta) = \tfrac{1}{2}(\cos(\varphi+\vartheta)+\cos(\varphi-\vartheta))$, followed by a finite geometric summation\footnote{See \url{https://en.wikipedia.org/wiki/Vi\%C3\%A8te\%27s_formula}.}), we find:
\begin{align*}
\left|_{K}g^{(N)}_{M}(l)\right|^{2} & = \chi_{\vep_N \Gamma^{(1)}_{N+M}}(l)\left|\frac{2\sin(\frac{1}{2}l)}{2^{M+1}\sin(\tfrac{l}{2^{M+1}})}\right|^{2K}\left(1+(2^{M+1}\sin(\tfrac{l}{2^{M+1}}))^{2}\right)^{\alpha}\prod\limits^{M}_{n=1}|\!_{K}\cL(2^{-n}l)|^{2}.
\end{align*}
Next, we apply the basic inequalities,
\begin{align*}
2^{M+1}\sin(2^{-(M+1)}|l|) & \leq |l|, & \forall l\in\RR, &  \\
\tfrac{2}{\pi}|l| & \leq 2^{M+1}\sin(2^{-(M+1)}|l|), & \forall l\in\RR: |l| & \leq 2^{M}\pi,
\end{align*}
to arrive at:
\begin{align*}
\left|_{K}g^{(N)}_{M}(l)\right|^{2} & \leq \chi_{\vep_N \Gamma^{(1)}_{N+M}}(l)\pi^{2K}\left(\frac{|\sin(\frac{1}{2}l)|}{|l|}\right)^{2K}\left(1+|l|^{2}\right)^{\alpha}\prod\limits^{M}_{n=1}|\!_{K}\cL(2^{-n}l)|^{2} \\
 & \leq \chi_{\vep_N \Gamma^{(1)}_{N+M}}(l)\pi^{2K}\left(\frac{|\sin(\frac{1}{2}l)|}{|l|}\right)^{2K}\left(1+|l|\right)^{2\alpha}\prod\limits^{M}_{n=1}|\!_{K}\cL(2^{-n}l)|^{2},
\end{align*}
where in the first inequality we used that if $l \in \vep_N \Gamma^{(1)}_{N+M}$ then $2^{-(M+1)}|l| \leq \pi/2$.
Thus, to dominate the sequence \eqref{eq:sobolevseq} by $\mu_{\ZZ}$-integrable function, it is sufficient to have an estimate of the form,
\begin{equation}
\label{eq:regularityestimate}
\prod\limits^{M}_{n=1}|\!_{K}\cL(2^{-n}l)|  \leq C(1+|l|)^{K-1-\epsilon},
\end{equation}
for some $C>0$ and arbitrary $\epsilon>0$, as this implies:
\begin{equation*}
\left|_{K}g^{(N)}_{M}(l)\right|^{2}  \leq C^{2}\pi^{2K}\left(\frac{|\sin(\frac{1}{2}l)|}{|l|}\right)^{2K}\left(1+|l|\right)^{2(K+(\alpha-\epsilon)-1)}.
\end{equation*} 
But the estimate~\eqref{eq:regularityestimate} is an immediate consequence of the regularity analysis for compactly supported wavelets in \cite[Section 7.1]{DaubechiesTenLecturesOn}. To see this, we observe that according to \cite[Lemma 7.1.4]{DaubechiesTenLecturesOn}:
\begin{equation}\label{eq:Lbound}
\sup_{l}|_{K}\cL(l)| < 2^{K-1}.
\end{equation}
Now, applying the reasoning of \cite[Lemma 7.1.1]{DaubechiesTenLecturesOn} to the finite product $\prod\limits^{M}_{n=1}|\!_{K}\cL(2^{-n}l)|$ for the momentum space shells $|l|\leq1$ and $2^{n-1}<|l|\leq 2^{n}$, $n\in\NN_{\leq M}$, as well as $2^{M}<|l|\leq 2^{M}\pi$, we deduce that
~\eqref{eq:Lbound}
implies \eqref{eq:regularityestimate} for $\epsilon >0$ sufficiently small.\hfill$\square$
\end{proof}
Actually, with the embedding $\beta_L: \cW_{\infty, L} \to \cW(\fh_L)$,
the expression \eqref{eq:freegroundlimit} is exactly the evaluation of the usual continuum vacuum state $\omega_{L}$ of mass $m$ \eqref{eq:vaccyl} on the continuum Weyl operators $W_{\mathrm{ct}}(R^N_\infty(\xi))$. 
\begin{lemma}\label{lem:isoalg2}
It holds that $\omega^{(N)}_{L,\infty}(W) = \omega_{L}(\beta_L(W))$ for $W \in \cW_{N,L}$.
\end{lemma}
\begin{proof}
The one-particle symplectic map $R^N_\infty$, which satisfies $\beta_L(W_N(\xi)) = W_{\mathrm{ct}}(R^N_\infty(\xi))$,
intertwines the states: ${\omega_L(W_{\mathrm{ct}}(R^N_\infty(\xi))}=\omega^{(N)}_{L,\infty}(W_N(\xi))$, $\xi\in \fh_{N,L}$.
Indeed, using~\eqref{eq:spcyl}, \eqref{eq:vaccyl} and \eqref{eq:freegroundlimit}, 
\begin{align*}
\omega_L&(W_{\mathrm{ct}}(R^N_\infty(\xi))=\exp\bigg\{\!\!\!-\tfrac{1}{2(2L)^d}\!\!\!\sum_{k\in\frac\pi L\ZZ^{d}}\!\!\!\big|\gamma_m^{-\frac12}(k)(\Re R^N_\infty(\xi))\hat{\,} (k)+i\gamma_m^{\frac12}(k) (\Im R^N_\infty(\xi))\hat{\,} (k)\big|^2\bigg\}\\
&=\exp\bigg\{\!\!\!-\tfrac1{2(2 r\vep_N)^{d}}\!\!\!\sum_{k\in\frac\pi L\ZZ^{d}}\!\!\!\bigg(\frac{1}{\gamma_m(k)}\bigg|\vep_N^{d/2}\sum_{x\in\Lambda_N}q(x)\!\!\!\!\!\!\int\limits_{[-L,L)^{d}}\!\!\!\!\!\!d^{d}y\,\phi_x^{(\vep_N)}(y)e^{-iky}\bigg|^2 \\
&\hspace{4cm}+\gamma_m(k)\bigg|\vep_N^{d/2}\sum_{x\in\Lambda_N}p(x)\!\!\!\!\!\!\int\limits_{[-L,L)^{d}}\!\!\!\!\!\!d^{d}y\,\phi_x^{(\vep_N)}(y)e^{-iky}\bigg|^2\bigg)\bigg\}\\
&=\exp\bigg\{\!\!\!-\tfrac1{2(2 r\vep_N)^{d}}\!\!\!\sum_{k\in\frac\pi L\ZZ^{d}}\!\!\!\bigg(\frac{\vep_N^{2d}}{\gamma_m(k)}\bigg|\sum_{x\in\Lambda_N}\!\!\!q(x)e^{-ikx}\bigg|^2\!\!\!\!+\vep_N^{2d}\gamma_m(k)\bigg|\sum_{x\in\Lambda_N}\!\!\!p(x)e^{-ikx}\bigg|^2\bigg)|\hat\phi(\vep_N k)|^2\bigg\}\\
&=\omega_{L,\infty}^{(N)}(W_N(\xi))
\end{align*}
where in the third equality the relation
\[
 \int\limits_{[-L,L)^{d}}d^{d}y\,\phi_x^{(\vep_N)}(y)e^{-iky}=\vep_N^{\frac{d}{2}} \,e^{-ikx} \hat\phi(\vep_N k),
\]
for $k \in \frac \pi L \ZZ$ was used, where $\hat{\phi}(\vep_N k) = \int_{[-L,L)^{d}}d^{d}y\,\phi(y)e^{-i\vep_N ky}$. \hfill$\square$
\end{proof}

The projective consistency \eqref{eq:staterglimit} of the limit states $\omega^{(N)}_{L,\infty}$ follows both abstractly from Proposition \ref{prop:weakstarconv}
and directly by applying \eqref{eq:unitscalingeqft}. Therefore, we conclude that the limit state $\omega^{(\infty)}_{L,\infty}$ exists on $\cW_{\infty, L}$. By the preceding Lemma, $\omega^{(\infty)}_{L,\infty}$ agrees with the usual continuum vacuum state $\omega_{L}$ on the span of Weyl operators indexed by functions $f\in L^{2}(\TT_{L}^{d})$ with finite wavelet expansion \eqref{eq:waveletexp} associated with $\phi$.

\begin{remark}
\label{rem:dimension}
The construction of suitable scaling functions $\phi$ in dimension $d\geq2$ by tensor products is probably not optimal:
what is needed is the condition $|\hat{\phi}(k)|\leq C(1+|k|)^{-\rho}$ with $\rho>\tfrac{d+1}{2}$ which is rotationally invariant,
while our wavelets are regular in $d$ directions separately.
\end{remark}

To summarize:
\begin{proposition}
\label{prop:waveletscaling}
Let $\{\mu_{N}\}_{N}$ satisfy \eqref{eq:Latmassscaling} with $m>0$ and $\{\alpha_{N'}^N\}$ be the scaling map defined through the Daubechies scaling function $_K\phi$ with $K\ge 2$ on the scalar lattice algebras $\{\cW_{N,L}\}$ as in Section \ref{sec:waveletscaling}.
Then the sequences of states $\{\omega_ {L,M}^{(N)}\}_{M\in\NN_{0}} = \{\omega_{\mu_{N+M},M}\}_{M\in\NN_{0}}$ has the scaling limit $\omega^{(\infty)}_{L,\infty}$ on $\cW_{\infty, L}$ in the sense of \eqref{eq:projstaterg}.
With the embedding $\beta_L$ of $\cW_{\infty, L}$ in $\cW(\fh_L)$ by Proposition \ref{prop:isoalg1}, it holds that
$\omega^{(\infty)}_{L,\infty}(W_N(\xi)) = \omega_{L}(\beta_L(W_N(\xi))$ for $\xi \in \fh_{N,L}$ with some $N$.
\end{proposition}

\subsubsection{Density of the wavelet observables in the continuum limit}
\label{sec:density}

Let $\pi^{(\infty)}_{L,\infty}$, be the GNS representation of $\cW_{\infty,L}$ w.r.t.~$\omega_{L,\infty}^{(\infty)}$. Here we show that the local $C^*$-algebra $\cW_{\infty,L}(S)$, defined as in~\eqref{eq:spalocalg}, is dense in $\cW_L(S)$ in the strong operator topology in the representation $\pi^{(\infty)}_{L,\infty}$,
if we use a sufficiently regular scaling function (with $K\ge 6$).
Recall that $\phi^{(\vep_N)}_x \in L^2(\TT_L^d)$, $x \in \Lambda_N$, is the $2L$-periodization of the function $y \in \RR^{d} \mapsto \vep_N^{-\frac{d}{2}}\phi(\frac{y-x}{\vep_N})$ restricted to an open cube of side length $2L$ containing its support.

\begin{lemma}\label{lem:density}
Let $K\ge 6$. Then the (complex) linear span of the functions $\phi^{(\vep_N)}_x$, $x \in \Lambda_N$, $N \in \NN$, is dense in $\fh_L$.
\end{lemma}

\begin{proof}
As the norms are equivalent, $\fh_L$ can be identified, as a real Hilbert space, with the direct sum of real Sobolev spaces $H_\RR^{-\frac{1}{2}}(\TT_L^d)\oplus H_\RR^{\frac{1}{2}}(\TT_L^d)$. With the notation of Section \ref{sec:waveletscaling},
we then recall that if $f \in H_\RR^{\pm\frac12}(\RR^d)$
 then for any $\ell_0 \geq 0$,
\begin{align}\label{eq:sobolevexp}
f = \sum_{n\in\ZZ^d}\langle\phi_{\ell_0,n},f\rangle\phi_{\ell_0,n} + \sum_{\ell\geq \ell_0}\sum_{n\in\ZZ^d}\langle\psi_{\ell,n},f\rangle \psi_{\ell,n}
\end{align}
where $\langle \cdot, \cdot \rangle$ denotes the standard $L^2(\RR^d)$ inner product, and this series converges also in $H^{\pm\frac12}_\RR$
 for sufficiently regular wavelets. Convergence in $H_\RR^{-\frac12}(\RR^d)$ comes from  convergence in $L^2(\RR^d)$ and convergence in $H_\RR^{\frac12}(\RR^d)$  comes with $K\ge 6$ by \cite[Corollary 9.1, Example 9.1]{WASA} since the $B^{\frac12}_{22}$ Besov norm  is equivalent to the $H_\RR^{\frac12}(\RR^d)$ norm.

 Let $f\in H_\RR^{\pm\frac{1}{2}}(\RR^{d})$ have support contained in $(-L+\delta,L-\delta)^{d}$ for some $\delta > 0$.
Then, since $\supp \phi_{\ell,n} = 2^{-\ell}(\supp \phi +n)$, and analogously for $\psi_{\ell,n}$, we can find a sufficiently large $\ell_0$
such that the non-zero components in the expansion \eqref{eq:sobolevexp} of the rescaled function $f_\vep = f(\vep \cdot)$ all have supports in $(-\vep^{-1}(L-\delta/2),\vep^{-1}(L-\delta/2))^{d}$.
Moreover, there holds the further expansion (actually a finite linear combination)
\[
\psi_{\ell,n} = \sum_{k \in \ZZ^d} \langle \phi_{\ell+1,k},\psi_{\ell,n}\rangle \phi_{\ell+1,k},
\]
and again due to the support properties of $\psi_{\ell,n}$ the only non-zero coefficients involve functions $\phi_{\ell+1,m}$ whose support is contained in $(-\vep^{-1}L,\vep^{-1}L)^d$. Summing up, we obtain that $f_\vep$ can be approximated, in the $H^{\pm \frac 1 2}_{\RR}(\RR^d)$ norm, by finite linear combinations of functions  $\phi_{j,n}$ with support contained in $(-\vep^{-1}L,\vep^{-1}L)^d$. Recalling now that $\phi^{(\vep_N)}_x$ is the $2L$-periodization of $\vep^{-d/2}\phi_{N,\vep_N^{-1} x}(\vep^{-1}\cdot)$, we conclude that $f$ can be approximated in $H_\RR^{\pm 1/2}(\TT_L^d)$ by finite real linear combinations of the functions $\phi^{(\vep_N)}_x$, $x \in \Lambda_N$, $N \geq \ell_0$.

Note that if $f\in H_\RR^{\pm\frac{1}{2}}(\TT_L^d)$ has a support contained in a shifted torus,
then $f$ can be approximated as above because the wavelet basis generated by $\{\phi^{(\vep_N)}_x\,:\, x \in \Lambda_N,\, N \geq \ell_0\}$
($\ell_0$ depends on the length of the support of $f$)
is invariant under dyadic translations.
For a general $f\in H_\RR^{\pm\frac{1}{2}}(\TT_L^d)$, we can take a smooth partition of unity on $\TT_L^d$
such that each support is contained in a shifted torus. Then $f$ can be written as a finite sum of functions whose supports are
contained in a shifted torus, and each of them can be approximated in the $H^s$-norm, hence so can $f$ itself. \hfill$\square$
\end{proof}

At this point we can show that our wavelet scaling limit gives the time-zero algebras of the usual massive free field on the cylinder.
\begin{theorem}\label{thm:localalgebras}
Let $K\ge 6$.
Let $\pi_{L,\infty}^{(\infty)}$, $\pi_{L}$ be the GNS representations
of $\cW_{\infty, L}$, $\cW(\fh_L)$ with respect to $\omega^{(\infty)}_{L,\infty}$, $\omega_L$,
and $\beta_L$ be the embedding $\cW_{\infty, L} \to \cW(\fh_L)$ as in Proposition \ref{prop:isoalg1}.
Then there is a unitary operator $V : \cH^{(\infty)}_{L,\infty} \to \cH_{L}$ such that
\[
 V \pi_{L,\infty}^{(\infty)}(W_N(\xi)) V^* = \pi_{L}(\beta_L(W_N(\xi)), \qquad \xi \in \fh_{N,L},\,N \in \NN.
\]
Furthermore, we have $ V \pi_{L,\infty}^{(\infty)}(\cW_{\infty, L}( S))'' V^* = \pi_{L}(\cW_L( S))''$ for any open set $ S \subset \TT_L^d$,  
and $V^*U_L(a,0)V$, $a \in \bigcup_{N \in \NN} \Lambda_N$, implements the spatial dyadic translations $\eta^{(\infty)}_{L|a}$ of $\cW_{\infty, L}$ in the representation $\pi_{L,\infty}^{(\infty)}$, where $U_L(a,0)$, $a \in \TT_L^d$, are the spatial
translations for $\pi_{L}$.
\end{theorem}

\begin{proof}
By Lemma \ref{lem:density}, and by the strong continuity of $\xi \mapsto \pi_{L}(W_{\mathrm{ct}}(\xi))$, we obtain $\pi_{L}(\beta_L(\cW_{\infty, L}))'' = \pi_{L}(\cW(\fh_L))'' = \cB(\cH_{L})$.
The unitary $V$ is given by the uniqueness of GNS representation induced by $\omega^{(\infty)}_{L,\infty}$, together with Lemma \ref{lem:isoalg2}
and the cyclicity of the GNS vector for $\pi_{L}(\beta_L(\cW_{\infty, L}))''  = \cB(\cH_{L})$.

As for local density, it is immediate to see that for any open set $ S \subset \TT_L^d$,
\[
 V \pi_{L,\infty}^{(\infty)}(\cW_{\infty, L}( S))'' V^* \subset \pi_{L}(\cW_L( S))''. 
\]

Let $ S\subset (-L,L)^d$ and
consider $\xi \in \fh_L$ such that $\mathrm{supp}\, \xi\subset  S$.
By arguing as in Lemma \ref{lem:density}, we can choose $N_0$ large enough such that we can approximate $\xi$ by 
finite linear combinations of scaling functions $\{\phi^{(\vep_N)}_x:\supp \phi^{(\vep_N)}_x\subset  S, N\geq N_0\}$. Moreover, as $R^N_\infty\delta^{(N)}_x = \vep_N^{-\frac 12} \phi^{(\vep_N)}_x$, $\mathrm{\supp}\,\phi_x^{(\vep_N)} \subset  S$ entails $W_N(\delta_x^{(N)}) \in \cW_{\infty, L}( S)$ and therefore,
by strong continuity of Weyl operators,
\[
\pi_{L}(W_{\mathrm{ct}}(\xi)) \in V \pi_{L,\infty}^{(\infty)}(\cW_{\infty, L}( S))'' V^*,
\]
hence we have the equality $V \pi_{L,\infty}^{(\infty)}(\cW_{\infty, L}( S))'' V^* = \pi_{L}(\cW_L( S))''$.

Let now $ S\subsetneq\TT^d_L$ be a general open region and $\xi \in\fh_L$ such that $\mathrm{supp}\, \xi \subset  S$.
Again by the argument of Lemma \ref{lem:density}, we can write $\xi$ as a finite sum of functions in $\fh_L$ each of which has support
in an open subset of $ S$ which can be brought inside the fundamental domain $(-L,L)^d$ by a translation in $\vep_{N_0} \ZZ^d$ with $N_0 \in \NN$ sufficiently large. The proof is then concluded by the above argument, together with the fact that the set $\{ \phi^{(\vep_N)}_x\,:\,x \in \Lambda_N, N \geq N_0\}$ is globally invariant under $\vep_{N_0}\ZZ^d$-translations.

Finally, it is clear that the map $R^N_\infty : \fh_{N,L} \to \fh_L$ in Eq.~\eqref{eq:jLN} intertwines the action of the dyadic translations on the lattice and on the torus $\TT^d_L$. Since $\beta_L$ is induced by $R^N_\infty$, this immediately entails the statement about the implementation of the dyadic translations by $V^*U_L(a,0)V$, $a \in \bigcup_{N \in \NN} \Lambda_N$.\hfill$\square$
\end{proof}

We stress the fact that by the above result, the dyadic translations of the inductive limit theory extend continuously to the whole torus translation group.

\subsubsection{Limit of dynamics}
\label{sec:limdyn}

The free lattice dynamics $\eta^{(N)}_{L}:\RR\curvearrowright\cW_{N,L}$ induced by the Hamiltonian $H^{(N)}_{L,0}$ \eqref{eq:freeham} results via second quantization from the harmonic time evolution, $\tau^{(N)}_{L} : \RR \curvearrowright \fh_{N,L}$, on the one-particle space (using the conventions of Section \ref{sec:freefield},
see e.g.\! \cite[(1.8.30)]{BattleWaveletsAndRenormalization}). In the momentum-space representation, i.e.~on $\hat{\fh}_{N,L}$, it is given by the multiplication operator: 
\begin{align}
\label{eq:freedyn}
\tau^{(N)}_{L|t}(\hat{\xi}) &\!=\![\cos(t \gamma_{\mu_{N}})\!+\!i\vep_{N}^{-1}\gamma_{\mu_{N}}^{-1}\sin(t \gamma_{\mu_{N}})] \Re \hat{\xi}\!+\!i[\cos(t \gamma_{\mu_{N}})\!+\!i\vep_{N}\gamma_{\mu_{N}} \sin(t \gamma_{\mu_{N}})]\Im \hat{\xi},
\end{align}
for $t\in\RR$, which is the analogue of the free (continuum) time evolution, $\tau_{L}:\RR\!\curvearrowright\fh_{L}$, given in~\eqref{eq:contdynamics} using $\gamma_{\mu_{N}}$ instead of $\gamma_{m}$.
As this evolution preserves the symplectic structure, $\sigma_{N,L}\circ\tau^{(N)}_{L|t} = \sigma_{N,L}$, we obtain $\eta^{(N)}_{L}$ via
\begin{align}
\label{eq:freedynind}
\eta^{(N)}_{L|t}(W_{N}(\hat{q},\hat{p})) & = W_{N}(\tau^{(N)}_{L|t}(\hat{q},\hat{p})), & t&\in\RR,
\end{align}
where we write $\xi \in \fh_{N,L}$ as the tuple $(\hat{q},\hat{p})$ as before. By construction, each initial state is preserved by the dynamics of the same scale, i.e. $\omega_{\mu_{N},0}\circ\eta^{(N)}_{L|t} = \omega_{\mu_{N},0}$, but the dynamics induced by $H^{(N+1)}_{L,0}$ does not preserve the image of $\alpha^{N}_{N+1}$. In other words, $\eta^{(N')}_{L}$ does not preserve $\alpha^{N}_{N'}(\cW_{N,L})\subset\cW_{N',L}$, $N'>N$. In terms of the one-particle spaces this follows from the fact that multiplying a periodically extended tuple of vectors $(\hat{q},\hat{p})\in\ltwo(\Gamma_{N},(2r_{N})^{-d}\mu_{\Gamma_{N}})$ by a function only periodic on $\Gamma_{N+1}$ does not preserve\footnote{See also Section \ref{sec:mom} for a scaling map scheme using sharp momentum cutoffs, where such a consistent dynamics induced by $H^{(N+1)}_{L,0}$ relative to $\alpha^{N}_{N+1}$ exists.} the periodicity on $\Gamma_{N}$.
Moreover, we need to analyze the convergence of the dynamics relative to inductive limit structure of $\cW_{\infty,L}$ and the GNS representation of $\omega^{(\infty)}_{L,\infty}$ because the lattice dynamics $\eta^{(N)}_{L}$ does not extend\footnote{In terms of the one-particle spaces, this issue manifests in the following way: The complete space  $\fh_{L}$ is isomorphic as real topological vector space to $H^{-\frac{1}{2}}(\TT^{d}_{L})\!\oplus\!H^{\frac{1}{2}}(\TT^{d}_{L})$ which makes it evident that the dynamics $\tau^{(N)}_{L}$ does not extend to $\fh_{L}$ because $\sF^{-1}[\gamma_{\mu_{N}}^{-1}\hat{g}]\!\notin\!H^{\frac{1}{2}}(\TT^{d}_{L})$ for $g \in H^{-\frac{1}{2}}(\TT^{d}_{L})$.} to the $\cW(\fh_{L})$ and the continuum dynamics $\eta_{L}$ is only defined with respect to $\pi^{(\infty)}_{L,\infty}$ (since $\tau_{L}$ does not preserve the subspace $\fh_{\infty,L}\subset\fh_{L}$).
To this end, we use \eqref{eq:jLN} to realize the map $\alpha^{N}_{\infty}$ explicitly. Because of \eqref{eq:scalingprod}, it is induced by
\begin{align}
\label{eq:continuumsmearingfourier}
R^{N}_{\infty} : \fh_{N,L} & \!\rightarrow \fh_{\infty,L}\subset\fh_{L}, & R^{N}_{\infty}(\hat{q},\hat{p}) & = \vep^{\frac{d}{2}}_{N}\hat{\phi}(\vep_{N}\!\ \cdot \!\ )(\hat{q},\hat{p})
\end{align}
in the Fourier space, where the abstract inductive limit $\fh_{\infty,L} = \varinjlim_{N}\fh_{N,L}$ was identified with the subspace of $\fh_L$ spanned by functions in the range of the maps $R^N_\infty$, $N \in \NN$ thanks to the consistency $R^{N'}_{\infty} \circ R^{N}_{N'} = R^{N}_{\infty}$ for $N'>N$ (shown in the proof of Proposition \ref{prop:isoalg1}) and the uniqueness of the inductive limit.
This allows us to investigate the convergence of the sequence of dynamics according to the general form given in \eqref{eq:limdyn} which can be facilitated at the level of one-particle spaces:
\begin{proposition}
\label{prop:limdyn}
For $N\in\NN_{0}$ and for $(\hat{q},\hat{p})\in\fh_{N,L}$, we have:
\begin{align*}
\Big\|R^{N'}_{\infty}(\tau^{(N')}_{L|t}(R^{N}_{N'}(\hat{q},\hat{p})))-\tau_{L|t}(R^{N}_{\infty}(\hat{q},\hat{p}))\Big\|_{L} & \rightarrow 0, \quad \text{ as }\quad N' \rightarrow \infty,
\end{align*}
uniformly on compact sets of $t\in\RR$.
\end{proposition}
\begin{proof}
As $\gamma_{\mu_{N'}}\!\rightarrow\!\gamma_{m}$ pointwise by \eqref{eq:Latmassscaling}, we use dominated convergence relative to the measure space $(\tfrac{\pi}{L}\ZZ^{d}\!,\!(2L)^{-d}\mu_{\frac{\pi}{L}\ZZ^{d}}\!)$ and the decay properties of $\hat{\phi}$. Moreover, we have an integrability estimate uniform in $t\in\RR$ (using \eqref{eq:oneparticlescalingiterated} and \eqref{eq:scalingprod} to obtain an overall integrable factor):
\begin{align*}
 & \Big\|R^{N'}_{\infty}(\tau^{(N')}_{L|t}(R^{N}_{N'}(\hat{q},\hat{p})))-\tau_{L|t}(R^{N}_{\infty}(\hat{q},\hat{p}))\Big\|^{2}_{L} \\
 & = \frac{\vep_{N}^{d}}{(2L)^{d}}\sum_{k\in\frac{\pi}{L}\ZZ^{d}}|\hat{\phi}(\vep_{N}k)|^{2}\left(\gamma_{m}(k)^{-1}\left|(\cos(\gamma_{\mu_{N'}}(k)t)-\cos(\gamma_{m}(k)t))\hat{q}(k) \right.\right. \\[-0.5cm]
 &\hspace{5cm} \left.-(\gamma_{\mu_{N'}}(k)\sin(\gamma_{\mu_{N'}}(k)t)-\gamma_{m}(k)\sin(\gamma_{m}(k)t))\hat{p}(k)\right|^{2} \\
 &\hspace{4.5cm} + \gamma_{m}(k)\left|(\cos(\gamma_{\mu_{N'}}(k)t)-\cos(\gamma_{m}(k)t))\hat{p}(k) \right. \\[-0.125cm]
 &\hspace{5cm} \left.\left. + (\gamma_{\mu_{N'}}(k)^{-1}\sin(\gamma_{\mu_{N'}}(k)t)-\gamma_{m}(k)^{-1}\sin(\gamma_{m}(k)t))\hat{q}(k)\right|^{2}\right) \\
 & \leq \frac{4\vep_{N}^{d}}{(2L)^{d}}\sum_{k\in\frac{\pi}{L}\ZZ^{d}}\gamma_{m}(k)|\hat{\phi}(\vep_{N}k)|^{2}\bigg((\gamma_{m}(k)^{-1}|\hat{q}(k)|+(1+c')^{\frac{1}{2}}|\hat{p}(k)|)^{2} \\[-0.5cm]
 &\hspace{5cm} +\left(|\hat{p}(k)|+\frac{(1-c'')^{-\frac{1}{2}}+1}{2}m^{-1}|\hat{q}(k)|\right)^{2}\bigg) \\
 & < \infty,
\end{align*}
where from the second to the third line, we use the fact that the integrand in the second line is pointwise dominated by the integrand in the third line. This estimate uses the basic inequalities $(1-c'')^{\frac{1}{2}}m\leq\gamma_{\mu_{N'}}$ and $\gamma_{\mu_{N'}}\leq(1+c')^{\frac{1}{2}}\gamma_{m}$ valid uniformly in $k \in \frac\pi L \ZZ^d$ for sufficiently large $N'\in\NN_{0}$ and constants $0<c',c''<1$.
Now, we can exploit the uniform continuity of $\cos$ and $\sin$ to conclude that the convergence is uniform on compact sets of $t\in\RR$. More precisely, we have for $|t|\leq T<\infty$:
\begin{align*}
 & \sup_{|t|\leq T}\Big\|R^{N'}_{\infty}(\tau^{(N')}_{L|t}(R^{N}_{N'}(\hat{q},\hat{p})))-\tau_{L|t}(R^{N}_{\infty}(\hat{q},\hat{p}))\Big\|^{2}_{L} \\
 &\leq\frac{\vep_{N}^{d}}{(2L)^{d}}\sum_{k\in\frac{\pi}{L}\ZZ^{d}}|\hat{\phi}(\vep_{N}k)|^{2}\left( \sup_{|t|\leq T}\bigg\{\gamma_{m}(k)^{-1}\left|(\cos(\gamma_{\mu_{N'}}(k)t)-\cos(\gamma_{m}(k)t))\hat{q}(k) \right.\right.  \\[-0.5cm]
 &\hspace{5cm} \left.-(\gamma_{\mu_{N'}}(k)\sin(\gamma_{\mu_{N'}}(k)t)-\gamma_{m}(k)\sin(\gamma_{m}(k)t))\hat{p}(k)\right|^{2}\bigg\} \\
 &\hspace{4.2cm} +  \sup_{|t|\leq T}\bigg\{\gamma_{m}(k)\left|(\cos(\gamma_{\mu_{N'}}(k)t)-\cos(\gamma_{m}(k)t))\hat{p}(k) \right. \\[-0.375cm]
 &\hspace{5cm} \left.\left. + (\gamma_{\mu_{N'}}(k)^{-1}\sin(\gamma_{\mu_{N'}}(k)t)-\gamma_{m}(k)^{-1}\sin(\gamma_{m}(k)t))\hat{q}(k)\right|^{2}\bigg\}\right).
\end{align*}
Now, because the preceding estimate is independent of $t$, we can use the dominated convergence theorem, the fact that the pointwise supremum of a sequence of measurable function is again measurable, and that for every $k\in\frac{\pi}{L}\ZZ^{d}$,
\begin{align*}
\lim_{N'\rightarrow\infty} \sup_{|t|\leq T} & |\cos(\gamma_{\mu_{N'}}(k)t)-\cos(\gamma_{m}(k)t)| = 0, \\
\lim_{N'\rightarrow\infty} \sup_{|t|\leq T} & |\gamma_{\mu_{N'}}(k)^{-1}\sin(\gamma_{\mu_{N'}}(k)t)-\gamma_{m}(k)^{-1}\sin(\gamma_{m}(k)t)| = 0
\end{align*}
by the uniform continuity of $\sin$ and $\cos$, to conclude the proof. \hfill$\square$
\end{proof}

With $V : \cH^{(\infty)}_{L,\infty} \to \cH_{L}$ the unitary operator in Thm.~\ref{thm:localalgebras} and $\{\eta_{L|t}\}$ the automorphisms of $\cW(\fh_L)$ realizing the free dynamics, it is clear that 
\[
U_L(0,t) V \pi^{(\infty)}_{L,\infty}(W) V^* U_L(0,t) =  \pi_{L}(\eta_{L|t}(\beta_L(W))), \qquad W \in \cW_{\infty,L}, 
\]
so we can use $V$ to transport the continuum free dynamics on $\cH^{(\infty)}_{L,\infty}$ as $U_{L|t} := V^* U_L(0,t) V$, and 
by a slight abuse of notation we will also use $\eta_{L|t}$ to denote $\mathrm{Ad}_{U_{L|t}} $.

\begin{corollary}
\label{cor:limdyn}
Given $N\in\NN_{0}$, $\{\pi^{(\infty)}_{L,\infty}(\alpha^{N'}_{\infty}(\eta^{(N')}_{L|t}(\alpha^{N}_{N'}(W))))\}_{N'>N}$ converges strongly to $\eta_{L|t}(\pi^{(\infty)}_{L,\infty}(\alpha^{N}_{\infty}(W)))$ for all $W\in\cW_{N,L}$ and uniformly on compact sets of $t\in\RR$.
\end{corollary}
\begin{proof}
This will follow from Proposition \ref{prop:limdyn} if we prove that Weyl operators in the representation $\pi^{(\infty)}_{L,\infty}$ viewed as maps from the Hilbert space $\fh_{L}$ to $B(\cH_{L})$ are Lipschitz maps, where $B(\cH_{L})$ is given the strong operator topology.
 To this end, we first observe that by Theorem \ref{thm:localalgebras},
\begin{align*}
 & \left\|\left[\pi^{(\infty)}_{L,\infty}(\alpha^{N'}_{\infty}(\eta^{(N')}_{L|t}(\alpha^{N}_{N'}(W))))-\eta_{L|t}(\pi^{(\infty)}_{L,\infty}(\alpha^{N}_{\infty}(W)))\right]\pi^{(\infty)}_{L,\infty}(W_{M}(\zeta))\Omega^{(\infty)}_{L,\infty}\right\| \\
= & \left\|\left[\pi_{L}(\beta_{L}(\alpha^{N'}_{\infty}(\eta^{(N')}_{L|t}(\alpha^{N}_{N'}(W)))))-\pi_{L}(\eta_{L|t}(\beta_{L}(\alpha^{N}_{\infty}(W))))\right]\pi_{L}(W_{\mathrm{ct}}(R^{M}_{\infty}(\zeta)))\Omega_{L}\right\|
\end{align*}
for $\zeta\in\fh_{M,L}$ and some $M\in\NN_{0}$. Here, $\Omega^{(\infty)}_{L,\infty}$ and $\Omega_{L}$ are the GNS vectors of $\omega^{(\infty)}_{L,\infty}$ and $\omega_{L}$. Next, we specify $W$ to be of the form $W = W_{N}(\xi)$ for $\xi\in\fh_{N,L}$ to write:
\begin{align*}
 & \left\|\left[\pi_{L}(\beta_{L}(\alpha^{N'}_{\infty}(\eta^{(N')}_{L|t}(\alpha^{N}_{N'}(W)))))-\pi_{L}(\eta_{L|t}(\beta_{L}(\alpha^{N}_{\infty}(W))))\right]\pi_{L}(W_{\mathrm{ct}}(R^{M}_{\infty}(\zeta)))\Omega_{L}\right\| \\
= & \left\|\left[\pi_{L}(W_{\mathrm{ct}}(R^{N'}_{\infty}(\tau^{(N')}_{L|t}(R^{N}_{N'}(\xi)))))-\pi_{L}(W_{\mathrm{ct}}(\tau_{L|t}(R^{N}_{\infty}(\xi))))\right]\pi_{L}(W_{\mathrm{ct}}(R^{M}_{\infty}(\zeta)))\Omega_{L}\right\| \\
= & \left\|\left[\pi_{L}(W_{\mathrm{ct}}(\xi^{(t)}_{N,N'}))-\pi_{L}(W_{\mathrm{ct}}(\xi^{(t)}_{N,\infty}))\right]\pi_{L}(W_{\mathrm{ct}}(R^{M}_{\infty}(\zeta)))\Omega_{L}\right\|,
\end{align*}
with the short hands $\xi^{(t)}_{N,N'}$ and $\xi^{(t)}_{N,\infty}$ for $R^{N'}_{\infty}(\tau^{(N')}_{L|t}(R^{N}_{N'}(\xi)))$ and $\tau_{L|t}(R^{N}_{\infty}(\xi))$ in the last line. By using the Weyl relations \eqref{eq:CCR} and the explicit form \eqref{eq:vaccyl} of the expectation values of Weyl operators in the state $\omega_{L}$, we find (cp. \cite[Proposition 5.2.29]{BratteliOperatorAlgebrasAnd2}):
\begin{align*}
 & \left\|\left[\pi_{L}(W_{\mathrm{ct}}(\xi^{(t)}_{N,N'}))-\pi_{L}(W_{\mathrm{ct}}(\xi^{(t)}_{N,\infty}))\right]\pi_{L}(W_{\mathrm{ct}}(R^{M}_{\infty}(\zeta)))\Omega_{L}\right\| \\
\leq & (\|R^{M}_{\infty}(\zeta)\|_{L}+\tfrac{1}{2}\|\xi^{(t)}_{N,\infty}\|_{L})\|\xi^{(t)}_{N,N'}-\xi^{(t)}_{N,\infty}\|_{L} + 2\left(1-e^{-\frac{1}{4}\|\xi^{(t)}_{N,N'}-\xi^{(t)}_{N,\infty}\|_{L}^{2}}\right) \\
\leq & (\|R^{M}_{\infty}(\zeta)\|_{L}+\tfrac{1}{2}\|\xi^{(t)}_{N,\infty}\|_{L}+\tfrac{1}{2})\|\xi^{(t)}_{N,N'}-\xi^{(t)}_{N,\infty}\|_{L} \\
= & (\|R^{M}_{\infty}(\zeta)\|_{L}+\tfrac{1}{2}\|R^{N}_{\infty}(\xi)\|_{L}+\tfrac{1}{2})\|\xi^{(t)}_{N,N'}-\xi^{(t)}_{N,\infty}\|_{L},
\end{align*}
where we use the unitary implementation of $\tau_{L}$ on $\fh_{L}$. Thus, we obtain the anticipated Lipschitz estimate. Now, Proposition \ref{prop:limdyn} states that for arbitrary $0<T<\infty$:
\begin{align*}
\lim_{N'\rightarrow\infty}\sup_{|t|\leq T}\left\|\xi^{(t)}_{N,N'}-\xi^{(t)}_{N,\infty}\right\|_{L} & = 0,
\end{align*}
and, thus, we have,
\begin{align*}
\lim_{N'\rightarrow\infty}\sup_{|t|\leq T}\left\|(\pi_{L}(W_{\mathrm{ct}}(\xi^{(t)}_{N,N'}))-\pi_{L}(W_{\mathrm{ct}}(\xi^{(t)}_{N,\infty})))\Psi\right\| & = 0,
\end{align*}
for the total set of vectors $\{\Psi\in\cH_{L}\!\ |\!\ \Psi = \pi_{L}(W_{\mathrm{ct}}(R^{M}_{\infty}(\zeta)))\Omega_{L}, \zeta\in\fh_{M,L}, M\in\NN_{0}\}$. Finally, the result follows from combining this with the fact that the Weyl operators $W=W_{N}(\xi)$, $\xi\in\fh_{N,L}$, form a total set of uniformly bounded operators in $\cW_{N,L}$. \hfill$\square$
\end{proof}

\begin{remark}
\label{rem:freedynfock}
We can identify $\fh_{L}$ with a subspace of $\ltwo(\tfrac{\pi}{L}\ZZ^{d}, (2L)^{-d}\mu_{\frac{\pi}{L}\ZZ^{d}})^{\oplus 2}$ via the Fourier transform and the dispersion relation $\gamma_{m}$:
\begin{align*}
(\hat{f},\hat{g}) & \mapsto (\gamma_{m}^{-\frac{1}{2}}\hat{f},\gamma_{m}^{\frac{1}{2}}\hat{g}). 
\end{align*}
This way $\tau_{L}$ is given by multiplication with the matrix-valued function $\exp(i\gamma_{m}\sigma_{2})$, where $\sigma_{2}$ is the second Pauli matrix. Clearly, we can implement a dynamics given by $\gamma_{\mu_{N}}$ on $\ltwo(\tfrac{\pi}{L}\ZZ^{d}, (2L)^{-d}\mu_{\frac{\pi}{L}\ZZ^{d}})^{\oplus 2}$ using multiplication with $\exp(i\gamma_{\mu_{N}}\sigma_{2})$. But, this is not the action of $\tau^{(N)}_{L}$ on the subspace $\fh_{N,L}\subset\fh_{L}$ and, thus, the convergence of the bounded matrix-functions
\begin{align*}
\exp(i\gamma_{\mu_{N}}\sigma_{2}) & \rightarrow \exp(i\gamma_{m}\sigma_{2}), & N & \rightarrow \infty,
\end{align*}
is not the convergence of dynamics sought after.
\end{remark}
By construction, $\eta_{L}:\RR\curvearrowright\pi^{(\infty)}_{L,\infty}(\cW_{\infty, L})''$ equals the usual time evolution of the continuum free scalar field
and it is implemented by a unitary group $U_{L|t}=e^{itH_{L}}$ with the (renormalized) free continuum Hamiltonian $H_{L}$ as its generator. Explicitly, $H_{L}$ is the second quantization of the generator $h_{L}$ of $\tau_{L}$ on its natural domain \cite{BuchholzCausalIndependenceAnd}. Since $\gamma_{m}$ is the free relativistic dispersion relation of mass $m$, we conclude that $\eta_{L}$ has propagation speed $c=1$, and we obtain a causal net of (spacetime) local von Neumann algebras for suitable $\cO\subset\RR\times\TT^{d}_{L}$ \cite{GlimmQuantumFieldTheory}:
\begin{align}
\label{eq:causalnet}
\cA_{L}(\cO) & = \left(\bigcup_{t\in\RR}\eta_{L|t}(\pi^{(\infty)}_{L,\infty}(\cW_{\infty, L}(\cO(t)))'')\right)'',
\end{align}
where $\cO(t) = \{x\!\ |\!\ (t,x)\in\cO\}\subset\TT^{d}_{L}$.

Let us summarize our results on the scaling limit of the free ground states of the lattice scalar field by the following theorem:

\begin{theorem}
\label{thm:main}
Let $K\ge 6$.
Let $m, \{\mu_{N}\}, \{\cW_{N,L}\}, \{\alpha_{N'}^N\}, \{\omega_{L,M}^{(N)}\}$ as before.
The scaling limit $\omega_{L,\infty}^{(\infty)}$ of $\{\omega_ {L, M}^{(N)}\}$
gives rise to a continuum time-zero net of local algebras which is unitarily
equivalent to the time-zero net of the free field with mass $m$ on the torus $\TT^{d}_{L}$.
\end{theorem}

\begin{remark}
\label{rem:limdyntriangle}
In view of Figures \ref{fig:trianglerg} and \ref{fig:statetrianglerg}, we note that while we used the scaling limit procedure to construct the horizontal sequences $\{\omega^{(N)}_{L,M}\}_{M\in\NN_{0}}$ to arrive at the limit state $\omega^{(\infty)}_{L,\infty}$, we only considered the diagonal sequence of Hamiltonian $H^{(N)}_{L,0}$, or more precisely their associated unitary groups $\eta^{(N)}_{L}$, to obtain the Hamiltonian $H_{L}$ associated with the continuum time evolution $\eta_{L}$.
\end{remark}
\begin{remark}
\label{rem:limdynlat}
There is yet another perspective on the convergence of dynamics because of the specific structure of the scaling maps $\alpha^{N}_{N+1}$: At each finite scale $N$, the dynamics $\eta^{(N)}_{L}$ is unitarily implemented in the GNS representation of $\omega^{(N)}_{L,0}$, and because of the von Neumann uniqueness theorem there is an associated strongly-continuous, unitary one-parameter group $U^{(N)}:\RR\rightarrow\cU(\cH^{(\infty)}_{N,L})$ on the GNS Hilbert space $\cH^{(\infty)}_{N,L}$ of $\omega^{(N)}_{L,\infty}$. In Section \ref{sec:MERA} below, we show that the scaling maps $\alpha^{N}_{N+1}$ extend as normal, unital $^*$-morphisms from $B(\cH_{N,L})$ to $B(\cH_{N+1,L})$, which in turn yields their extension from $B(\cH^{(\infty)}_{N,L})$ to $B(\cH^{(\infty)}_{N+1,L})$. This implies that each $U^{(N)}$ is represented in $\cU(\cH^{(\infty)}_{L,\infty})$ by\footnote{The asymptotic unital $^*$-morphisms $\alpha^{N}_{\infty}:B(\cH_{N,L})\rightarrow B_{\infty}$ is not necessarily normal, cf.~Section \ref{sec:MERA} for the notation.}:
\begin{align*}
U^{(N)}_{L|t} & = \pi^{(\infty)}_{L,\infty}(\alpha^{N}_{\infty}(U^{(N)}_{t})), & t & \in\RR.
\end{align*}
Therefore, as an alternative to the convergence considered in Corollary \ref{cor:limdyn}, one could analyze the convergence:
\begin{align*}
U^{(N)}_{L|t} & \longrightarrow U_{L|t}, & N & \longrightarrow \infty.
\end{align*}
\end{remark}

\paragraph{Locality by Lieb-Robinson bounds.}
Interestingly, there is another way to show a weaker form of locality via Lieb-Robinson bounds \cite{OsborneContinuumLimitsOf, NachtergaeleQuasiLocalityBounds}. To this end, we observe that the finite-scale dynamics $\eta^{(N)}_{L}$ extends to the $C^{*}$-inductive limit $\cW_{\infty, L}$:
For $N\in\NN_{0}$ we consider the periodic extension of $\gamma_{\mu_{N}}$ from $\Gamma_{N}$ to $\Gamma_{\infty}=\tfrac{\pi}{L}\ZZ^{d}$, which we also denote by $\gamma_{\mu_{N}}$. Then, we extend $\tau^{(N)}_{L}$ to $\fh_{M,L}$ for all $M>N$ by
\begin{align*}
\tau^{(N)}_{L|t}(\hat{q},\hat{p}) & = (\cos(\gamma_{\mu_{N}}t)\hat{q}\!-\!\gamma_{\mu_{N}}\sin(\gamma_{\mu_{N}}t)\hat{p}, \quad \cos(\gamma_{\mu_{N}}t)\hat{p}\!+\!\gamma_{\mu_{N}}^{-1}\sin(\gamma_{\mu_{N}}t)\hat{q}), & (\hat{q},\hat{p})\in\fh_{M,L},
\end{align*}
which is well-defined because the restriction of $\gamma_{\mu_{N}}$ to $\Gamma_{M}\subset\Gamma_{\infty}$ is well-defined by periodicity. Comparing the extension for $M'>M>N$, we find the consistency relation:
\begin{align*}
\tau^{(N)}_{L|t}(R^{M}_{M'}(\hat{q},\hat{p})) & = R^{M}_{M'}(\tau^{(N)}_{L|t}(\hat{q},\hat{p})), \qquad (\hat q,\hat p) \in \fh_{M,L}
\end{align*}
which is due to the multiplicative character of $R^{M}_{M'}$ in momentum space, see \eqref{eq:oneparticlescalingiterated}. Therefore, we can extend $\tau^{(N)}_{L}$ to the inductive-limit symplectic space $\fh_{\infty,L}$ by:
\begin{align}
\label{eq:finiteinddyn1p} \nonumber
\tau^{(N)}_{L|t}\!(R^{M}_{\infty}(\hat{q},\hat{p})\!) & \!=\! \vep_{M}^{\frac{d}{2}}\hat{\phi}(\vep_{M}\!\ \cdot \!\ )(\cos(\gamma_{\mu_{N}}t)\hat{q}\!-\!\gamma_{\mu_{N}}\sin(\gamma_{\mu_{N}}t)\hat{p},\!\ \cos(\gamma_{\mu_{N}}t)\hat{p}\!+\!\gamma_{\mu_{N}}^{-1}\sin(\gamma_{\mu_{N}}t)\hat{q}) \\ \nonumber
& = R^{M}_{\infty}(\cos(\gamma_{\mu_{N}}t)\hat{q}\!-\!\gamma_{\mu_{N}}\sin(\gamma_{\mu_{N}}t)\hat{p},\cos(\gamma_{\mu_{N}}t)\hat{p}\!+\!\gamma_{\mu_{N}}^{-1}\sin(\gamma_{\mu_{N}}t)\hat{q}) \\
& = R^{M}_{\infty}(\tau^{(N)}_{L|t}(\hat{q},\hat{p})),
\end{align}
for any $M>N$.

This extension of $\tau^{(N)}_{L}$ to $\fh_{\infty,L}$ provides another justification for considering the convergence expressed in Proposition \ref{prop:limdyn} because it implies that the convergence also holds on $\fh_{\infty,L}\subset\fh_{L}$ consistently over all scales. An explicit computation shows that the extension of $\tau^{(N)}_{L}$ is symplectic for each $M>N$, i.e.~$\sigma_{M,L}(\tau^{(N)}_{L|t}(q,p),\tau^{(N)}_{L|t}(q',p'))=\sigma_{M,L}((q,p),(q',p'))$, which permits the extension of $\eta^{(N)}_{L}$ to $\cW_{\infty,L}$ as a $^*$-automorphism satisfying
\begin{align}
\label{eq:finiteinddyn}
\eta^{(N)}_{L|t}(\alpha^{M}_{\infty}(W_{M}(q,p))) & = \alpha^{M}_{\infty}(\eta^{(N)}_{L|t}(W_{M}(q,p))) =  \alpha^{M}_{\infty}(W_{M}(\tau^{(N)}_{L|t}(q,p))),
\end{align}
by general principles \cite{EvansQuantumSymmetriesOn}. Now, we can invoke the Lieb-Robinson bounds for harmonic lattice systems \cite{NachtergaeleLRBoundsHarmonic} adapted to the Hamiltonian $H^{(N)}_{L,0}$ and the Weyl algebra $\cW_{N,L}$:
\begin{lemma}
\label{lem:liebrob}
For arbitrary subsets $X,Y\subset\Lambda_{N}$, all volume parameters $L=\vep_{N}r_{N}>0$, any one-particle vector $\xi,\xi'\in\fh_{N,L}$ with supports $\supp\xi\subset X$, $\supp\xi'\subset Y$, and any $\delta>0$, we have:
\begin{align*}
\left\|\left[\eta^{(N)}_{L|t}(W_{N}(\xi)),W_{N}(\xi')\right]\right\| & \leq C_{N}\|\xi|\|_{\sup}\|\xi'\|_{\sup}\sum_{x\in X,y\in Y}e^{-\frac{\delta}{\vep_{N}}\left(d_{N}(x,y)-\frac{1}{2}c_{\mu_{N}}\max\left\{\frac{2}{\delta},e^{\frac{\delta}{2}+1}\right\}|t|\right)},
\end{align*}
where
\begin{align*}
d_{N}(x,y) & = \sum^{d}_{j=1}\min_{n_j\in\ZZ}|x_{j}-y_{j}+2Ln_{j}|, & x,y&\in\Lambda_{N},
\end{align*}
is the $1$-distance on the torus $\TT^{d}_{L}$ restricted to $\Lambda_{N}$, $\|\!\ \cdot\!\ \|_{\sup} = \sup_{\Lambda_{N}}|\!\ \cdot\!\ |$, and
\begin{align*}
c_{\mu_{N}} & = (\mu_{N}^{2}+2d)^{\frac{1}{2}}, & C_{N} & = 2+c_{\mu_{N}}e^{\frac{\delta}{2}}+c_{\mu_{N}}^{-1}.
\end{align*}
\end{lemma}
\begin{proof}
The estimate is a direct consequence of the estimate given in \cite[Theorem 3.1]{NachtergaeleLRBoundsHarmonic} after reinstating the lattice spacing $\vep_{N}$. \hfill$\square$
\end{proof}
The estimate given in the Lemma leads to a speed of propagation $c'\ge 1$:
\begin{proposition}
\label{prop:LRalgind}
Let $K\ge 2$.
Provided the renormalization condition \eqref{eq:Latmassscaling} holds, there exists a constant $1 \le c' \le d^\frac12\max\{\frac2 \delta, e^{\frac\delta 2 + 1}\}$, the \textit{scaling-limit Lieb-Robinson velocity}, such that for $ S'\cap S_{c'T}=\emptyset$ with $ S_{c'T} = \{x\!\ |\!\ \textup{dist}(x, S)<c'T\}$,
\begin{align*}
\lim_{N\rightarrow\infty}\left\|\left[\eta^{(N)}_{L|t}(W),W'\right]\right\| & = 0,
\end{align*}
for all $W\in\cW_{\infty, L}(S), W'\in\cW_{\infty, L}(S')$ exponentially fast and uniformly for $|t|<T$.
\end{proposition}
\begin{proof}
First, we consider two elementary Weyl operators $W, W'\in\cW_{\infty, L}$. By construction there exist $N, N'\in\NN_{0}$ and $\xi\in\fh_{N,L}$, $\xi'\in\fh_{N',L}$ such that $W = \alpha^{N}_{\infty}(W_{N}(\xi))$ and $W'=\alpha^{N'}_{\infty}(W_{N'}(\xi'))$ (without loss of generality we may assume $N\leq N'$).
For any $M\geq N'\geq N$ we have:
\begin{align*}
\left[\eta^{(M)}_{L|t}(W),W'\right] & = \alpha^{M}_{\infty}\left(\left[\eta^{(M)}_{L|t}(W_{M}(R^{N}_{M}(\xi))),W_{M}(R^{N'}_{M}(\xi'))\right]\right).
\end{align*}
Since $\alpha^{M}_{\infty}$ is an injective $C^{*}$-homomorphism, we can apply the estimate of Lemma \ref{lem:liebrob}:
\begin{align*}
\left\|\left[\eta^{(M)}_{L|t}(W),W'\right]\right\| & \leq C_{M}\|R^{N}_{M}(\xi)\|_{\sup}\|R^{N'}_{M}(\xi')\|_{\sup}\!\!\!\!\sum_{x\in X,y\in Y}\!\!\!\!\!e^{-\frac{\delta}{\vep_{M}}\left(d_{M}(x,y)-\frac{1}{2}c_{\mu_{M}}\max\left\{\frac{2}{\delta},e^{\frac{\delta}{2}+1}\right\}|t|\right)},
\end{align*}
and by the renormalization conditions, we have
\[
\lim_{M\rightarrow\infty}c_{\mu_{M}}=2d^{\frac{1}{2}},\quad \lim_{M\rightarrow\infty}C_{M}=2+2d^{\frac{1}{2}}e^{\frac{\delta}{2}}+2^{-1}d^{-\frac{1}{2}}=C_{d}.
\]

Now, if $\supp R^{N}_{\infty}(q_{\xi},p_{\xi})\subset S$ and $\supp R^{N'}_{\infty}(q_{\xi'},p_{\xi'})\subset S'$, we can apply the estimate with $\textup{dist}_{1}( S, S')\leq\lim_{M\rightarrow\infty}d_{M}(x,y)<\infty$, where $\textup{dist}_{1}( S, S')$ is the $1$-distance between subsets of $\TT^{d}_{L}$. Moreover, by \eqref{eq:canonicaldimension},
\begin{align*}
\|R^{N}_{M}(\xi)\|_{\sup} & = \vep_{M}^{\frac{d}{2}}\|\vep_{M}^{\frac{1}{2}}R^{N}_{M}(q_{\xi})+i\vep_{M}^{-\frac{1}{2}}R^{N}_{M}(p_{\xi})\|_{\sup}, \\
\|R^{N'}_{M}(\xi')\|_{\sup} & = \vep_{M}^{\frac{d}{2}}\|\vep_{M}^{\frac{1}{2}}R^{N'}_{M}(q_{\xi'})+i\vep_{M}^{-\frac{1}{2}}R^{N'}_{M}(p_{\xi'})\|_{\sup},
\end{align*}
which can be combined with \eqref{eq:sobolevseq} \& \eqref{eq:regularityestimate} for $\alpha=0$ in Lemma \ref{lem:freegroundconv} to show:
\begin{align*}
\lim_{M\rightarrow\infty}\|R^{N}_{M}(q_{\xi})\|_{\sup} & \leq \|\hat{\phi}^{(\vep_{N})}\hat{q_{\xi}}\|_{L^1,L} < \infty, &  \lim_{M\rightarrow\infty}\|R^{N}_{M}(p_{\xi})\|_{\sup} & \leq \|\hat{\phi}^{(\vep_{N})}\hat{p_{\xi}}\|_{L^1,L} < \infty,
\end{align*}
for all $N\in\NN_{0}$, and similarly for $\xi'$, where $\|\!\ .\!\ \|_{L^1,L}$ is the norm of $L^1(\tfrac{\pi}{L}\ZZ^{d},(2L)^{-d}\mu_{\frac{\pi}{L}\ZZ^{d}})$.
Thus, we find (in the limit $M\rightarrow\infty$):
\begin{align*}
\left\|\left[\eta^{(M)}_{L|t}(W),W'\right]\right\| & \lesssim C_{d}\!\ \vep_{M}^{-d}\!\ \left(\vep_{M}^{\frac{1}{2}}\|\hat{\phi}^{(\vep_{N})}\hat{q_{\xi}}\|_{L^1,L}+\vep_{M}^{-\frac{1}{2}}\|\hat{\phi}^{(\vep_{N})}\hat{p_{\xi}}\|_{L^1,L}\right) \\
& \hspace{1cm}\times\left(\vep_{M}^{\frac{1}{2}}\|\hat{\phi}^{(\vep_{N})}\hat{q_{\xi'}}\|_{L^1,L}+\vep_{M}^{-\frac{1}{2}}\|\hat{\phi}^{(\vep_{N})}\hat{p_{\xi'}}\|_{L^1,L}\right) \\
& \hspace{1cm}\times\vol( S)\vol( S')\!\ e^{-\frac{\delta}{\vep_{M}}\left(\textup{dist}_{1}( S, S')-d^{\frac{1}{2}}\max\left\{\frac{2}{\delta},e^{\frac{\delta}{2}+1}\right\}|t|\right)}. \\
& \rightarrow 0,
\end{align*}
where the expressions $\vol(S)$, $\vol(S')$ arise by collecting an overall factor $\vep_{M}^{2d}$ in front of the sum $\sum_{x\in X,y\in Y}$ and choosing $X = \Lambda_{M}\cap S$, $Y = \Lambda_{M}\cap S'$:
\begin{align*}
\vol( S) & = \lim_{M\rightarrow\infty}\vep_{M}^{d}\sum_{x\in\Lambda_{M}\cap S}1, & \vol( S') & = \lim_{M\rightarrow\infty}\vep_{M}^{d}\sum_{y\in\Lambda_{M}\cap S'}1.
\end{align*}
Using the equivalence of $\textup{dist}_{1}$ and $\textup{dist}$, i.e.~$\textup{dist}\leq\textup{dist}_{1}\leq d^{\frac{1}{2}}\textup{dist}$, and $\lim_{M\rightarrow\infty}\vep_{M}^{-k}e^{-\frac{\delta}{\vep_{M}}} = 0$ for any $k\in\ZZ$, the statement is proved for elementary Weyl operators and their finite linear combinations with the scaling-limit Lieb-Robinson velocity bounded by:
\begin{align*}
1\leq\max\left\{\frac{2}{\delta},e^{\frac{\delta}{2}+1}\right\} & \leq c'\leq d^{\frac{1}{2}}\max\left\{\frac{2}{\delta},e^{\frac{\delta}{2}+1}\right\}.
\end{align*}
Then, because $\eta^{(M)}_{L}$ is a group of $^*$-automorphisms of $\cW_{\infty,L}$ and  Weyl operators form a bounded total set in $\cW_{\infty,L}(S)$, $\cW_{\infty,L}(S')$, we obtain the desired statement. \hfill$\square$
\end{proof}
\begin{corollary}
\label{cor:LRstrongind}
Let $K\ge 2$.
With the scaling-limit Lieb-Robinson velocity $c'$ from Proposition \ref{prop:LRalgind} and $ S'\cap S_{c'T}=\emptyset$ with $ S_{c'T} = \{x\!\ |\!\ \textup{dist}(x, S)<c'T\}$,
\begin{align*}
\sotlim_{N\rightarrow\infty}\left[\eta^{(N)}_{L|t}(W),W'\right] & = 0,
\end{align*}
for all $W\in\cW_{\infty, L}(S), W'\in\cW_{\infty, L}(S')$ exponentially fast and uniformly for $|t|<T$ relative to the representation given by $\omega^{(\infty)}_{L,\infty}$. Moreover, for all such $t\in\RR$ and for all $W\in\pi^{(\infty)}_{L,\infty}(\cW_{\infty, L}(S))'', W'\in\pi^{(\infty)}_{L,\infty}(\cW_{\infty, L}(S'))''$:
\begin{align*}
\left[\eta_{L|t}(W),W'\right] & = 0.
\end{align*}
\end{corollary}
\begin{proof}
The statement about SOT (strong operator topology) convergence is immediate from the norm convergence of Proposition \ref{prop:LRalgind}. The commutativity for weak closures follows from the SOT convergence of $\eta^{(N)}_{L}$ to $\eta_{L}$ and the SOT continuity of $\eta_{L}$. \hfill$\square$
\end{proof}
\begin{remark}
\label{rem:LRvelocity}
Note, that even in the case $\textup{dist}_{1}=d^{\frac{1}{2}}\textup{dist}$, the (optimal) scaling-limit Lieb-Robinson velocity is $c'=\tfrac{2}{\delta_{0}}\sim 3.59$, given by the Lambert W function,
\begin{align*}
\tfrac{\delta_{0}}{2}e^{\frac{\delta_{0}}{2}} & = e^{-1},
\end{align*}
which is still far from the correct propagation speed $c=1$.
\end{remark}

\subsubsection{Infinite volume limit}
\label{sec:therm}
Here we compare the finite volume theory and the infinite volume theory.
Given a bounded open set $ S \subset \RR^d$, one has $ S \Subset (-L,L)^d$ for sufficiently large $L$. 
The $C^*$-algebra $\cW_{\infty, L}( S)$ actually does not depend on $L$ (possibly $L=\infty$). Indeed, with the notation  $\fh_{N,L}( S) := \ltwo(\Lambda_{N}( S))$ (cf. Sec.~\ref{sec:localg}), $\beta_L (\cW_{\infty, L}( S)) \subset \cW(\fh_L)$ is generated by Weyl operators $W_{\mathrm{ct}}(\xi)$, $\xi \in \fh_{\infty,L}( S) := \bigcup_N R^N_\infty \fh_{N,L}( S)$, and 
each $\xi \in \fh_{\infty,L}( S)$ can be also thought as a compactly supported function on $\RR^d$,
and $\sigma_L(\xi,\eta) = \sigma_\infty(\xi,\eta)$ for $\xi, \eta\in \fh_L( S)$ (see \eqref{eq:sympcont}, \eqref{eq:sympcontinf}).
Therefore, $\cW_{\infty, L}( S)$ can also be 
identified with the $C^*$-subalgebra of the infinite volume free field $C^*$-algebra $\cW(\fh_\infty)$ generated
by the Weyl operators $W_{\mathrm{ct}}(\xi)$, $\xi \in \fh_{\infty,L}( S)$. 

With this embedding  and Lemma \ref{lem:isoalg2} in mind, we can consider on $\cW_{\infty, L}( S)$ the family of continuum states $\{ \omega_L\}_{L > 0}$,
and the limit $L\to \infty$.
\begin{proposition}
For each $\xi \in \fh_{\infty,L}( S)$ we have
\[
\lim_{L \to +\infty} \omega_L(W_{\mathrm{ct}}(\xi)) = \omega_{\infty}(W_{\mathrm{ct}}(\xi)),
\]
and $\{\omega_L\}$ converges to $\omega_{\infty}$ in the weak* topology on $\cW_{\infty, L}( S)$. 
\end{proposition}
\begin{proof}
It is enough to compute the exponent which appears in the expectation values:
\[
\lim_{L \to +\infty} \tfrac{1}{(2L)^d}\sum_{k\in\frac\pi L\ZZ^d}\bigg|\frac{\hat q_\xi(k)}{\gamma_m^{1/2}(k)}+i\gamma_m^{1/2}(k)\hat p_\xi(k)\bigg|^2 = 
\tfrac{1}{(2\pi)^d} \int_{\RR^d} d^dk\,\bigg|\frac{\hat q_\xi(k)}{\gamma_m^{1/2}(k)}+i\gamma_m^{1/2}(k)\hat p_\xi(k)\bigg|^2.
\]
Indeed, if the right-hand side is integrable, one can restrict it to a sufficiently large compact interval with a small error, and the corresponding restriction of the sum on the left hand side is then a corresponding Riemann sum, which is convergent in the $L \to \infty$ limit due to the regularity of $\hat q_\xi,\hat p_\xi$.
The weak$^*$ convergence follows by approximation in norm. \hfill$\square$
\end{proof}

We now prove that the GNS representations with respect to the vacuum states with different volumes $L$ are locally equivalent.
This enables us to identify locally the finite volume algebras in a cylinder and in the Minkowski space, showing that we can obtain the local algebras of the infinite volume continuum free field through our finite volume scaling limit procedure.

Following the steps of \cite[Appendix A]{ContiMorsella}, we first show that the two representations are quasi-equivalent 
by the main theorem in \cite{ArakiYamagami}, then this is actually unitary equivalence because local algebras are type III factors.
The assumptions of \cite{ArakiYamagami} are satisfied if one proves the following facts.
\begin{enumerate}
\item The topologies induced by the two-point functions on the one-particle spaces are equivalent.
\item The  kernels of the quadratic forms associated with the difference of the two point functions on the torus and on $\RR^d$ ($Q_+$ and $Q_-$ in Lemma~\ref{lem:diff2}) are sufficiently regular. This entails the Hilbert-Schmidt property of the Araki-Yamagami operator.
\end{enumerate}

We show them in Lemma \ref{lem:normequivalence} and Lemma \ref{lem:diff2}, respectively.
As a preparation for the equivalence of topologies we first prove two versions of the Young inequality relating continuous $L^2$-norms with discrete convolutions and vice versa.

\begin{lemma}\label{lem:young}
Given $f \in \cS(\RR)$, $g \in L^2(\RR)$, there holds
\begin{align*}
\int_{\RR} dk\, \bigg| \sum_{p \in \frac\pi L\ZZ} |f(k-p)| |g(p)| \bigg|^2 &\leq \max_{k \in [0,\frac \pi L]} \| f_k\|_{\ell^1(\frac \pi L \ZZ)} \|f\|_{L^1(\RR)} \|g\|_{\ell^2(\frac\pi L \ZZ)}^2,\\
\sum_{k \in \frac\pi L \ZZ} \left| \int_{\RR}dp\, |f(k-p)| |g(p)| \right|^2 &\leq \max_{k \in [0,\frac \pi L]} \| f_k\|_{\ell^1(\frac \pi L \ZZ)} \|f\|_{L^1(\RR)} \|g\|_{L^2(\RR)}^2,
\end{align*}
where $f_k(p) := f(k-p)$, $k, p \in \RR$ and $\| h\|^s_{\ell^s(\frac \pi L \ZZ)} = \sum_{p \in \frac \pi L\ZZ} |h(k-p)|^s$, $s=1,2$.
\end{lemma}

\begin{proof}
There holds
\[\begin{split}
\int_{\RR} dk\, \bigg| \sum_{p \in \frac\pi L\ZZ} |f(k-p)| |g(p)| \bigg|^2 &= \int_{\RR} dk\, \bigg| \sum_{p \in \frac\pi L\ZZ} |f(k-p)|^{1/2} |g(p)| |f(k-p)|^{1/2} \bigg|^2\\
&\leq \int_{\RR} dk \bigg[\sum_{p \in \frac\pi L \ZZ} |f(k-p)| |g(p)|^2\bigg] \bigg[\sum_{p \in \frac \pi L\ZZ} |f(k-p)|\bigg],
\end{split}\]
as the last expression contains more positive terms than the previous one.
Since $f \in \cS(\RR)$, the function $k \mapsto \sum_{p \in \frac \pi L\ZZ} |f(k-p)| =  \| f_k\|_{\ell^1(\frac \pi L \ZZ)}$ is $\frac\pi L$-periodic and continuous (the series converges uniformly in $k \in [0,\frac \pi L]$), and therefore $\sum_{p \in \frac \pi L\ZZ} |f(k-p)| \leq \max_{k \in [0,\frac \pi L]} \| f_k\|_{\ell^1(\frac \pi L \ZZ)} < \infty$ for all $k \in \RR$. Moreover
\[
\int_{\RR} dk \bigg[\sum_{p \in \frac\pi L \ZZ} |f(k-p)| |g(p)|^2\bigg]  = \sum_{p \in \frac \pi L \ZZ} |g(p)|^2 \int_{\RR} dk\, |f(k-p)| = \|g\|_{\ell^2(\frac\pi L \ZZ)}^2  \|f\|_{L^1(\RR)},
\]
which proves the first inequality. The second one is proven in the same way, exchanging the roles of sums and integrals. \hfill$\square$
\end{proof}

Let us show the first of the conditions for unitary equivalence.
\begin{lemma}\label{lem:normequivalence}
For $S \Subset (-L,L)$ the norms $\|\cdot\|_L$ and $\|\cdot\|_\infty$ defined by~\eqref{eq:spcyl} and \eqref{eq:scalarmink} respectively,
are equivalent on $\fh_{\infty,L}( S)$.
\end{lemma}

\begin{proof}
Setting, for a function $\xi : \RR \to \RR$,
\[
\|\xi \|_{L,\pm}^2 := \frac1{2L} \sum_{k \in \frac \pi L \ZZ} \gamma_m(k)^{\pm 1} | \hat \xi(k)|^2, \qquad \|\xi \|_{\infty,\pm}^2 := \frac1{2\pi} \int_{\RR}dk\, \gamma_m(k)^{\pm 1} | \hat \xi(k)|^2,
\]
it is sufficient to show that on $\fh_{\infty,L}( S)$ the norm $\| \cdot\|_{L,\pm}$ is equivalent to $\|\cdot\|_\pm$ respectively.

To this end, let $\xi \in \fh_{\infty,L}( S)$ be a real function, and let $\chi \in C_c^\infty((-L,L))$ be such that $\chi(x)=1$ for all $x \in \supp \xi$. Since $\supp \xi \subset (-L,L)$, we can define its $2L$-periodization
\[
P\xi(x) := \sum_{n \in \ZZ} \xi(x+2Ln) = \sum_{k \in \frac\pi L \ZZ} \hat \xi(k) e^{ikx},
\]
and thanks to $\sum_{k \in \frac\pi L \ZZ} |\hat \xi(k)| < +\infty$ we obtain, for all $k \in \RR$,
\begin{align*}
\hat \xi(k) & = \widehat{\chi P\xi}(k) = \int_\RR\! dx\, \chi(x) P\xi(x) e^{-ikx} = \!\!\!\sum_{p \in \frac \pi L \ZZ}\!\!\!\hat\xi(p)\!\!\!\int_{\RR}\! dx\, \chi(x) e^{-i(k-p)x} = \!\!\!\sum_{p \in \frac \pi L \ZZ} \!\!\!\hat\chi(k-p) \hat\xi(p). 
\end{align*}
Now since $k \mapsto \gamma_m(k)$ is a convex function,
\[
\gamma_m(k) = \gamma_m(k-p+p) \leq \frac12(\gamma_m(2(k-p))+\gamma_m(2p)) = \gamma_{\frac m 2}(k-p)+\gamma_{\frac m 2}(p) \leq \gamma_m(k-p)+\gamma_m(p),
\]
which entails $\gamma_m(k)^{\frac12} \leq \gamma_m(k-p)^{\frac12}+\gamma_m(p)^{\frac12}$. As a consequence
\[
\gamma_m(k)^{\frac12} |\hat \xi(k)| \leq \sum_{p \in \frac \pi L \ZZ} \gamma_m(k-p)^{\frac 12} |\hat \chi(k-p)| |\hat\xi(p)| + \sum_{p \in \frac \pi L \ZZ}  |\hat \chi(k-p)| \gamma_m(p)^{\frac 12}|\hat\xi(p)| ,
\]
and in view of the fact that $\hat \chi, \gamma_m^{\frac 12}\hat \chi \in \cS(\RR)$, from the first inequality of Lemma \ref{lem:young} the (squared) $L^2(\RR)$-norms of the two terms in the right hand side can be estimated by $C \| \hat \xi\|^2_{\ell^2(\frac\pi L \ZZ)} \leq 2LC \|\xi\|^2_{L,+}$  and by $C\| \gamma_m^{1/2}\hat \xi \|^2_{\ell^2(\frac\pi L \ZZ)} =2L C \|\xi\|^2_{L,+}$ respectively, for a suitable constant $C > 0$. This implies
\[
\| \xi\|^2_{\infty,+} =\frac1{2\pi} \| \gamma_m^{1/2} \hat \xi\|_{L^2(\RR)}^2 \leq \frac{4LC}{\pi} \|\xi\|_{L,+}^2.
\]
Conversely, we can also write, for all $k \in \frac \pi L \ZZ$,
\[\begin{split}
\gamma_m(k)^{\frac12}|\hat \xi(k)| &= \gamma_m(k)^{\frac12}|\widehat{\chi \xi}(k)| \leq \gamma_m(k)^{\frac12} \int_{\RR}dp\, |\hat\chi(k-p) | |\hat \xi(p)|\\
&\leq  \int_{\RR}dp\, \gamma_m(k-p)^{\frac 12} |\hat \chi(k-p)| |\hat\xi(p)| + \int_{\RR} dp\,|\hat \chi(k-p)| \gamma_m(p)^{\frac 12}|\hat\xi(p)| ,
\end{split}\]
and using now the second inequality of the previous Lemma we obtain, arguing as above,
\[
\|\xi\|_{L,+}^2 = \frac1{2L} \|\gamma_m^{1/2} \hat \xi\|_{\ell^2(\frac\pi L \ZZ)}^2 \leq \frac{4\pi C}{L} \| \xi\|_{\infty,+}^2,
\]
which shows the required equivalence of $\|\cdot\|_{L,+}$ and $\|\cdot\|_{\infty,+}$.

In order to prove the equivalence of the norms $\|\cdot\|_{L,-}$, $\|\cdot\|_{\infty,-}$, we start by observing that
\[\begin{split}
\frac{\gamma_m(p)^{1/2}}{\gamma_m(k)^{1/2}} &\leq \!C\!\left[ \frac{1+|p|^2}{1+|k|^2}\right]^{1/4} \!\!\!\leq \!C'\!\left[\frac{1+|p|}{1+|k|}\right]^{1/2} \!\!\!\leq \!C'\!\left[\frac{1+|k|+|p-k|}{1+|k|}\right]^{1/2} \!\!\!\leq \!C'(1+|k-p|)^{1/2},
\end{split}\]
for suitable constants $C, C'>0$, where in the second inequality we used the fact that the function $k \mapsto (1+|k|^2)^{1/2}/(1+|k|)$ is bounded, continuous, non-vanishing on $\RR$ and converges to 1 for $|k| \to +\infty$. We obtain then as above the inequality
\[
\frac1{\gamma_m(k)^{1/2}} |\hat \xi(k)| \leq C'\sum_{p \in \frac\pi L \ZZ} (1+|k-p|)^{1/2}|\hat\chi(k-p)| \frac{1}{\gamma_m(p)^{1/2}}|\hat \xi(p)|,
\]
and the analogous one with the integral over $p \in \RR$ replacing the sum. Using then the fact that $k \mapsto (1+|k|)^{1/2}|\hat \chi(k)|$ is in $\cS(\RR)$ the equivalence of the norms is proven by the same argument based on the previous lemma as above. \hfill$\square$
\end{proof}

\begin{lemma}\label{lem:diff2}
For $ S \Subset (-L ,L)$ there are smooth functions 
$Q_\pm$ on $(-2L,2L)$ such that, for all real functions $\xi, \eta \in \fh_{\infty,L}( S)$
\begin{align}\label{eq:Q-}
\int_\RR dk\, \frac1{\gamma_m(k)} \overline{\hat \xi(k)}\hat \eta(k) - \frac\pi L \sum_{k \in \frac\pi L \ZZ}\frac1{\gamma_m(k)} \overline{\hat \xi(k)}\hat \eta(k)  &= \int_{ S \times  S} dx dy\, Q_-(x-y) \xi(x) \eta(y),
\end{align}
\begin{align}\label{eq:Q+}
\int_\RR dk\, \gamma_m(k) \overline{\hat \xi(k)}\hat \eta(k) - \frac\pi L \sum_{k \in \frac\pi L \ZZ}\gamma_m(k) \overline{\hat \xi(k)}\hat \eta(k)  &= \int_{ S \times  S} dx dy\, Q_+(x-y) \xi(x) \eta(y).
\end{align}
\end{lemma}

\begin{proof}
By \cite[Formula 9.6.21]{AbramowitzStegun}
\[
\int_\RR dk\, \frac1{\gamma_m(k)} \overline{\hat \xi(k)}\hat \eta(k) = 2\int_{ S \times  S}dxdy\,K_0(m|x-y|)\xi(x)\eta(y)
\]
where $K_0(z)$, $z > 0$, is the modified Bessel function of order zero, which is positive, smooth for $z \neq 0$  and satisfies
\[
K_0(z) \leq e^{-z}
\]
for $z$ sufficiently large.  Moreover, the Poisson summation formula claims that
$$\tfrac{\pi}{L}\sum_{k\in\ZZ} \delta\left(x-k\tfrac{\pi}{L}\right)=\sum_{k\in\mathbb Z} e^{i2L  k x}$$
in the distributional sense, and therefore, 
since $k \mapsto  \frac1{\gamma_m(k)} \overline{\hat \xi(k)}\hat \eta(k)$ decreases fast enough for $|k| \to +\infty$ \cite[Corollary VII.~2.6]{SteinWeiss}, 
\[\begin{split}
\tfrac{\pi}{L}\!\!\!\sum_{k \in \frac\pi L \ZZ}\!\!\!\frac{\overline{\hat \xi(k)}\hat \eta(k)}{\gamma_m(k)} &= 
\sum_{n \in \ZZ}\!\ \int\limits_{\RR}\!dk\,\frac{\overline{\hat \xi(k)}\hat \eta(k)}{\gamma_m(k)} e^{-i2Lnk} = 2\!\!\!\int\limits_{ S \times  S}\!\!\!dxdy\,\bigg[\sum_{n \in \ZZ} K_0(m|x-y-2Ln|)\bigg]\xi(x) \eta(y)
\end{split}\]
where the sum inside the integral converges uniformly on $ S \times  S$ thanks to the exponential decay of $K_0$. Therefore, looking at \eqref{eq:Q-} we find the following expression for $Q_-$
\[
Q_-(x) = \sum_{n \in \ZZ\setminus\{0\}} 2 \,K_0(m|x-2Ln|)
\]
The series is term by term smooth in $(-2L,2L)$ since it does not contain the term with $n=0$. Moreover, one can see that the  $s$-th derivative of $K_0$ can be written in terms of $K_0$ and $K_0'$, by the following formulas (cf.~\cite[sec. 9.6]{AbramowitzStegun}) 
\[
K_0'(z) = - K_1(z),\qquad    z^2 K''_0(z)+zK_0(z)'-z^2K_0(z)=0\]
and one gets therefore
\[\qquad K_0^{(s)}(z) = P_{s-2}(1/z) K_0(z) + R_{s-1}(1/z) K_0'(z), \quad s \geq 2,
\]
where $P_s$, $R_s$ are polynomials.  As a consequence, since also $|K_1(z)|\leq e^{-z}$ for $z\rightarrow +\infty$, $K_0^{(s)}$ is exponentially decaying for $z \to +\infty$ too, and therefore the term by term $s$-th derivative of the series defining $Q_-$ is uniformly convergent on $ S\times S$, showing that $Q_-$ is smooth on $(-2L,2L)$.

Concerning~\eqref{eq:Q+}, by the argument in the proof of \cite[Lemma A.5]{ContiMorsella}, there holds
\[
\int_\RR dk\, \gamma_m(k) \overline{\hat \xi(k)}\hat \eta(k) = -\int_{ S \times  S} dxdy\, \frac m{|x-y|} K_1(m|x-y|) \xi(x)\eta(y),
\] 
and therefore, by a similar argument as the above one,
\[
Q_+(x) = -\sum_{n \in \ZZ\setminus\{0\}} \frac m{|x-2Ln|}K_1(m|x-2Ln|)
\]
is smooth in $(-2L,2L)$. \hfill$\square$
\end{proof}

\begin{theorem}\label{thm:quasiequivalence}
Let $K \ge 6$.
For each bounded open set $ S \Subset (-L,L)^d$, the finite volume representation $\pi^{(\infty)}_{L,\infty} \cong \pi_L$ of $\cW_{\infty, L}( S)$ is unitarily equivalent to the infinite volume representation $\pi_{\infty}$.
\end{theorem}
\begin{proof}
Assume $d=1$. Consider the completion of the space of real functions in $\fh_L( S)$ (thought as compactly supported functions on $\RR$) with respect to  the norms $\|\cdot \|_{\infty,\pm}$, and, on these spaces, the operators $Q_\pm$ defined by
\[
\langle \xi, Q_\pm \eta \rangle_{\infty,\pm} = \int_{ S \times  S} dxdy\,Q_\pm(x-y) \xi(x) \eta(y).
\]
Then using Lemma \ref{lem:diff2} and the argument of the proofs of Lemmas A.4 and A.6 of \cite{ContiMorsella}, one verifies that $Q_\pm$
are trace-class, and this, together with Lemma~\ref{lem:normequivalence}, implies the quasiequivalence statement by \cite{ArakiYamagami}.

Now, as we take $K\ge 6$ and the density result (Lemma \ref{lem:density}) is a local property and hence
holds also for $L=\infty$, we have $\pi_{\infty}(\cW_{\infty, L}( S))'' = \pi_{\infty}(\cW_{L}( S))''$, and the latter is a type III factor
\cite{Araki64}.
Then also $\pi_{L,\infty}^{(\infty)}(\cW_{\infty, L}( S))''$ is a type III factor. 
Therefore the last statement follows directly form \cite[Sect. 5.6.6]{DixmierC}.

If $d > 1$, Lemmas~\ref{lem:young} and \ref{lem:normequivalence} and their proofs continue to hold with the obvious modifications. Concerning lemma~\ref{lem:diff2}, the smoothness of the kernels $Q_\pm$ around $x = 0$ can be proven by \cite[IX.46]{ReedMethodsOfModern2}, providing the smoothness of the infinite volume 2-point function in spatial directions away from the origin, and by a Payley-Wiener argument, providing its exponential decay. \hfill$\square$
\end{proof}

\begin{remark}\label{rmk:consistence}
The spatial isomorphism defined in Theorem \ref{thm:quasiequivalence} depends in principle on $ S$ and $L$ with $ S \Subset (-L,L)^d$, and should be denoted by $\theta_{ S,L}$. However for $\xi \in \fh_L( S)$ there holds
$$\theta_{ S, L}(\pi_{L}(W_{\mathrm{ct}}(\xi)))=\pi_{\infty}(W_{\mathrm{ct}}(\xi)),$$
where on the left side $\xi$ is considered as a function on $\TT^d_L$ and on the right side as a  compactly supported function on $\RR^d$. As a consequence, given tori $\TT^d_L$ and $\TT^d_{L'}$ with $L<L'$ and $ S\Subset(-L,L)^d$, then 
$$\theta_{ S,L}(\pi_{L}(\cW_L( S))'')=\theta_{ S,L'}(\pi_{L'}(\cW_{L'}( S))'').$$ 
Similarly, if $ S \subset  S' \Subset (-L,L)^d$ then the consistency property $\theta_{ S',L} \upharpoonright \pi_{L}(\cW_L( S))'' = \theta_{ S,L}$ holds. Because of this, we will use the simplified notation $\theta$ when no confusion can arise.
\end{remark}

In view of the above remark, it also follows that $\theta$ locally intertwines the action of space translations on the torus and on $\RR^d$, namely
\[
{ \theta \Ad U_L(a,0) \upharpoonright \cA_L( S) = \Ad U_\infty(a,0) \,\theta \upharpoonright \cA_L( S),}
\]
where $ S \subset \RR^d$ is a bounded open set, $a\in\RR^d$ and $L$ is large enough that the closures of $ S$ and $ S+a$ are both contained in $(-L,L)^d$.

A similar statement also holds for time translations. Indeed, in view of the fact that the solutions of the Klein-Gordon equation on the torus have propagations speed $c=1$, with $\cB_1(0)$ the unit ball centered at the origin, there holds 
\[
\Ad U_L(0,t) (\cA_L( S)) = \cA_L( S + |t| \cB_1(0))
\] 
whenever $ S$ is a bounded open set, $t\in\RR$ and $L$ is large enough that $ S$ and $ S+t B_1(0)$  are both contained in $(-L,L)^d$. This clearly entails
$$
\theta\,\mathrm{Ad}\, U_L(0,t)\upharpoonright \cA_L( S) = \mathrm{Ad}\,U_\infty(0,t)\, \theta \upharpoonright \cA_L( S).
$$

\section{Other scaling-map schemes}
\label{sec:otherscaling}
From the perspective of the wavelet scaling map of Section \ref{sec:waveletscaling} being a generalization of the block-spin scaling map, we discuss the relations between other scaling map schemes in this section: In Section \ref{sec:blockspinlim}, further details on the block-spin case.
In Section \ref{sec:real}, a singular real-space scheme leading to point-like localization of lattice operators in the scaling limit.
In Section \ref{sec:mom}, a momentum-space scheme implementing sharp momentum cutoff. In Section \ref{sec:irvsuv}, the distinction between our scaling limit construction and the scaling algebra approach by Buchholz and Verch \cite{BuchholzScalingAlgebrasAnd1, BuchholzScalingAlgebrasAnd2}. In Section \ref{sec:MERA}, the connection of our approach with the multi-scale entanglement renormalization ansatz \cite{VidalAClassOf, EvenblyAlgorithmsForEntanglement}.

The second and third of these schemes can be formulated analogously to the wavelet scaling map at the one-particle level as in \eqref{eq:waveletscale} in Definition \ref{def:waveletscaling}. The main difference consists in the use of low-pass filters $\{h_{n}\}_{n\in\ZZ^{d}}$ concentrated at $n=0$ (point-like localization) or with scale-dependent full support, $h_{n} = h^{(N)}_{n}$, $\supp h^{(N)} = \ZZ^{d}_{2r_{N}}$ (sharp momentum cutoff). In contrast with the wavelet scaling map this results in the drawback that these schemes do not directly lead to local algebras in terms of lattice operators in the scaling limit. More specifically,
the block spin renormalization does not have enough regularity to include the momentum operator, 
point-like localization leads to sharply localized field-like objects and the scaling limit of initial states can only be defined in a distributional sense necessitating a Wightman reconstruction, while sharp momentum cutoff leads to a scaling limit in terms of fully non-local algebras also requiring an indirect reconstruction of local operators.

\subsection{The block-spin transformation case}
\label{sec:blockspinlim}
The expression \eqref{eq:freegroundlimit} does not define limit states $\omega^{(N)}_{L,\infty}$ on $\cW_{N,L}$ if the renormalization group is defined by the block-spin scaling map \eqref{eq:blockspinscale}. Nevertheless, it is still possible to interpret \eqref{eq:freegroundlimit} in a distributional sense. To this end, we exploit the fact that the states $\omega_{\mu_{N+M},M} = \omega^{(N)}_{L,M}$ are quasi-free and, thus, are determined by their two-point functions:
\begin{align*}
W^{(N)}_{L,M|\Phi\Phi}(x,y)& = \omega^{(N)}_{L,M}(\alpha^{N}_{N+M}(\vep_{N}^{-\frac{1+d}{2}}\Phi_{N}(x)\vep_{N}^{-\frac{1+d}{2}}\Phi_{N}(y))) \\
 & = \tfrac{1}{2 (2r_{N})^{d}}\hspace{-0.2cm}\sum_{k\in\Gamma_{N+M}}\hspace{-0.2cm}\vep_{N}^{-1}\gamma_{\mu_{N+M}}(k)^{-1} \\[-0.3cm]
 & \hspace{2.75cm}\vep_{N+M}^{2d}\hspace{-0.4cm}\sum_{x',y'\in\Lambda_{N+M}}\hspace{-0.4cm}\vep_{N}^{-d}\chi_{[0,\vep_{N})^{d}}(x'-x)\vep_{N}^{-d}\chi_{[0,\vep_{N})^{d}}(y'-y)e^{ik\cdot(x'-y')}, \\[0.1cm]
W^{(N)}_{L,M|\Pi\Pi}(x,y)& = \omega^{(N)}_{L,M}(\alpha^{N}_{N+M}(\vep_{N}^{\frac{1-d}{2}}\Pi_{N}(x)\vep_{N}^{\frac{1-d}{2}}\Pi_{N}(y))) \\
 & = \tfrac{1}{2(2r_{N})^{d}}\hspace{-0.2cm}\sum_{k\in\Gamma_{N+M}}\hspace{-0.2cm}\vep_{N}\gamma_{\mu_{N+M}}(k) \\[-0.3cm]
 & \hspace{2.75cm}\vep_{N+M}^{2d}\hspace{-0.4cm}\sum_{x',y'\in\Lambda_{N+M}}\hspace{-0.4cm}\vep_{N}^{-d}\chi_{[0,\vep_{N})^{d}}(x'-x)\vep_{N}^{-d}\chi_{[0,\vep_{N})^{d}}(y'-y)e^{ik\cdot(x'-y')}, \\[0.1cm]
W^{(N)}_{L,M|\Phi\Pi}(x,y)& = \omega^{(N)}_{L,M}(\alpha^{N}_{N+M}(\vep_{N}^{-\frac{1+d}{2}}\Phi_{N}(x)\vep_{N}^{\frac{1-d}{2}}\Pi_{N}(y))) \\
 & = \tfrac{i}{2 (2r_{N})^{d}}\hspace{-0.1cm}\sum_{k\in\Gamma_{N+M}}\hspace{-0.1cm}\vep_{N+M}^{2d}\hspace{-0.4cm}\sum_{x',y'\in\Lambda_{N+M}}\hspace{-0.4cm}\vep_{N}^{-d}\chi_{[0,\vep_{N})^{d}}(x'-x)\vep_{N}^{-d}\chi_{[0,\vep_{N})^{d}}(y'-y)e^{ik\cdot(x'-y')},
\end{align*}
But, these expression are well-defined for arbitrary $x,y\in\TT_{L}^{d}$ and consistent with respect to the inclusions $\Lambda_{N}\subset\Lambda_{N'}\subset\TT_{L}^{d}$ for $N<N'$ as a consequence of the scaling equation \eqref{eq:unitscalingeq}. Then, as a consequence of the Paley-Wiener-Schwartz characterization of the (continuum) Fourier transform of functions $f,f'\in C^{\infty}(\TT_{L}^{d})$, the expression
\begin{align*}
W^{(N)}_{L,M|\Phi\Phi}(f,f') & = \tfrac{1}{2 (2r_{N})^{d}}\hspace{-0.2cm}\sum_{k\in\Gamma_{N+M}}\hspace{-0.2cm}\vep_{N}^{-1}\gamma_{\mu_{N+M}}(k)^{-1} \\[-0.35cm]
 & \hspace{2cm}\vep_{N+M}^{2d}\hspace{-0.4cm}\sum_{x',y'\in\Lambda_{N+M}}\hspace{-0.4cm}\vep_{N}^{-d}(\chi_{[0,\vep_{N})^{d}}\ast f)(x')\vep_{N}^{-d}(\chi_{[0,\vep_{N})^{d}}\ast f')(y')e^{ik\cdot(x'-y')}, \\
W^{(N)}_{L,M|\Pi\Pi}(f,f') & = \tfrac{1}{2 (2r_{N})^{d}}\hspace{-0.2cm}\sum_{k\in\Gamma_{N+M}}\hspace{-0.2cm}\vep_{N}\gamma_{\mu_{N+M}}(k) \\[-0.35cm]
 & \hspace{2cm}\vep_{N+M}^{2d}\hspace{-0.4cm}\sum_{x',y'\in\Lambda_{N+M}}\hspace{-0.4cm}\vep_{N}^{-d}(\chi_{[0,\vep_{N})^{d}}\ast f)(x')\vep_{N}^{-d}(\chi_{[0,\vep_{N})^{d}}\ast f')(y')e^{ik\cdot(x'-y')}, \\
W^{(N)}_{L,M|\Phi\Pi}(f,f') & = \tfrac{i}{2 (2r_{N})^{d}}\hspace{-0.2cm}\sum_{k\in\Gamma_{N+M}}\hspace{-0.2cm}\vep_{N+M}^{2d}\hspace{-0.4cm}\sum_{x',y'\in\Lambda_{N+M}}\hspace{-0.4cm}\vep_{N}^{-d}(\chi_{[0,\vep_{N})^{d}}\ast f)(x')\vep_{N}^{-d}(\chi_{[0,\vep_{N})^{d}}\ast f')(y')e^{ik\cdot(x'-y')},
\end{align*}
admit well-defined limits for $M\rightarrow\infty$ (by a Riemann-sum argument):
\begin{align*}
W^{(N)}_{L,\infty|\Phi\Phi}(f,f') & = \tfrac{1}{2(2L)^{d}}\hspace{-0.2cm}\sum_{k\in\frac{\pi}{L}\ZZ^{d}}\hspace{-0.2cm}\vep_{N}^{-1+d}\gamma_{m}(k)^{-1} \\[-0.35cm]
 & \hspace{2cm}\int_{\RR^{2d}}\hspace{-0.4cm}d^{d}x'd^{d}y'\vep_{N}^{-d}(\chi_{[0,\vep_{N})^{d}}\ast f)(x')\vep_{N}^{-d}(\chi_{[0,\vep_{N})^{d}}\ast f')(y')e^{ik\cdot(x'-y')}, \\
W^{(N)}_{L,\infty|\Pi\Pi}(f,f') & = \tfrac{1}{2 (2L)^{d}}\hspace{-0.2cm}\sum_{k\in\frac{\pi}{L}\ZZ^{d}}\hspace{-0.2cm}\vep_{N}^{1+d}\gamma_{m}(k) \\[-0.35cm]
 & \hspace{2cm}\int_{\RR^{2d}}\hspace{-0.4cm}d^{d}x'd^{d}y'\vep_{N}^{-d}(\chi_{[0,\vep_{N})^{d}}\ast f)(x')\vep_{N}^{-d}(\chi_{[0,\vep_{N})^{d}}\ast f')(y')e^{ik\cdot(x'-y')}, \\
W^{(N)}_{L,\infty|\Phi\Pi}(f,f') & = \tfrac{i}{2 (2L)^{d}}\hspace{-0.2cm}\sum_{k\in\frac{\pi}{L}\ZZ^{d}}\hspace{-0.2cm}\vep_{N}^{d}\hspace{-0.15cm}\int_{\RR^{2d}}\hspace{-0.2cm}d^{d}x'd^{d}y'\vep_{N}^{-d}(\chi_{[0,\vep_{N})^{d}}\ast f)(x')\vep_{N}^{-d}(\chi_{[0,\vep_{N})^{d}}\ast f')(y')e^{ik\cdot(x'-y')}.
\end{align*}
Finally, we can take the limit $N\rightarrow\infty$, and we obtain the scaling limit of the two-point functions (after adjusting the scaling dimension),
\begin{align}
\label{eq:cyl2ptff}
W^{(\infty)}_{L,\infty|\Phi\Phi}(f,f') & = \lim_{N\rightarrow\infty}\vep_{N}^{1-d}W^{(N)}_{m,\infty|\Phi\Phi}(f,f') = \tfrac{1}{2(2L)^{d}}\hspace{-0.2cm}\sum_{k\in\frac{\pi}{L}\ZZ^{d}}\hspace{-0.2cm}\gamma_{m}(k)^{-1}\hat{f}(-k)\hat{f'}(k),\\
\label{eq:cyl2ptmm}
W^{(\infty)}_{L,\infty|\Pi\Pi}(f,f') & = \lim_{N\rightarrow\infty}\vep_{N}^{-(1+d)}W^{(N)}_{m,\infty|\Pi\Pi}(f,f') = \tfrac{1}{2(2L)^{d}}\hspace{-0.2cm}\sum_{k\in\frac{\pi}{L}\ZZ^{d}}\hspace{-0.2cm}\gamma_{m}(k)\hat{f}(-k)\hat{f'}(k), \\
\label{eq:cyl2ptfm}
W^{(\infty)}_{L,\infty|\Phi\Pi}(f,f') & = \lim_{N\rightarrow\infty}\vep_{N}^{-d}W^{(N)}_{m,\infty|\Phi\Pi}(f,f') = \tfrac{i}{2(2L)^{d}}\hspace{-0.2cm}\sum_{k\in\frac{\pi}{L}\ZZ^{d}}\hspace{-0.2cm}\hat{f}(-k)\hat{f'}(k),
\end{align}
as distributions over $C^{\infty}(\TT_{L}^{d})$ (with respect to the weak* topology), which are the Wightman two-point functions $W_{L}$ of the free vacuum state $\omega_{L}$ of continuum scalar field on $\TT_{L}^{d}$.

In summary, we find that, although the block-spin scaling map does not define a scaling limit on the algebra $\cW_{\infty, L}$, it is still possible to recover the vacuum state of mass $m$ of the continuum scalar field via its Wightman two-point functions $W^{(\infty)}_{\infty,L}$ defined as scaling limits of the two-point functions $W^{(N)}_{L,M}$ associated with $\omega^{(N)}_{L,M}$ on $\cW_{N,L}$. We state this as a theorem:
\begin{theorem}
\label{thm:blockspinscale}
Given a sequence of lattice ``masses'', $\{\mu_{N}\}_{N}$, satisfying \eqref{eq:Latmassscaling} for some $m>0$ and the block-spin scaling map, the sequences of two-point functions $\{W^{(N)}_{L,M}\}_{M\in\NN_{0}}$ associated with the states $\{\omega^{(N)}_{L,M}\}_{M\in\NN_{0}}$ on $\cW_{N,L}$ for all scales $N\in\NN_{0}$, converge in a scale-coherent way to the free (time-zero) two-point functions of mass $m$ of the continuum scalar field $(\cW_{L},\omega_{L})$ (see Section \ref{subsec:cont}).
\end{theorem}
\begin{proof}
It is well-known that \eqref{eq:cyl2ptff} -- \eqref{eq:cyl2ptfm} are the free (time-zero) two-point functions of mass $m$ over $C^{\infty}(\TT_{L}^{d})$. \hfill$\square$
\end{proof}

\subsection{Real-space renormalization: Point-like localization}
\label{sec:real}
The point-localization renormalization group is a precise formulation of the idea that the algebras $\fA_{N}(x)$, $x\in\Lambda_{N}$, correspond to sharply localized operators at $x\in\RR^{d}$ in the continuum. The definition of point-like localized operators is achieved at the level of lattice fields $\Phi_{N}$ and momenta $\Pi_{N}$ by the natural inclusion of lattices $\Lambda_{N}\subset\Lambda_{N'}$ for $N<N'$ as subset of $\RR^{d}$.

\begin{definition}
\label{def:pointlikescaling}
The point-localization renormalization group $\{\alpha^{N}_{N'}\}_{N<N'}$ is the inductive family of $^*$-linear maps defined by the point-localization scaling map between one-particle Hilbert spaces:
\begin{align*}
R^{N}_{N+1} : \fh_{N,L} & \longrightarrow \fh_{N+1,L},
\end{align*}
where
\begin{align}
\label{eq:pointlikescale}
R^{N}_{N+1}(q,p)(x') & = 2^{d}(q,p)(x')\chi_{\Lambda_{N}\subset\Lambda_{N+1}}(x'), & N & \in\NN_{0}
\end{align}
and
\begin{align*}
R^{N}_{N'} & = R^{N'-1}_{N'}\circ R^{N'-2}_{N'-1}\circ ... \circ R^{N}_{N+1}, & N & < N'.
\end{align*}
$\Lambda_{N}\subset\Lambda_{N+1}$ is the canonical inclusion as subsets of $\RR^{d}$, and $\chi_{\Lambda_{N}\subset\Lambda_{N+1}}$ is the corresponding characteristic function. The right-hand side of \eqref{eq:pointlikescale} is to be understood as an extension of $q,p$ by the zero-function on $\Lambda_{N+1}\setminus\Lambda_{N}$.
\end{definition}
Applying the lattice Fourier transform and its inverse, we find the point-localization scaling map is given by periodic extension in momentum space:
\begin{align*}
R^{N}_{N+1}(\hat{q},\hat{p})(k') & = \vep_{N+1}^{\frac{d}{2}}\!\!\!\sum_{x'\in\Lambda_{N+1}}\!\!\!2^{d}(q,p)(x')\chi_{\Lambda_{N}\subset\Lambda_{N+1}}(x')e^{-ik'\cdot x'} 
 = 2^{\frac{d}{2}}(\hat{q},\hat{p})(k').
\end{align*}
In contrast with the previous scaling maps and as a consequence of the rescaling by $2^{d}$ in \eqref{eq:pointlikescale}, we find that $R^{N}_{N'}$ is only a symplectic homothety for $N<N'$:
\begin{align*}
\sigma_{N'}\circ(R^{N}_{N'}\times R^{N}_{N'}) & = 2^{d(N'-N)}\sigma_{N}.
\end{align*}
Therefore, the induced map $\alpha^{N}_{N'}:\cW_{N,L}\rightarrow\cW_{N',L}$ is a $^*$-compatible linear map but not a $^*$-homomorphism. Nevertheless, we can rescale the symplectic form $\sigma_{N}$ determining the Weyl relations to show that for a given $N\in\NN_{0}$ the pullbacks of ground states $\omega_{\mu_{N'}}$ along $\alpha^{N}_{N'}$ for $N<N'$, which we use to define the scaling limit, yield a sequence of states on $\cW_{N,L}$.

\begin{lemma}
\label{lem:realscale}
Given a symplectic homothety $R^{N}_{N'}:\fh_{N,L}\rightarrow\fh_{N',L}$ for $N<N'$, the induced $^*$-linear map $\alpha^{N}_{N'}:\cW_{N,L}\rightarrow\cW_{N',L}$ becomes a $^*$-homomorphism after rescaling the symplectic form according to:
\begin{align*}
\tilde{\sigma}_{N'} & = \lambda_{N,N'}^{-1}\sigma_{N'},
\end{align*}
where $\lambda_{N,N'}\neq0$ is the scaling factor of $R^{N}_{N'}$, i.e. $\sigma_{N'}\circ(R^{N}_{N'}\times R^{N}_{N'}) = \lambda_{N,N'}\sigma_{N}$
\end{lemma}
\begin{proof}
We only need to check the compatibility of $R^{N}_{N'}$ with the rescaled symplectic form $\tilde{\sigma}_{N'}$ as this implies the compatibility of $\alpha^{N}_{N'}$ with the Weyl relations \eqref{eq:CCR}:
\begin{align*}
\tilde{\sigma}_{N'}\circ(R^{N}_{N'}\times R^{N}_{N'}) & = \lambda_{N,N'}^{-1}\sigma_{N'}\circ(R^{N}_{N'}\times R^{N}_{N'}) = \lambda_{N,N'}^{-1}\lambda_{N,N'}\sigma_{N} = \sigma_{N}.
\end{align*}
\hfill$\square$
\end{proof}
We also recall a general result of Segal about states on Weyl algebras.
\begin{remark}[Scaled Fock states]
\label{rem:scaledfock}
We know from sections \ref{sec:latscalar} and \ref{sec:latvac}, that the symplectic structure $\sigma_{N}$ can be expressed as:
\begin{align*}
\sigma_{N}((q,p),(q',p')) & = -2\Im(\hat{\xi}^{(m)},\hat{\xi'}^{(m)})_{N}.
\end{align*}
Moreover, since $\fh_{N,L}\cong\CC^{2r_{N}}$ is finite dimensional, we know that
\begin{align*}
\left\|\hat{\xi}^{(m)}\right\|^{2}_{N} & = \tfrac{1}{2(2r_{N})^{d}}\sum_{k\in\Gamma_{N}}\left(\gamma^{(N)}_{m}(k)^{-1}|\hat{q}(k)|^{2} + \gamma^{(N)}_{m}(k)|\hat{p}(k)|^{2}\right), & \xi & \in\fh_{N,L},
\end{align*}
defines non-degenerate (for $m>0$) norm equivalent to $\|\ \cdot \ \|_{N}$. Thus, by a general theorem of Segal \cite{SegalMathematicalCharacterizationOf}, we have a one-parameter family of regular states on the Weyl algebra $\cW_{N,L}$:
\begin{align*}
\omega^{(\lambda_{N})}_{m,0}(W_{N}(\xi)) & = e^{-\frac{1}{2}\lambda_{N}\left\|\hat{\xi}^{(m)}\right\|^{2}_{N}}, & \lambda_{N} & \geq1.
\end{align*}
The lattice vacuum $\omega_{\mu_{N}}$ results from the extremal choice $\lambda_{N}=1$ and $\mu_{N}^{2} = \vep_{N}^{2}m^{2}+2d$, cp.~\eqref{eq:freeground} and \eqref{eq:freegroundmass}.
\end{remark}

From the lemma and the remark we infer the result: 
\begin{corollary}
\label{cor:realscale}
For $N,M\in\mathds{N}_{0}$, the pullback $\omega_{\mu_{N+M},M} = \omega_{\mu_{N+M}}\circ\alpha^{N}_{N+M}$ defines state on $\cW_{N,L}$ explicitly given by:
\begin{align*}
\omega_{\mu_{N+M}}(\alpha^{N}_{N+M}(W_{N}(\xi))) & = e^{-\frac{1}{4(2r_{N})^{d}}\sum_{k'\in\Gamma_{N+M}}\left(\gamma_{\mu_{N+M}}(k')^{-1}|\hat{q}(k')|^{2} + \gamma_{\mu_{N+M}}(k')|\hat{p}(k')|^{2}\right)},
\end{align*}
where by a slight abuse of notation $\hat{q}$ and $\hat{p}$ also denote their periodic extensions from $\Gamma_{N}$ to $\Gamma_{N+M}$.
\end{corollary}
\begin{proof}
Given $N,M\in\NN_{0}$, the lattice vacuum $\omega_{\mu_{N+M}}$ defines a scaled Fock state $\omega_{m,0}^{(\lambda_{N,N+M})}$ on $\cW_{N+M,L}$ with rescaled symplectic form $\tilde{\sigma}_{N+M} = \lambda_{N,N+M}^{-1}\sigma_{N+M}$, where $\lambda_{N,N+M} = 2^{dM}$ and $\mu_{N+M}^{2} = \vep_{N+M}^{2}m^{2}+2d$. \hfill$\square$
\end{proof}
\begin{remark}
\label{rem:deltascaling}
In view of the general wavelet scaling map scheme introduced in Section \ref{sec:waveletscaling}, the point-localization renormalization group arises from the extremal choice $h_{0}=2^{\frac{d}{2}}\delta_{n,0}$ as low-pass filter. That this choice should correspond to point-like localized field operators in the scaling limit is suggested by $\hat{\phi}(\!\ \cdot \!\ )=\prod^{\infty}_{N=1}m_{0}(2^{-N}k)=1$ which is the inverse Fourier transform of the delta distribution: $\phi^{(\vep)}_{0}=\vep^{\frac{d}{2}}\delta_{0}$. Therefore, we expect that the local operators in the continuum are only indirectly accessible through a Wightman-type reconstruction from correlation functions.
\end{remark}
In analogy with the results for the scaling limit in terms of block-spin transformation, the sequences $\{\omega_{\mu_{N+M},M} = \omega^{(N)}_{L,M}\}_{M\in\NN_{0}}$ for $N\in\NN_{0}$ do not converge as we can see from \eqref{eq:freegroundlimit} and the fact that $\hat{\phi}^{(\vep_{N})} = \vep_{N}^{\frac{d}{2}}$ according to remark \ref{rem:deltascaling}. But, we can use corollary \ref{cor:realscale} to compute the following two-point functions of the time-zero fields $\Phi_{N}$ and momenta $\Pi_{N}$ at level $N$, i.e. for any $x,y\in\Lambda_{N}$:
\begin{align*}
W^{(N)}_{L,M|\Phi\Phi}(x,y)& \!=\! \omega^{(N)}_{L,M}(\!\alpha^{N}_{N+M}(\!\vep_{N}^{-\frac{1+d}{2}}\!\Phi_{N}(x)\vep_{N}^{-\frac{1+d}{2}}\!\Phi_{N}(y)\!)\!) 
 \!=\! \tfrac{1}{2 (2r_{N})^{d}}\hspace{-0.2cm}\sum_{k\in\Gamma_{N+M}}\hspace{-0.25cm}\vep_{N}^{-1}\gamma_{\mu_{N+M}}(k)^{-1}e^{ik\cdot(x-y)}, \\[0.1cm]
W^{(N)}_{L,M|\Pi\Pi}(x,y)& \!=\! \omega^{(N)}_{L,M}(\!\alpha^{N}_{N+M}(\!\vep_{N}^{\frac{1-d}{2}}\Pi_{N}(x)\vep_{N}^{\frac{1-d}{2}}\Pi_{N}(y)\!)\!)
 \!=\! \tfrac{1}{2(2r_{N})^{d}}\hspace{-0.2cm}\sum_{k\in\Gamma_{N+M}}\hspace{-0.2cm}\vep_{N}\gamma_{\mu_{N+M}}(k)e^{ik\cdot(x-y)} \\[0.1cm]
W^{(N)}_{L,M|\Phi\Pi}(x,y)& \!=\! \omega^{(N)}_{L,M}(\!\alpha^{N}_{N+M}(\!\vep_{N}^{-\frac{1+d}{2}}\Phi_{N}(x)\vep_{N}^{\frac{1-d}{2}}\Pi_{N}(y)\!)\!) 
 \!=\! \tfrac{i}{2 (2r_{N})^{d}}\hspace{-0.2cm}\sum_{k\in\Gamma_{N+M}}\hspace{-0.3cm}e^{ik\cdot(x-y)},
\end{align*}
Since $\Gamma_{N}\stackrel{N\rightarrow\infty}{\longrightarrow}\tfrac{\pi}{L}\ZZ^{d}$ and comparing with the calculations in Section \ref{sec:blockspinlim}, we obtain a result analogous to Theorem \ref{thm:blockspinscale}:
\begin{theorem}
\label{thm:pointlikescale}
Given a sequence of lattice ``masses'', $\{\mu_{N}\}_{N\in\NN_{0}}$, satisfying \eqref{eq:Latmassscaling} for some $m>0$ and the point-localization scaling map, the sequences of two-point functions, $\{W^{(N)}_{L,M}\}_{M\in\NN_{0}}$, associated with the states $\{\omega^{(N)}_{L,M}\}_{M\in\NN_{0}}$ on $\cW_{N,L}$ for all scales $N\in\NN_{0}$, converge in a scale-coherent way to the free (time-zero) two-point functions of mass $m$ of the continuum scalar field $(\cW_{L},\omega_{L})$ (see Section \ref{subsec:cont}).
\end{theorem}
\begin{proof}
The proof works almost identical to that for the block-spin scaling map. The only notable difference is that the convolution with the characteristic function $\vep_{N}^{-d}\chi_{[0,\vep_{N})^{d}}$ is replaced by a convolution with the delta distribution $\delta_{0}$. \hfill$\square$
\end{proof}

\subsection{Momentum-space renormalization: Sharp momentum cutoff}
\label{sec:mom}
Since renormalization group schemes in momentum space have a long tradition \cite{Wilson-71-Renormalization2, FisherTheRenormalizationGroup, WilsonTheRenormalizationGroupKondo, ZinnJustinPhaseTransitionsAnd}, it is sensible to connect these to the general framework of Section \ref{sec:oaren}. While the point-localization scaling map in the previous subsection is induced by the natural inclusion of lattices, $\Lambda_{N}\subset\Lambda_{N'}$, in real space, the momentum-cutoff renormalization group is obtained from the inclusion of lattices, $\Gamma_{N}\subset\Gamma_{N'}$, in momentum space. This is made precise in the following definition and precisely corresponds to the idea of taking a partial trace with respect to high-momentum states when coarse graining from states at a fine scale $N'$ to a coarse scale $N$. Since the momentum cutoff has been used extensively in constructive quantum field theory \cite{GlimmQuantumFieldTheory}, we also discuss the limit of dynamics in this setting.

\begin{definition}
\label{def:momcutoff}
The momentum-cutoff renormalization group $\{\alpha^{N}_{N'}\}_{N<N'}$ is the inductive family of $^*$-homomorphisms defined by the momentum-cutoff scaling map between one-particle Hilbert spaces:
\begin{align*}
R^{N}_{N+1} : \fh_{N,L} & \longrightarrow \fh_{N+1,L},
\end{align*}
where
\begin{align}
\label{eq:momcutoffscale}
R^{N}_{N+1}(\hat{q},\hat{p})(k') & = 2^{\frac{d}{2}}(\hat{q},\hat{p})(k')\chi_{\Gamma_{N}\subset\Gamma_{N+1}}(k'), & N & \in\NN_{0}
\end{align}
and
\begin{align*}
R^{N}_{N'} & = R^{N'-1}_{N'}\circ R^{N'-2}_{N'-1}\circ ... \circ R^{N}_{N+1}, & N & < N'.
\end{align*}
$\Gamma_{N}\subset\Gamma_{N+1}$ is the canonical inclusion as subsets of $\tfrac{\pi}{L}\ZZ^{d}$, and $\chi_{\Gamma_{N}\subset\Gamma_{N+1}}$ is the corresponding characteristic function. The right-hand side of \eqref{eq:momcutoffscale} is to be understood as an extension of $\hat{q},\hat{p}$ by the zero-function on $\Gamma_{N+1}\setminus\Gamma_{N}$.
\end{definition}
As a trivial consequence of the rescaling by $2^{\frac{d}{2}}$, we have that $R^{N}_{N'}$ is symplectic for any $N<N'$:
\begin{align*}
\sigma_{N'}\circ(R^{N}_{N'}\times R^{N}_{N'}) & = \sigma_{N},
\end{align*}
Applying the lattice Fourier transform and its inverse, we find the momentum-cutoff scaling map is given by Fourier interpolation in real space:
\begin{align*}
R^{N}_{N+1}(q,p)(x') & = \vep_{N+1}^{-\frac{d}{2}}(2r_{N+1})^{-d}\sum_{k'\in\Gamma_{N+1}}2^{\frac{d}{2}}(\hat{q},\hat{p})(k')\chi_{\Gamma_{N}\subset\Gamma_{N+1}}(k')e^{ik'\cdot x'} \\
 & = \vep_{N}^{-\frac{d}{2}}(2r_{N})^{-d}\sum_{k\in\Gamma_{N}}(\hat{q},\hat{p})(k)e^{ik\cdot x'} = \sum_{x\in\Lambda_{N}}(q,p)(x)(2r_{N})^{-d}\sum_{k\in\Gamma_{N}}e^{-ik\cdot (x-x')}.
\end{align*}
The locality of the momentum-cutoff scaling map entails non-locality in real space, which is formally reflected by the use of completely delocalized smearing functions to relate the lattice fields and momenta with their continuum analogues, cp. \eqref{eq:blockfield} and \eqref{eq:blockmom}:
\begin{align}
\label{eq:ftblockfield}
\hat{\Phi}_{N}(k) & = \vep_{N}^{-\frac{d}{2}}\Phi(\sum_{x\in\Lambda_{N}}\chi_{x+[0,\vep_{N})^{d}}e^{-ik\cdot x}), & \hat{\Pi}_{N}(k) & = \vep_{N}^{-\frac{d}{2}}\Pi(\sum_{x\in\Lambda_{N}}\chi_{x+[0,\vep_{N})^{d}}e^{-ik\cdot x}).
\end{align}
\begin{remark}
\label{rem:momentumrg1}
We can translate the momentum-cutoff scaling map into a real-space scaling map by the lattice Fourier transform. This real-space scaling map can be cast into a form analogous to the wavelet scaling map \eqref{eq:waveletscale} with non-local and level-dependent coefficients:
\begin{align*}
h^{(N)}_{n} & = \left(2^{\frac{1}{2}}r_{N+1}\right)^{-d}\prod^{d}_{j=1}\frac{\sin\left(\frac{\pi}{2}n_{j}\right)}{\sin\left(\frac{\pi}{2r_{N+1}}n_{j}\right)}e^{\frac{i\pi n_{j}}{2r_{N+1}}}, &\!\! n\in\{-r_{N+1},...,r_{N+1}-1\}^{d}\subset\ZZ^{d}.
\end{align*}
While such a momentum space renormalization group can be used to determine the scaling limit as we show below, the non-locality of the filter makes it unefficient for controlling local operators.
\end{remark}
The scaling limit of the free lattice vacua \eqref{eq:freeground} can be investigated by the momentum-cutoff renormalization group in a manner similar to that of the wavelet method of Section \ref{sec:groundstate}:
\begin{align}
\label{eq:freegroundscalingmom}
\omega^{(N)}_{L,M}(W_{N}(\xi)) & = \omega_{\mu_{N+M},M}(W_{N}(\xi)) = e^{-\frac{1}{4}\left(\left|\left|\gamma_{\mu_{N+M}}^{-1/2}\hat{q}_{N}\right|\right|_{N,L}^{2}+\left|\left|\gamma_{\mu_{N+M}}^{1/2}\hat{p}_{N}\right|\right|_{N,L}^{2}\right)},
\end{align}
where, for fixed $N\in\NN_{0}$, the dispersion relation $\gamma_{\mu_{N+M}}$ is considered as a function on $\Gamma_{N}\subset\Gamma_{N+M}$ by restriction.
Again, choosing the initial states according to the renormalization condition \eqref{eq:Latmassscaling} leads to the continuum dispersion relation in the limit $M\rightarrow\infty$. Moreover, the pointwise convergence to the continuum dispersion relation implies the weak* convergence of \eqref{eq:freegroundscalingmom} to a limit $\omega^{(N)}_{L,\infty}$ which is given by the free vacuum state $\omega_{L}$ of mass $m$ of the continuum scalar field evaluated on Weyl operators with a sharp momentum cut-off at scale $k_{\max}\sim\vep_{N}^{-1}$:
\begin{align}
\label{eq:freegroundlimmom}
\omega^{(N)}_{L,\infty}(W_{N}(\xi)) & = e^{-\frac{1}{4}\left(\left|\left|\gamma_{m}^{-1/2}\hat{q}_{N}\right|\right|_{N,L}^{2}+\left|\left|\gamma_{m}^{1/2}\hat{p}_{N}\right|\right|_{N,L}^{2}\right)}.
\end{align}
Now, it is almost trivial to check that the sequence of limit states, $\{\omega^{(N)}_{L,\infty}\}_{N\in\NN_{0}}$, is projectively consistent for the momentum-cutoff renormalization group,
\begin{align*}
\omega^{(N)}_{L,\infty} & = \omega^{(N')}_{L,\infty}\circ\alpha^{N}_{N'}, & N<N', 
\end{align*}
which implies the existence of the scaling limit $\cW_{\infty, L}$:
\begin{align*}
\varprojlim_{N\in\mathds{N}_{0}}\omega^{(N)}_{L,\infty} & = \omega^{(\infty)}_{L,\infty}.
\end{align*}
In this way, we deduce the following result corresponding to theorems \ref{thm:localalgebras} and \ref{thm:blockspinscale}.
\begin{theorem}
\label{thm:momscale}
The GNS representation of $\cW_{\infty, L}$ given by the scaling limit $\omega^{(\infty)}_{L,\infty}$ constructed via the momentum-cutoff renormalization group is equivalent to the GNS representation of the continuum scalar field with respect to the free vacuum state $\omega_{L}$ restricted to the subalgebra generated by Weyl operators with finite support in momentum space.
\end{theorem}
\begin{proof}
By construction the inner product on the GNS Hilbert space $\cH^{(\infty)}_{L,\infty}$ is given in terms of the inner product of the inductive limit of one-particle Hilbert spaces $\fh_{\infty,L}$ with dense subspace $\bigcup_{N\in\mathds{N}_{0}}\fh_{N,L}$:
\begin{align}
\label{eq:momcutoffGNS}
 \omega^{(\infty)}_{L,\infty}(\alpha^{N}_{\infty}(W_{N}(\xi))^{*}\alpha^{N'}_{\infty}(W_{N'}(\xi'))) & = e^{-\frac{1}{2}\|\xi\|^{2}}e^{-\frac{1}{2}\|\xi'\|^{2}}e^{(\xi,\xi')} = (G(\xi),G(\xi')),
\end{align}
for $N,N'\in\mathds{N}_{0}$ with
\begin{align}
\label{eq:momcutoff1part}
(\xi,\xi') & = \tfrac{2^{-(d+1)}}{\sqrt{r_{N}r_{N'}}^{d}}\hspace{-0.25cm}\sum_{k\in\frac{\pi}{L}\mathds{Z}^{d}}\hspace{-0.15cm}\left(\hspace{-0.05cm}\tfrac{\overline{\hat{q}_{\xi}(k)}\hat{q}_{\xi'}(k)}{\gamma_{m}(k)}\hspace{-0.075cm}+\hspace{-0.075cm}\gamma_{m}(k)\overline{\hat{p}_{\xi}(k)}\hat{p}_{\xi'}(k)\hspace{-0.075cm}+\hspace{-0.075cm}i(\overline{\hat{p}_{\xi}(k)}\hat{q}_{\xi'}(k)\hspace{-0.075cm}-\hspace{-0.075cm}\overline{\hat{q}_{\xi}(k)}\hat{p}_{\xi'}(k)\!)\hspace{-0.075cm}\!\right),
\end{align}
where we slightly abuse notation to denote by $\hat{p}_{\xi}$, $\hat{q}_{\xi}$, $\hat{q}_{\xi'}$ and $\hat{p}_{\xi'}$ also their trivial extensions to the infinite lattice $\Gamma_{\infty}=\tfrac{\pi}{L}\mathds{Z}^{d}$ (cp.~\eqref{eq:normxiN}). On right-hand side of \eqref{eq:momcutoffGNS}, $G(\xi) = \exp(\xi)$ is the Glauber vector or coherent state associated with $\xi\in\bigcup_{N\in\mathds{N}_{0}}\fh_{N,L}$, and it is know that the map $G:\bigcup_{N\in\mathds{N}_{0}}\fh_{N,L}\rightarrow\cH^{(\infty)}_{L,\infty}$ is continuous and injective \cite{HoneggerPhotonsInFock, GuichardetSymmetricHilbertSpaces}. Now, \eqref{eq:momcutoffGNS} also shows that the GNS inner product of $\cH^{(\infty)}_{L,\infty}$ is identical to that of the symmetric Fock space, $\fF_{+}(\fh_{\infty,L})$, over $\fh_{\infty,L}$ endowed with \eqref{eq:momcutoff1part}. Because of this we define an embedding of $\iota_{L}:\fh_{\infty,L}\rightarrow\fh_{L}$ via the inductive-limit structure of $\fh_{\infty,L}$ analogous to \eqref{eq:jLN}. Precisely, we define isometries $\iota^{(N)}_{L}:\fh_{N,L}\rightarrow\fh_{L}$ for all $N\in\NN_{0}$:
\begin{align*}
\iota^{(N)}_{L}(\hat{q},\hat{p}) & = \vep_{N}^{\frac{d}{2}}(\hat{q},\hat{p})\chi_{\Gamma_{N}\subset\Gamma_{\infty}}.
\end{align*}
Since we have the identity $\iota^{(N')}_{L}\circ R^{N}_{N'} = \iota^{(N)}_{L}$ for any $N<N'$, we obtain by standard reasoning the sought after $\iota_{L}$. From the definition of $\iota^{(N)}_{L}$ its is clear that $\iota_{L}(\fh_{\infty,L})$ consist of the closure of the subspace of Fourier series with finite support. Therefore, $\iota_{L}(\fh_{\infty,L}) = \fh_{L}$ and we have a unitary equivalence of $U_{L}:\fF_{+}(\fh_{\infty,L})\rightarrow\fF_{+}(\fh_{L})$ by functoriality. Similarly, functoriality and the isometricity of $\iota^{(N)}_{L}$ for each $N\in\NN_{0}$, implies the existence of an injective $^*$-homomorphism $\Phi_{L}:\cW_{\infty, L}\rightarrow\cW_{L}$ such that 
\begin{align*}
\pi^{(\infty)}_{L,\infty}(W) & = U_{L}^{*}\pi_{\omega_{L}}(\Phi_{L}(W)) U_{L}, & \forall W&\in\cW_{\infty, L}.
\end{align*}
\hfill$\square$
\end{proof}
To conclude our discussion about the momentum-cutoff scaling map, we also take a look at the limit of dynamics. We observe that this limit is considerably simpler with respect to this scheme compared to that in the wavelet scheme (see Section \ref{sec:limdyn}), i.e. the limit of dynamics exist even in the sense of Hamiltonians as depicted in figure \ref{fig:trianglerg}. \\[0.1cm]
At the level of the one-particle space $\fh_{N,L}\cong\ltwo(\Gamma_{N},(2r_{N})^{-d}\mu_{\Gamma_{N}})$, the action of $H^{(N)}_{L,0}$ is (essentially) given by multiplication with the dispersion relation $\gamma_{\mu_{N}}$ (cp.~remark \ref{rem:freedynfock}). Since the momentum-cutoff scaling map is (up to rescaling) induced by the inclusion of lattices $\Gamma_{N}\subset\Gamma_{N'}$ for $N<N'$, the action of $H^{(N')}_{L,0}$ restricts consistently to the image of $\cH_{N,L}$ inside $\cH_{N',L}$ for $N<N'$ because the aforesaid multiplication preserves the support in momentum space. This observation leads to the following result.
\begin{theorem}
\label{thm:limdynmom}
Given the momentum-cutoff renormalization group $\{\alpha^{N}_{N'}\}_{N<N'}$ and a sequence of lattice ``masses'' $\{\mu_{N}\}_{N\in\NN_{0}}$ satisfying condition \eqref{eq:Latmassscaling} with a physical mass $m>0$, we obtain a sequence of (self-adjoint) renormalized Hamiltonians $\{H^{(N)}_{L,M}\}_{M\in\NN_{0}}$ on $\cH_{N,L}$ for each $N\in\NN_{0}$. Moreover, each sequence converges in the norm resolvent sense,
\begin{align*}
H^{(N)}_{L,M} & \stackrel{\textup{nR}}{\longrightarrow} H^{(N)}_{L,\infty},
\end{align*}
to a self-adjoint operator $H^{(N)}_{L,\infty}$ and the sequence $\{H^{(N)}_{L,\infty}\}_{N\in\NN_{0}}$ defines a self-adjoint inductive limit,
\begin{align*}
\varinjlim_{N\in\NN_{0}}H^{(N)}_{L,\infty} & = H^{(\infty)}_{L,\infty},
\end{align*}
on the GNS Hilbert space $\cH^{(\infty)}_{L,\infty}$. The Hamiltonian $H^{(\infty)}_{L,\infty}$ of the scaling limit $\omega^{(\infty)}_{L,\infty}$ is identified with the free Hamiltonian $H_{L}$ of mass $m$ on $\cH_{L}$ by theorem \ref{thm:momscale}.
\end{theorem}
\begin{proof}
We realize the GNS Hilbert space $\cH^{(\infty)}_{L,\infty}$ as an inductive limit of a sequence of the Hilbert spaces $\{\cH^{(\infty)}_{N,L}\}_{N\in\NN_{0}}$ by applying the GNS construction with respect to the restriction $\omega^{(\infty)}_{L,\infty}$ to each $\cW_{N,L}$. By construction, the inductive system of isometries $V^{N}_{N'}:\cH^{(\infty)}_{N,L}\rightarrow\cH^{(\infty)}_{N',L}$, $N<N'$, is induced by the momentum-cutoff renormalization group. Applying the GNS construction also to the renormalized states \eqref{eq:freegroundscalingmom}, we obtain another sequence of Hilbert spaces, $\{\cH^{(M)}_{N,L}\}_{M\in\NN_{0}}$, for each $N\in\NN_{0}$. By von Neumann's uniqueness theorem the representations of $\cW_{N,L}$ on all $\cH^{(M)}_{N,L}$ and $\cH^{(\infty)}_{N,L}$ are unitarily equivalent to the representation on $\cH_{N,L}$, because all the GNS Hilbert spaces are symmetric Fock space over the finite-dimensional vector space $\ltwo(\Gamma_{N},(2r_{N})^{-d}\mu_{\Gamma_{N}})$ with inner products given in terms of either some $\gamma_{\mu_{N+M}}$ or $\gamma_{m}$ which are bounded above and strictly bounded away from zero on $\Gamma_{N}$\footnote{Alternatively, it also easy to check that the Bogoliubov transformations which realize the equivalence are unitarily implemented.}. Therefore, we implicitly identify all the GNS Hilbert spaces with $\cH_{N,L}$ and define the renormalized Hamiltonians $H^{(N)}_{L,M}$ by restriction of $H^{(N+M)}_{L,0}$ to $\cH^{(M)}_{N,L}$ inside $\cH_{N+M,L}$, i.e.:
\begin{align*}
H^{(N)}_{L,M} & = H^{(N+M)}_{L,0|\cH_{N,L}\subset\cH_{N+M,L}}.
\end{align*}
Each $H^{(N)}_{L,M}$ is self-adjoint as it arises as the second quantization of a self-adjoint multiplication operator (with $\gamma_{\mu_{N+M}}$) on $\fh_{N,L}$. Analogously, we define $H^{(N)}_{L,\infty}$ as the self-adjoint second quantization of the self-adjoint multiplication operator by $\gamma_{m}$ on $\fh_{N,L}$. The uniform convergence $\gamma_{\mu_{N+M}}\rightarrow\gamma_{m}$ on $\Gamma_{N}$ implies the convergence $H^{(N)}_{L,M}\rightarrow H^{(N)}_{L,\infty}$ in the norm resolvent sense, cp.~\cite{KatoPerturbationTheoryFor}. Using the identification of $\cH_{N,L}$ with $\cH^{(\infty)}_{N,L}$, we observe that the sequence $\{H^{(N)}_{L,\infty}\}_{N\in\NN_{0}}$ satisfies the compatibility condition:
\begin{align*}
H^{(N')}_{L,\infty}V^{N}_{N'} & = V^{N}_{N'}H^{(N)}_{L,\infty}, & \forall N&<N'.
\end{align*}
Moreover, the dense subspaces $\cD^{(\infty)}_{N,L}\subset\cH^{(\infty)}_{N,L}$, $N\in\NN_{0}$, of vectors with finite particle number are cores for each $H^{(N)}_{L,\infty}$ and satisfy $V^{N}_{N'}\cD^{(\infty)}_{N,L}\subset\cD^{(\infty)}_{N',L}$. Therefore, the inductive-limit operator $H^{(\infty)}_{L,\infty}$ exists and is essentially self-adjoint on the algebraic inductive limit of the subspaces $\cD^{(\infty)}_{N,L}$ which is dense in $\cH^{(\infty)}_{L,\infty}$. It is easy to check that this operators agrees with the free Hamiltonian $H_{L}$ of mass $m$ on $\cH_{L}$. \hfill$\square$
\end{proof}

\subsection{UV vs. IR scaling limits: Momentum transfer to small scales }
\label{sec:irvsuv}
The momentum space setting of Section \ref{sec:mom} is suitable to further discuss distinctive aspects between ultraviolet and infrared scaling limits raised in Section \ref{sec:intro}. To be rather specific, we compare our setting, that primarily deals with the construction of continuum models via infrared scaling limits of lattice models, and the scaling algebra approach by Buchholz and Verch \cite{BuchholzScalingAlgebrasAnd1, BuchholzScalingAlgebrasAnd2}, which aims at the construction of ultraviolet scaling limits of continuum field theories. To this end we modify the momentum-cutoff renormalization group of Definition \ref{def:momcutoff}  through a shift of the support of one-particle vectors $\hat{q},\hat{p}\in\fh_{N,L}$ towards higher momentum (apart from the zero mode $k=0$) via the canonical inclusion of dual lattices $2\Gamma_{N}\subset\Gamma_{N'}$ for $N<N'$.
The real-space version of the resulting renormalization map is
\begin{align*}
R^{N}_{N+1}(q,p)(x')
 & = (2^{\frac{1}{2}} q, 2^{-\frac{1}{2}} p)(2x').
\end{align*}
The latter can be directly compared with the scaling map (up to an additional factor of $2^{\frac{d}{2}}$ due to periodic boundary conditions on the lattice),
\begin{align*}
R^{N}_{N+1}(f)(x') & = 2^{\frac{d+1}{2}}\Re(f)(2x') + i 2^{\frac{d-1}{2}}\Im(f)(2x'),
\end{align*}
used in \cite[Eq. (2.7)]{BuchholzScalingAlgebrasAnd2} for complex one-particle vectors $f\in\cS(\dR^{d})$ in the continuum subject to the symplectic form \eqref{eq:sympcontinf}.  
One finds that the action of the momentum-transfer renormalization group identifies the renormalized state $\omega^{(N)}_{L,M} = \omega^{(N)}_{\mu_{N+M},M}$ at scale $N$ after $M$ steps with the lattice vacuum $\omega^{(N)}_{\mu_{N+M}}$ of lattice mass $\mu_{N+M}$ at scale $N$. Thus, assigning a physical mass $m>0$ to the lattice mass $\mu_{N}$ uniformly among all scales $N$ according to \eqref{eq:freegroundmass}, this entails (at least for $d>1$):
\begin{align}
\label{eq:momentumtransfermasslimit}
\omega^{(N)}_{L,M}(W_{N}(\xi)) & = \omega^{(N)}_{2^{-M}m,0}(W_{N}(\xi)).
\end{align}
In this respect, the relation between lattice mass $\mu_{N}$ and physical mass $m$, the momentum-transfer renormalization group universally\footnote{This will hold more generally as long as $\lim_{M\rightarrow\infty}2^{-2M}\vep_{N+M}^{-2}(\mu_{N+M}^{2}-2d)=0$ for all $N\in\NN_{0}$.} leads to the family of massless lattice vacua $\{\omega^{(N)}_{0,0}\}_{N\in\NN_{0}}$ as the projective family defining the scaling limit, i.e. \eqref{eq:freegroundmass} in the limit $m\rightarrow0+$. At first sight, this appears to be similar to the observation made in \cite{BuchholzScalingAlgebrasAnd2} for ultraviolet scaling limits of massive continuum free fields (in $d>1$), and it is tempting to assume that the scaling limit of \eqref{eq:momentumtransfermasslimit} equals the vacuum of the massless continuum free field in finite volume $\omega_{L}$. But, this is not as evident as for the wavelet renormalization group because the relation between the scaling-limit algebra $\cW_{\infty, L}$ and the continuum algebra $\cW_{L}$ remains rather obscure. \\[0.1cm]
In summary, we conclude that while there are technical similarities between the scaling algebra approach of \cite{BuchholzScalingAlgebrasAnd1, BuchholzScalingAlgebrasAnd2}, the latter should be distinguished from the approach detailed here. Moreover, although it is possible to use the general method of Section \ref{sec:oaren} to implement ultraviolet scaling limits as explicated above, it is conceptually unlikely to be meaningful in the context of scaling limits of lattice models as in Section \ref{sec:continuumlimit}, where the goal is to obtain an infrared scaling limit.

\subsection{Connections with multi-scale entanglement renormalization (MERA)}
\label{sec:MERA}
In view of Example \ref{ex:twistedpartialtrace}, we note that there is a natural decomposition of the wavelet scaling map \eqref{eq:waveletscale} into a tensor-product embedding of the form $W_{N}(\xi)\mapsto W_{N}(\xi)\otimes\mathds{1}_{\cH_{N+1|N}}$ and the adjoint action by a unitary $U_{N+1}\in\cU(\cH_{N+1,L})$:
\begin{align}
\label{eq:tensorconj}
\alpha^{N}_{N+1}(W_{N}(\xi)) & = \Ad_{U_{N+1}}(W_{N}(\xi)\otimes\mathds{1}_{\cH_{N+1|N}}).
\end{align}
where $\cH_{N+1|N} = \fF_{+}(\ltwo(\Lambda_{N+1}\setminus\Lambda_{N}))$ is the natural tensor-product complement of $\cH_{N,L}$ inside $\cH_{N+1,L}$. This follows from the decomposition of $R^{N}_{N+1}$ into a symplectic automorphism of $\fh_{N+1,L}$ and the embedding $\ltwo(\Lambda_{N})\rightarrow\ltwo(\Lambda_{N+1})$ defined by the inclusion $\Lambda_{N}\subset\Lambda_{N+1}$ as well as (multiplicative) second quantization:
\begin{align*}
\alpha^{N}_{N+1} & = \fF_{+}(R^{N}_{N+1}) = \fF_{+}(S_{N+1})\circ\fF_{+}(\ltwo(\Lambda_{N})\rightarrow\ltwo(\Lambda_{N+1})) = \Ad_{U_{N+1}}\circ(\cW_{N,L}\rightarrow\cW_{N+1,L}).
\end{align*}
The symplectic automorphism $S_{N+1}$ is determined by the scaling equation \eqref{eq:scalingeq},
\begin{align}
\label{eq:symplecticrot}
S_{N+1}(x',y') & = \sum_{n\in\ZZ^{d}}h_{n}\delta_{\epsilon_{N+1}n,x'-y'},
\end{align}
as a kernel relative to the standard basis of $\ltwo(\Lambda_{N+1})$.\\[0.1cm]
From this we can make three basic observations, see also \cite{BrothierConstructionsOfConformal, BrothierAnOperatorAlgebraic}:
\begin{itemize}
 \item The $C^{*}$-inductive limit $\cW_{\infty, L}$ is $^*$-isomorphic to an infinite (minimal) tensor-product of the basic algebra $\cW(L^{2}(\RR))$. This follows from an inductive construction using the unitary conjugacy of the renormalization group $\{\alpha^{N}_{N'}\}_{N'>N}$ with the tensor-product embedding, cp. \eqref{eq:tensorconj}.
 \item A projectively-consistent family of states $\{\omega^{(N)}_{L,\infty}\}_{N\in\NN_{0}}$ or alternatively a state $\omega^{(\infty)}_{L,\infty}$ as in the definition of the scaling limit \eqref{eq:projstaterg} on $\cW_{\infty, L}$ is analogous to a MERA (see e.g.\!
 \cite{VidalAClassOf, EvenblyTNRMERA}): The GNS vectors $\{\Omega_{N}\}_{N\in\NN_{0}}$ given by the projective family or the corresponding restrictions of $\omega^{(\infty)}_{L,\infty}$ satisfy, see Section \ref{sec:indlim}:
\begin{align}
\label{eq:gnsmera}
V^{N}_{N'}\Omega_{N} & = \Omega_{N'}, & N&<N',
\end{align}
where $V^{N}_{N'}$ is the GNS (partial) isometry induced by $\alpha^{N}_{N'}$. The connection of the latter with a MERA follows from the observation that a partial isometry compatible with \eqref{eq:tensorconj} is given by:
\begin{align}
\label{eq:compiso}
\psi_{N} & \mapsto U_{N+1}(\psi_{N}\otimes\Omega^{(0)}_{N+1|N}),
\end{align}
leading to a representation of $\cW_{\infty, L}$. Here, $\Omega^{(0)}_{N+1|N}$ is the Fock vacuum of the space
$\fF_{+}(\ltwo(\Lambda_{N+1}\setminus\Lambda_{N}))$, and the MERA is given by the partial isometry $I^{N}_{N+1}:\fF_{+}(\ltwo(\Lambda_{N})\rightarrow\fF_{+}(\ltwo(\Lambda_{N+1}))$, $I^{N}_{N+1}(\psi_{N})=\psi_{N}\otimes\Omega^{(0)}_{N+1|N} $, and the unitary disentangler $U_{N+1}$. In this sense \eqref{eq:gnsmera} generalizes \eqref{eq:compiso} subsuming $U_{N+1}$ and $I^{N}_{N+1}$ by the partial isometry $V^{N}_{N'}$ coming from the GNS construction.
 \item Each $\alpha^{N}_{N+1}$ extends to a normal, unital $^*$-morphism from $B(\cH_{N,L})$ to $B(\cH_{N+1,L})$, which allows us to consider the inductive limit $\varinjlim_{N\in\NN_{0}}B(\cH_{N,L}) = B_{\infty,L}$. Thus, as long as we are interested in properties pertaining to the scaling limit $\omega^{(\infty)}_{L,\infty}$ and the weak closure of the associated GNS representation of a family of regular initial states $\omega^{(N)}_{L,0}$, we can substitute $B_{\infty,L}$ for $\cW_{\infty,L}$ in the analysis.
\end{itemize}

\section{Outlook}
\label{sec:out}
To conclude the article, we discuss some future directions.

\subsection{Other continuum models and locality from lattices}
\label{sec:other}
First of all, we considered only massive free fields in this paper.
One difference that appears with massless fields is the zero mode:
the inverse of the dispersion relation is not defined at $k=0$, and this must be avoided somehow (possibly by excluding the zero-mode before taking the infinite volume limit or by enforcing anti-periodic (twisted) boundary conditions).
In addition, massless free fields with different helicity might require different symplectic structures \cite{LMPR19}.

In view of the classical construction of interacting scalar fields in 1+1-dimensional spacetimes \cite{GlimmQuantumFieldTheory}, we point out that the general concept of limit dynamics \eqref{eq:limdyn} together with the momentum-cutoff renormalization group (see Section \ref{sec:mom}) can be used to recover the existence of interacting dynamics on the continuum time-zero net of the scalar field with respect to free vacuum. Clearly, this observation should be generalized in two aspects: First, it should be possible to recover the interacting dynamics using the wavelet renormalization group as well, which would be advantageous from the point of view of locality. Second, it would be conceptually preferred (and necessary for $d>1$) to work with a scaling limit of the interacting lattice ground states instead of the free lattice ground states as the use of the latter is only possible because of the local Fock property in two-dimensional spacetime \cite{GlimmJaffe3}.

If we consider more general models than lattice free fields, we hope to obtain (interacting) quantum field theories by the operator-algebraic renormalization procedure explained in Section \ref{sec:oaren}. A general justification of this expectation is due to the result on finite propagation speed of the scaling limit dynamics via Lieb-Robinson bounds for the approximate dynamics on the lattice, cp.~\cite{OsborneContinuumLimitsOf, BrothierAnOperatorAlgebraic}. The main reason why a finite propagation speed can be deduced from these bounds in the scaling limit of the lattice free field is due to the correct scaling of the Lieb-Robinson velocity (with respect to the lattice spacing $\vep$) and a subexponential bound on the field-strength renormalization, i.e. the rescaling of the one-particle vectors $q,p$ with respect to the lattice spacing $\vep$. Therefore, if said scaling is preserved in the presence of interactions on the lattice, a finite propagation speed in the scaling limit will be attained. Explicit Lieb-Robinson bounds for discretized scalar fields possibly with interaction are available in \cite{NachtergaeleLRBoundsHarmonic}, see also \cite{NachtergaeleLiebRobinsonBounds}.\\[0.1cm]
Thus, as also suggested in \cite{BrothierAnOperatorAlgebraic}, this might provide a way to approximate local spacetime algebras at the lattice level by adapting the definition of these algebras in the Hamiltonian setting \eqref{eq:causalnet} to the discretized setting, cp.~\cite[Chapter IV, Definition 2.8]{GlimmBosonQuantumField}.:
\begin{align*}
\cA_{N,L}(\cO) & = \left(\bigcup_{t\in\RR}\eta^{(N)}_{L|t}\left(\cW_{N,L}(\cO(t)\cap\Lambda_{N})\right)\right)'', & N&\in\NN_{0}.
\end{align*}
$\cO\subset\RR\times\mathds{T}^{d}_{L}$ should be a suitable open bounded subset with time slices $\cO(t)$, $t\in\RR$. $\cW_{N,L}(\cO(t)\cap\Lambda_{N})\subset\cW_{N,L}$ is the Weyl subalgebra of the sublattice $\cO(t)\cap\Lambda_{N}$ (extension by $\mathds{1}$ outside). $\eta^{(N)}_{L}:\RR\curvearrowright\cW_{N,L}$ is the approximate dynamics at level $N\in\mathds{N}_{0}$. 

As a more general remark, we notice that for general interacting models there are two sources of non-uniqueness of the scaling limit states $\omega^{(\infty)}_\infty$ in~\eqref{eq:projstaterg}, both of them of physical relevance. The first one comes from the different possible choices of the sequence of initial states $\{\omega^{(N)}_0\}_{N \in \NN_{0}}$, which reflects the freedom in the choice of renormalization conditions for the parameters of the model (e.g., different choices of the sequence $\{\mu_N\}_{N \in \NN_{0}}$ in~\eqref{eq:Latmassscaling} lead to different physical masses). The second one is due to the possible multiplicity of weak* limit points of the sequences of renormalized states $\{\omega^{(N)}_M\}_{M \in \NN_{0}}$. As already pointed out after Prop.~\ref{prop:weakstarconv}, one gets a projective system of renormalized states $\{\omega^{(N)}_\infty\}_{N \in \NN_{0}}$ in this case too, which corresponds to the existence of inequivalent vacua of the continuum theory.

\subsection{Extension to fermions}
\label{sec:fermi}
Fermionic systems can by treated by the renormalization group scheme of Section \ref{sec:oaren} in a similar manner as bosonic systems, cp.~\cite{EvenblyEntanglementRenormalizationAnd, HaegemanRigorousFreeFermion, WitteveenQuantumCircuitApproximations} for a related construction in the context of MERA. Since a detailed discussion of this extension will be presented in a forthcoming article by one of the authors on scaling limit of lattice fermions and the approximation of conformal symmetry \cite{OsborneScalingLimitFermions}, we only state some of the essential structures: As in Section \ref{sec:latscalar} the kinematical data of the quantum lattice models is derived from one-particle spaces by second quantization,
\begin{align*}
\fh_{N,L} & = \ltwo(\Lambda_{N})\otimes\CC^{s}, & \fA_{N} & = \fA_{\CAR}(\fh_{N,L}), & \cH_{N,L} & = \fF_{-}(\fh_{N,L}),
\end{align*}
where $s$ denotes the number of spinor components, $\fA_{\CAR}(\fh_{N,L})$ is the $C^{*}$-algebra of the canonical anti-commutation relations, and $\fF_{-}(\fh_{N,L})$ the anti-symmetric Fock space. The wavelet renormalization group is then adapted to the one-particle level by (cf.~Section \ref{sec:waveletscaling}):
\begin{align*}
R^{N}_{N+1} : \fh_{N,L} & \longrightarrow \fh_{N+1,L}, & R^{N}_{N+1}(\xi) & = \sum_{x\in\Lambda_{N}}\xi(x)\sum_{n\in\ZZ^{d}}h_{n}\delta^{(N+1)}_{x+\vep_{N}n}, & \xi & \in\fh_{N,L}.
\end{align*}
In analogy with the free scalar quantum field, the simplest situations that allow for an implementation of the scaling limit procedure are those with initial families of states, $\omega^{(N)}_{L,0}$, of quasi-free type \cite{EvansQuantumSymmetriesOn} (covering the case of ground states for free lattice fermions):
\begin{align}
\label{eq:twopointF}
\omega^{(N)}_{L,0}(\psi(\xi)\psi^{\dagger}(\eta)) & = \langle\xi,T^{(N)}_{L,0}\eta\rangle_{N,L}, & \xi,\eta&\in\fh_{N,L},
\end{align}
for operators, $T^{(N)}_{L,0}:\fh\rightarrow\fh$, with $0\leq T^{(N)}_{L,0}\leq 1$.

\subsection{Entropy and Type I approximation of local algebras}
\label{sec:entropy}
In a continuum quantum field theory, the local algebras are type III factors.
Although the the entanglement entropy in quantum field theory is a focus of attention in recent years \cite{NishiokaEntanglementEntropy},
its direct generalization is infinite for any state on a type $\mathrm{III}_1$ factor \cite{OhyaPetzQuantumEntropy}.
As an alternative, one can take a sequence of type I factors inside a local algebra and look at the divergence rate of the entanglement entropy \cite{LongovonNeumannEntropy}. The lattice algebras in the present construction will give one such sequence.
In addition, as the Bisognano-Wichmann property (the Lorentz boosts are given by the modular group of the wedge algebra with respect
to the vacuum \cite{BisognanoWichmann76}) follows in the continuum from general structural assumptions (cf.\! \cite{Morinelli18, DybalskiMorinelli20}),
it would be interesting to understand what the modular groups of the lattice algebras are, and how the Bisognano-Wichmann property in the continuum
is restored \cite{LongoSimpleProof}.

In this regard, the construction of diffeomorphism symmetry in two-dimensional conformal field theory (2d CFT)
through representations of Thompson's group $T$ \cite{JonesANoGo} would require a substantial modification, because the diffeomorphism symmetry
in any 2d CFT does not extend to the Thompson's group $T$ while preserving the covariance
(because such an action cannot be locally normal when the group element has discontinous derivative \cite[Theorem 3.6]{DelVecchioSolitonsAndNonsmooth}).


\end{document}